\documentclass[11pt,letterpaper]{article}

\usepackage[bookmarks,colorlinks,breaklinks]{hyperref}
\hypersetup{urlcolor=blue, colorlinks=true, citecolor=green!50!black, linkcolor=blue}
\usepackage[letterpaper, left=1in, right=1in, top=0.9in, bottom=0.9in]{geometry}
\usepackage[utf8]{inputenc}
\usepackage[american]{babel}
\usepackage[normalem]{ulem}
\usepackage{amsmath, amssymb, cases, amsthm}
\usepackage{thmtools}
\usepackage[capitalize]{cleveref}
\usepackage{enumitem}
\usepackage{mdframed}
\usepackage{bbm}
\usepackage{bm}
\usepackage{microtype}
\usepackage{xcolor}
\usepackage{makecell}
\usepackage{mathtools}
\usepackage{float}

\declaretheorem[numberwithin=section,refname={Theorem,Theorems},Refname={Theorem,Theorems}]{theorem}

\declaretheorem[numberlike=theorem]{lemma}
\declaretheorem[numberlike=theorem]{proposition}
\declaretheorem[numberlike=theorem]{corollary}
\declaretheorem[numberlike=theorem]{definition}
\declaretheorem[numberlike=theorem]{claim}
\declaretheorem[numberlike=theorem]{algorithm}
\declaretheorem[numberlike=theorem,style=remark]{remark}

\declaretheorem[numberlike=theorem, refname={Observation,Observations},Refname={Observation,Observations},name={Observation}]{observation}

\declaretheorem[numberlike=theorem,refname={Technique,Techniques},Refname={Technique,Techniques},name={Technique}]{technique}

\usepackage[ruled]{algorithm}
\usepackage[noend]{algpseudocode}

\DeclarePairedDelimiter{\ceil}{\lceil}{\rceil}
\newcommand{\eps}{\varepsilon}

\newcommand{\mindeg}{\mathrm{mindeg}}
\newcommand{\nbh}{\mathrm{nbh}}
\newcommand{\spf}{\mathrm{spf}}
\newcommand{\direct}{MDCP }
\newcommand{\best}{\mathsf{best}}
\newcommand{\polylog}{\mathrm{polylog}}
\newcommand{\outdeg}{\mathrm{outdeg}}
\newcommand{\ones}{\mathrm{ones}}
\newcommand{\euler}{\mathrm{e}}
\newcommand{\var}{\mathrm{Var}}
\newcommand{\set}[2][ ]{\{#2 \ifthenelse{\equal{#1}{ }}{ }{~|~#1}\}}

\newcommand{\bbE}{{\mathbb E}}

\newcommand{\Pcal}{\mathcal{P}}
\newcommand{\Acal}{\mathcal{A}}
\newcommand{\Scal}{\mathcal{S}}

\newcommand{\pflag}{\mathrm{pflag}}
\newcommand{\R}{\mathbb{R}}

\newcommand{\one}{\mathbf{1}}
\newcommand{\inact}{\mathsf{Inactive}}
\newcommand{\act}{\mathsf{Active}}
\newcommand{\cut}{\mathrm{cut}}
\newcommand{\cutd}{\overrightarrow{\mathrm{cut}}}
\newcommand{\Ed}{\overrightarrow{E}}
\newcommand{\mincut}{\lambda}
\newcommand{\tOh}{\widetilde{O}}
\newcommand{\ith}{i^{\scriptsize \mbox{{\rm th}}}}
\newcommand{\jth}{j^{\scriptsize \mbox{{\rm th}}}}

\title{Cut query algorithms with star contraction}

\author{Simon Apers\thanks{Université de Paris, CNRS, IRIF, F-75013, Paris, France, \texttt{smgapers@gmail.com}}\and Yuval Efron\thanks{Columbia University, USA, \texttt{ye2210@columbia.edu}. Work done while affiliated with University of Toronto.} \and Pawe\l{} Gawrychowski\thanks{University of Wroc\l{}aw, Poland, \texttt{gawry@cs.uni.wroc.pl}} \and Troy Lee\thanks{Centre for Quantum Software and Information, University of Technology Sydney, \texttt{troyjlee@gmail.com}} \and Sagnik Mukhopadhyay\thanks{University of Sheffield, UK, \texttt{s.mukhopadhyay@sheffield.ac.uk}. Work done while affiliated to University of Copenhagen and KTH.} \and Danupon Nanongkai\thanks{University of Copenhagen and KTH, \texttt{danupon@gmail.com}}}

\date{}

\begin{document}

\begin{titlepage}
	\maketitle \pagenumbering{roman}
	
\begin{abstract}
We study the complexity of determining the edge connectivity of a simple graph with cut queries.
We show that {\bf (i)} there is a bounded-error randomized  algorithm that computes edge connectivity with $O(n)$ cut queries, and {\bf (ii)} there is a bounded-error quantum algorithm that computes edge connectivity with $\tilde O(\sqrt{n})$ cut queries.
To prove these results we introduce a new technique, called \emph{star contraction}, to randomly contract edges of a graph while preserving non-trivial minimum cuts.  In star contraction vertices randomly contract an edge incident on a small set of randomly chosen ``center'' vertices. 
In contrast to the related 2-out contraction technique of Ghaffari, Nowicki, and Thorup [SODA'20], star contraction only contracts vertex-disjoint star subgraphs, which allows it to be efficiently implemented via cut queries.

The $O(n)$ bound from item~(i) was not known even for the simpler problem of connectivity, and it improves the $O(n \log^3 n)$ upper bound by Rubinstein, Schramm, and Weinberg [ITCS'18].
The bound is tight under the reasonable conjecture that the randomized communication complexity of connectivity is $\Omega(n \log n)$, an open question since the seminal work of Babai, Frankl, and Simon [FOCS'86]. 
The bound also excludes using edge connectivity on simple graphs to prove a superlinear randomized query lower bound for minimizing a symmetric submodular function.
The quantum algorithm from item~(ii) gives a nearly-quadratic separation with the randomized complexity, and addresses an open question of Lee, Santha, and Zhang [SODA'21].
The algorithm can alternatively be viewed as computing the edge connectivity of a simple graph with $\tilde O(\sqrt{n})$ matrix-vector multiplication queries to its adjacency matrix.

Finally, we demonstrate the use of star contraction outside of the cut query setting by designing a one-pass semi-streaming algorithm for computing edge connectivity in the complete vertex arrival setting. 
This contrasts with the edge arrival setting where two passes are required.
\end{abstract}

 	\setcounter{tocdepth}{2}
 	\newpage
 	\tableofcontents
	\newpage
\end{titlepage}

\newpage
\pagenumbering{arabic}

\section{Introduction and contribution}
\label{sec:intro} 
The minimization of a submodular function is a classic problem in combinatorial optimization.  Over a universe $V$, a submodular 
function $f: 2^V \rightarrow \R$ is a function that satisfies $f(S) + f(T) \ge f(S \cap T) + f(S \cup T)$ for all subsets $S,T \subseteq V$.  
The submodular function minimization (SFM) problem is the task of computing $\min_{S \subseteq V} f(S)$.  The SFM problem generalizes several 
well known combinatorial optimization problems such as computing the minimum weight of an $st$-cut in a directed graph and the matroid intersection 
problem.  The SFM problem comes in another flavor when the submodular function $f$ is \emph{symmetric}, i.e.\ also satisfies $f(S) = f(V \setminus S)$ 
for all $S$.  In this case, $\emptyset$ and $V$ are trivial minimizers, so the interesting problem is to compute $\min_{\emptyset \subsetneq S \subsetneq V} f(S)$, 
the non-trivial minimum.  The global minimum cut problem on an undirected graph is an instance of non-trivial symmetric submodular function minimization, which we denote by sym-SFM.

The size of the truth table of a submodular function is exponential in the size of $V$, so (sym-)SFM is typically 
studied in the setting where we have access to an evaluation oracle for $f$, that is, we can query any $S \subseteq V$ and receive the 
answer $f(S)$.  When $|V| = n$, Gr\"{o}tschel, Lov\'{a}sz, and Schrijver \cite{GLS1988} showed that the ellipsoid method can be used to solve SFM with $\tOh(n^5)$ 
oracle queries and overall running time $\tOh(n^7)$ \cite[Theorem 2.8]{McCormick}.  Since then a long line of work has developed faster and simpler (combinatorial)
algorithms for SFM \cite{Schrijver00, IwataFF01,Orlin09, LSW15, DadushVZ21, Jiang21}.  The work of \cite{Jiang21} shows that SFM can be solved by a 
deterministic algorithm making $O(n^2 \log n)$ queries to an evaluation oracle.  By the isolating cut lemma of \cite{LiP20-isocut} this immediately also gives an $\tOh(n^2)$ query randomized 
algorithm for sym-SFM \cite{ChekuriQ21a-isocut, MukhopadhyayN21-isocut}.  In the deterministic case, the best upper bound on the number of queries to solve sym-SFM remains the $O(n^3)$ algorithm of Queyranne \cite{Queyranne98}.

While sym-SFM is a much more general problem, it has a close relationship with one of its simplest instantiations: the global minimum cut problem.  In this problem we are given a weighted and undirected graph $G = (V,E,w)$ 
and the task is to find the minimum weight of set of edges whose removal disconnects $G$.  For a subset $S \subseteq V$ let $\cut_G(S)$ be the set of edges of $G$ with exactly one endpoint in $S$.  The cut function 
$f: 2^V \rightarrow \R$, where $f(S) = w(\cut_G(S))$ is the total weight of edges in $\cut_G(S)$, is a symmetric submodular function.  Evaluation queries in this case are called \emph{cut queries} and the goal is to compute 
$\mincut(G):=\min_{\emptyset \subsetneq S\subsetneq V} w(\cut_G(S))$ with as few cut queries as possible.

Both the best known deterministic and randomized algorithms for sym-SFM use ideas that originated in the study of minimum cuts: Queyranne's algorithm is based on the Stoer-Wagner minimum cut algorithm \cite{StoerW97}, and 
the best randomized algorithm makes use of the isolating cut lemma originally developed in the context of a deterministic minimum cut algorithm \cite{LiP20-isocut}.  
On the lower bound side, the best known bounds on the query complexity of sym-SFM are $\Omega(n)$ in the deterministic case \cite{HajnalMT88, HarveyThesis} and $\Omega(n/\log n)$ in the randomized case \cite{BabaiFS86} (see \cref{table:bounds}).   Both 
of these bounds can be shown for the cut query complexity of determining the weight of a minimum cut in a simple graph.\footnote{We say that a graph is {\em simple} if it is undirected and unweighted and contains at most one edge between any pair of vertices.}
The weight of a minimum cut in simple graph $G$ is known as the \emph{edge connectivity} of $G$, and is the minimum number of edges whose removal disconnects the graph.
The aforementioned lower bounds even hold for the more special case of determining if the edge connectivity is positive, i.e.\ if the graph is connected or not.  

Recent work has given randomized algorithms that can compute $\mincut(G)$ with $O(n \log^3 n)$ cut queries in the case of simple graphs \cite{RubinsteinSW18} and $n \log^{O(1)}(n)$ cut queries in the 
case of weighted graphs \cite{MukhopadhyayN20}.  For the deterministic case, however, the best upper bound remains $O(n^2/\log n)$ \cite{GK00} and proceeds by learning the entire graph.  Researchers continue to study the minimum cut problem 
as a candidate to show superlinear query lower bounds on sym-SFM.  Graur, Pollner, Ramaswamy and Weinberg \cite{GraurPRW20} introduced a linear-algebraic lower bound technique known as the 
\emph{cut dimension} to show a deterministic cut query lower bound of $3n/2-2$ for minimum cut on weighted graphs.  Lee, Li, Santha, and Zhang \cite{LeeLSZ21} show that the cut dimension cannot show lower bounds larger
than $2n-3$, but use a generalization of the cut dimension to show a cut query lower bound of $2n-2$, the current best lower bound known 
on sym-SFM in general.  

Despite this work, showing a superlinear lower bound on the query complexity of sym-SFM remains elusive.  In this paper, we show that for randomized algorithms and the special case of edge connectivity 
there is actually a \emph{linear upper bound}. 
\begin{theorem}\label{thm:intro:classical result}
There is a randomized algorithm that makes $O(n)$ cut queries and outputs the edge connectivity of a simple input graph $G$ with probability at least $2/3$.
\end{theorem}
In particular one cannot hope to prove superlinear lower bounds on the randomized query complexity of sym-SFM via the edge connectivity problem.  
It remains open if \cref{thm:intro:classical result} is tight.  The best known lower bound is $\Omega(n \log \log(n)/\log(n))$ which follows from the $\Omega(n\log\log n)$ randomized communication complexity lower bound for edge connectivity 
by Assadi and Dudeja \cite{AssadiD21}.  
An $o(n)$ randomized cut query upper bound on edge connectivity would in particular imply a randomized communication complexity protocol for determining if a graph is connected with $o(n \log n)$ bits,\footnote{With shared randomness the parties can simulate a randomized cut query algorithm with an $O(\log n)$ multiplicative overhead: whenever the algorithm makes a cut query, the parties communicate the number of cut edges in their part of the graph with $O(\log n)$ bits. By Newman's theorem, this protocol can be simulated without shared randomness (and with only an additive $O(\log n)$ overhead).} resolving one of the longest standing open problems in 
communication complexity.  Graph connectivity was a focus of many early works on communication complexity \cite{HajnalMT88, BabaiFS86,RazS95}, and while a deterministic lower bound of $\Omega(n \log n)$ was 
established early on \cite{HajnalMT88}, to this day the randomized communication complexity is only known to be between $\Omega(n)$ and $O(n \log n)$.  

\cref{thm:intro:classical result} even improves the previous best cut query upper bound for deciding if a graph is connected.  Harvey \cite[Theorem 5.10]{HarveyThesis} gave a deterministic $O(n\log n)$ cut query upper 
bound for connectivity, and we are not aware of any better upper bound in the randomized case.  For the case of connectivity we can give a linear upper bound even for zero-error algorithms.
\begin{restatable}{theorem}{spanning}
\label{thm:spanning}
Let $G = (V,E)$ be a simple $n$-vertex graph.
There is a zero-error randomized algorithm that makes $O(n)$ cut queries in expectation and outputs a spanning forest of $G$. 
\end{restatable}
The best lower bound we are aware of in this case is $\Omega(n \log \log(n)/\log(n))$ which follows from the non-deterministic communication complexity lower bound for connectivity of 
$\Omega(n \log \log(n))$ by Raz and Spieker \cite{RazS95}.

\begin{table}
\centering
\renewcommand{\arraystretch}{1.2}
\begin{tabular}{| l | c | c | c | c |}
\hline
& \multicolumn{2}{|c|}{Connectivity} & \multicolumn{2}{|c|}{Edge Connectivity} \\
\hline
& Lower & Upper & Lower & Upper \\
\hline
&&&&\\[-1.3em]
Deterministic & \makecell{$\Omega(n)$ \\ \cite{HajnalMT88}} & \makecell{$O(n \log n)$ \\ \cite{HarveyThesis}}& \makecell{$\Omega(n)$ \\ \cite{HajnalMT88}} & \makecell{$O\left(\frac{n^2}{\log n}\right)$ \\ \cite{GK00}} \\
\hline
&&&&\\[-1.3em]
Zero-error & \makecell{$\Omega\left(\frac{n \log \log(n)}{\log(n)}\right)$ \\ \cite{RazS95}} & \makecell{$O(n)$ \\ (\cref{thm:spanning})} & \makecell{$\Omega(n)$ \\ \cite{LeeS21-private}} & \makecell{$O\left(\frac{n^2}{\log n}\right)$ \\ \cite{GK00}} \\
\hline
&&&&\\[-1.3em]
Bounded-error & \makecell{$\Omega\left(\frac{n}{\log n}\right)$ \\ \cite{BabaiFS86}} & \makecell{$O(n)$ \\ (\cref{thm:spanning})} & \makecell{$\Omega\left(\frac{n \log \log(n)}{\log(n)}\right)$ \\ \cite{AssadiD21}} & \makecell{$O(n)$ \\ (\cref{thm:intro:classical result})}\\
\hline
&&&&\\[-1.3em]
Quantum & $\Omega(1)$ & \makecell{$O(\log^5(n))$ \\ \cite{AuzaL21}} & $\Omega(1)$ & \makecell{$\tOh(\sqrt{n})$ \\ (\cref{thm:intro:quantum result})} \\ 
\hline
\end{tabular}
\caption{The cut query complexity of connectivity and edge connectivity on simple graphs in various models. The upper bounds on edge connectivity in the deterministic and zero-error models follow from using \cite[Section 4.1]{GK00} to learn the full graph with $O(n^2/\log(n))$ cut queries.  The bound in \cite{GK00} is stated for additive queries, but the same argument holds for cut queries: with $O(n/\log(n))$ cut queries we can learn the neighborhood of a vertex by \cref{lem:learn}.}
\label{table:bounds}
\end{table}

A key to both \cref{thm:intro:classical result} and \cref{thm:spanning} is to think in terms of matrix-vector multiplication queries.  If $A$ is the adjacency matrix of an $n$-vertex simple graph $G$, 
in a matrix-vector multiplication query we can query any vector $x \in \{0,1\}^n$ and receive the answer $Ax$.  If $G$ has maximum degree $d$, and so $A$ has at most $d$ ones in every row, we 
can learn the entire graph $G$ with only $O(d \log n)$ matrix-vector multiplication queries---this is one of the key ideas behind compressed sensing.  As a single
matrix-vector multiplication query can be simulated with $O(n)$ cut queries, this shows that we can learn $G$ with $O(nd \log n)$ cut queries.  Grebinski and Kucherov \cite{GK00} show the surprising fact 
that if $G$ is bipartite with maximum degree $d$, and the left and right hand sides are roughly the same size, then one can actually learn $G$ with only $O(nd)$ cut queries.  This savings of a $\log n$ factor over the trivial simulation is key to our improved algorithms.

We use this idea to design a primitive called Recover-$k$-From-All.  Given two disjoint subsets $S,T$ of vertices, with the promise that all vertices in $S$ have 
at least $k$ neighbors in $T$, Recover-$k$-From-All makes $O(kn)$ cut queries and learns at least $k$ neighbors in $T$ of 
every vertex in $S$.  This routine is the heart of our algorithm for \cref{thm:spanning}, which uses it to implement Bor\r{u}vka's spanning forest algorithm. 

It is less obvious how such a primitive is useful to compute edge connectivity as it only gives us local snapshots of sparse bipartite induced subgraphs.  To this end we develop a new 
technique for edge connectivity called \emph{star contraction}.  Star contraction is inspired by the randomized 2-out contraction algorithm of Ghaffari, Nowicki and Thorup \cite{GNT20}.  
In that algorithm, each vertex independently and uniformly at random selects two incident edges.  Ghaffari et al.\ show that when the selected edges are contracted the resulting 
graph $G'$ has only $O(n/\delta(G))$ vertices with high probability, where $\delta(G)$ is the minimum degree of $G$, and further with constant probability no edge of a non-trivial minimum cut\footnote{We call a cut {\em trivial} if it isolates a single vertex.} is 
contracted.  When these good 
things happen the edge connectivity of $G$ is the minimum of $\delta(G)$ and the edge connectivity of $G'$.  

2-out contraction is not very compatible with our primitive Recover-$k$-From-All because of the combination of requiring independent sampling and choosing an 
edge incident to a vertex uniformly at random.  Instead, in star contraction we first randomly choose a subset $R$ of size $\Theta(n \log(n)/\delta(G))$.  
With high probability every vertex in $V \setminus R$ will have a neighbor in $R$, and in star contraction we only contract edges with an endpoint in $R$.  The fact that 
the edges that we want to contract are incident on a small number of vertices is a key to the savings of star contraction over 2-out contraction in the cut query model.
Further, if for every vertex in $ v \in V \setminus R$ we contract an edge connecting it to $R$ then the contracted graph $G'$ will \emph{automatically} have its size bounded 
by $|R| = \Theta(n \log(n)/\delta(G))$.  While proving the contracted graph has few vertices is the most difficult part of the argument in 2-out contraction, for star contraction 
it is trivial (although the bound we get is larger by a $\log n$ factor).

The tricky part remaining is \emph{how} to choose a neighbor in $R$ for each $v \in V \setminus R$ without having too high a probability of choosing a neighbor on the 
other side of a non-trivial minimum cut.  In our main technical contribution, we show that each $v \in V \setminus R$ can learn just a \emph{constant} number of
neighbors in $R$ without too high a fraction of them being on the opposite side of a non-trivial minimum cut.  Moreover, we can allow correlations between the neighbors 
learned for different vertices which allows us to efficiently learn these constant number of neighbors using Recover-$k$-From-All with constant $k$.

Not surprisingly, the aforementioned matrix-vector multiplication perspective also leads to efficient randomized algorithms for edge connectivity in the matrix-vector multiplication query model.  
While this model has been used previously in the study of sequential graph algorithms \cite{OrecchiaSV12} and (implicitly) in streaming algorithms for graph problems \cite{AGM12}, it began to 
be studied in and of itself relatively recently in the work of \cite{SWYZ21}, and has since seen several several follow-ups \cite{CHL21, AuzaL21}.  
More surprisingly, it turns out that the study of \emph{quantum} algorithms 
with cut query access to a graph is also closely related to the matrix-vector multiplication model.  This is because with $O(\log n)$ cut queries a quantum algorithm can simulate a 
restricted form of a matrix-vector multiplication query, namely it can compute $Ax$ in the entries where $x$ is zero (this is implicit in \cite{LeeSZ21} and made explicit in \cite[Corollary 11]{AuzaL21}).
Lee, Santha and Zhang \cite{LeeSZ21} used this to show that a quantum algorithm making only $O(\log^6 n)$ cut queries can decide if an $n$-vertex simple graph is connected \cite[Theorem 44]{LeeSZ21}, a nearly exponential speedup over the best possible randomized algorithm.
They left it as an open question whether any quantum speedup is possible for the problem of edge connectivity.
This problem is particularly interesting because not much is known about the complexity of (sym)-SFM with respect to quantum algorithms, either in terms 
of upper or lower bounds (some work has been done on approximation algorithms for SFM, see \cite{HamoudiRRS19}).  The best classical algorithms for (sym)-SFM tend to be highly sequential, a feature which is typically hard to speed up quantumly.
Our matrix-vector multiplication perspective leads to a quantum improvement for the cut query complexity of edge connectivity, as shown by the following theorem.
\begin{theorem}\label{thm:intro:quantum result}
There is a quantum algorithm that makes $\tOh(\sqrt{n})$ cut queries and outputs the edge connectivity of the input simple graph $G$ correctly with high probability.  Similarly, there is a randomized 
algorithm making $\tOh(\sqrt{n})$ matrix-vector multiplication queries to the adjacency matrix of $G$ that outputs the edge connectivity of the input simple graph $G$ correctly with high probability.
\end{theorem}
The quantum part of this theorem gives a near-quadratic speedup over the best possible randomized algorithm.  
Moreover, there is a natural bottleneck to giving a $o(\sqrt{n})$ quantum cut query algorithm 
for edge connectivity, which is that even computing the minimum degree of a graph seems to require $\Omega(\sqrt{n})$ quantum cut queries.\footnote{It is intuitive that computing the edge connectivity 
of a graph is more difficult than computing the minimum degree, and we formalize this via a simple reduction in \cref{appensec:reduction}.}  There is a very natural $O(\sqrt{n})$ quantum 
algorithm for computing the minimum degree: the degree of a single vertex can be computed with one cut query, and one can then use quantum minimum finding \cite{DurrH96} on top of this to find the 
minimum degree with $O(\sqrt{n})$ cut queries.  We conjecture that this simple algorithm is optimal, which would imply that the quantum statement of \cref{thm:intro:quantum result} is tight up to 
polylogarithmic factors.

As a final application, we use our new star contraction technique to obtain a one-pass $\tOh(n)$-space algorithm for computing edge connectivity with high probability in the \textit{complete vertex-arrival} streaming model. In this model, the vertices of the graph $G$ arrive in an arbitrary order with all incident edges. This contrasts with the edge-arrival streaming model, where edges of $G$ arrive in arbitrary order, for which a $\tilde \Omega(n^2)$ lower bound was proven on the space complexity of a one-pass algorithm that computes the edge connectivity \cite{Zelke11}. This bound can be modified to also prove an $\Omega(n^2)$ lower bound on the one-pass space complexity of edge connectivity in the more restrictive \textit{explicit vertex-arrival} model, where the vertices of $G$ arrive only with the edges incident on the previously seen vertices, as was considered in e.g.\ \cite{CormodeDK19}.  For completeness, we include a proof sketch of this lower bound in \Cref{appensec:streaming}.
If however the vertices arrive with edges incident on the previously seen vertices in a {\em random} order, then our technique still implies an $\tilde O(n)$-space algorithm.
For comparison, we also discuss why it is not clear how to use the related 2-out contraction technique to achieve these results.

\section{Technical overview} \label{sec:tech overview}

In the following sections we introduce one of the main tools in this work, star contraction, and give a sketch of the classical and quantum cut query algorithms that make use of star contraction.

\subsection{Star contraction}
The main workhorse for proving our results is a new technique for randomly contracting edges of a simple graph while preserving a non-trivial minimum cut with constant probability.
The idea of contracting edges while preserving non-trivial minimum cuts comes from the celebrated result of Kawarabayashi and Thorup (Fulkerson Prize 2021) \cite{KawarabayashiT19}, which gave the first near-linear time deterministic algorithm for computing the edge connectivity of a simple graph.
A critical observation in their work is the following: we can contract edges in a simple graph $G$ to get a graph $G'$ so that (i) $G'$ has $\tOh(n/\delta(G))$ vertices and $\tOh(n)$ edges, where $\delta(G)$ is the minimum degree of $G$, and (ii) all non-trivial minimum cuts in $G$ are preserved (i.e., no edge participating in a non-trivial minimum cut is contracted).
In particular, if $G$ has a non-trivial minimum cut then $\lambda(G') = \lambda(G)$.
Such a contraction is useful for computing edge connectivity since when the edge connectivity is large and there is a non-trivial minimum cut (which is usually harder to handle), the contraction significantly sparsifies and reduces the 
number of vertices of the graph. 
We call this type of contraction a {\em KT contraction}.

The KT contraction technique has been highly influential, and many works have since used and studied it. The algorithm for KT contraction given in \cite{KawarabayashiT19} takes time $O(m\log^{12} n)$ in the sequential setting when 
the graph has $m$ edges. This was later improved by Henzinger, Rao, and Wang \cite{HRW20} to $O(m \log^2 n (\log \log n)^2)$. Using an expander decomposition algorithm \cite{ChuzhoyGLNPS19,SaranurakW19,NanongkaiS16,Wulff-Nilsen16a} as a black box, Saranurak \cite{Saranurak21} showed a slower but simpler $\tOh(m^{1+o(1)})$ time algorithm to compute a KT contraction. All these algorithms are deterministic but rather complicated, making them hard to adapt to other settings. Rubinstein, Schramm, and Weinberg \cite{RubinsteinSW18} provide a randomized algorithm for computing a KT contraction that is efficient 
in the cut-query setting, and leads to their aforementioned $O(n\log^3 n)$ cut query algorithm for edge connectivity. 

Most relevant for our work is the beautiful {\em 2-out contraction} algorithm by Ghaffari, Nowicki, and Thorup \cite{GNT20}. In this algorithm, every vertex independently at random (with replacement) chooses two of its incident edges to contract. Ghaffari et al.\ show that the resulting contracted graph $G'$ has only $O(n/\delta(G))$ vertices with high probability, and moreover if $G$ has a non-trivial minimum cut then $\mincut(G)=\mincut(G')$ with constant probability.
They use this algorithm to get the current fastest randomized algorithm for edge connectivity with runtime\footnote{The stated bound in \cite{GNT20} is $O(\min\{m + n \log^3 n, m \log n\})$, but more recent work on the 
minimum cut problem by \cite{GawrychowskiMW20} improves it to the bound we quote here.} $O(\min\{m + n \log^2 n, m \log n\})$, and they also obtain improved algorithms for edge connectivity in the distributed setting.

Although \cite{GNT20} did not study the cut query model, the 2-out contraction approach gives a simple randomized algorithm for edge connectivity with $O(n \log n)$ cut queries, improving the 
bound from \cite{RubinsteinSW18}.  As this is very related to our approach, we give an outline of the proof here.  First, we can compute $\delta(G)$ with $n$ cut queries by querying $|\cut(\{v\})|$ for every 
vertex $v$.  The next thing to notice is that for any vertex $v$ we can randomly choose a neighbor 
of $v$ with $O(\log n)$ cut queries using a randomized version of binary search.  This is because with $3$ cut queries we can compute $|E(v,S)|$ for any set $S \subseteq V \setminus \{v\}$ (see \cref{prop:3}), and thus can continue searching for a neighbor of 
$v$ in the set $S$ with probability proportional to this number.  Thus with $O(n \log n)$ queries we can perform 2-out contraction and identify the sets of vertices forming the ``supervertices'' of the contracted 
graph $G'$.  By the main theorem of \cite{GNT20}, with high probability $G'$ will have $O(n/\delta(G))$ supervertices.  The remaining task is to compute the edge connectivity of $G'$.  To do this we can make use of a very useful
tool developed by Nagamochi and Ibaraki \cite{NI92} called a \emph{sparse $r$-edge connectivity certificate}.  Let $F$ be the set of edges found by repeating $r$ times: (i) find a spanning forest of $G'$ and (ii) add the edges of this spanning forest to $F$ and 
remove them from $G'$.  Then Nagamochi and Ibaraki show that if $|\cut_G(S)| \le r$ then $\cut_G(S) = \cut_F(S)$.  In particular, the edge connectivity of a sparse $\delta(G)$-edge connectivity certificate of $G'$ will 
equal $\tau = \min\{\delta(G),\mincut(G')\}$.  Contraction cannot decrease edge connectivity, so if $\mincut(G) = \delta(G)$ then $\tau$ will always be the correct answer; if $\mincut(G) < \delta(G)$ then it will be correct 
whenever we do not contract an edge of a non-trivial minimum cut in the 2-out contraction, which happens with constant probability.
We can find a single spanning forest of $G'$ deterministically with $O(n \log(n)/\delta(G))$ cut queries \cite[Theorem 5.10]{HarveyThesis}, thus we can find a sparse $\delta(G)$-edge connectivity certificate of $G'$ with $O(n \log n)$ cut queries overall.

It is not obvious how to independently sample a uniformly random neighbor of every vertex without spending $\Omega(\log n)$ cut queries 
on average per vertex.  Even with the very powerful matrix-vector multiplication queries, where one can learn the entire neighborhood of a vertex with a single query, it is not clear 
what else one can do to avoid spending $\Omega(1)$ queries per vertex on average to implement 2-out contraction.

We introduce a new graph contraction technique called \emph{star contraction} that allows one to take advantage of the power of cut and 
matrix-vector multiplication queries to process vertices in parallel.  For greater flexibility in the applications to different types of queries, we state this as a 
general method that can be instantiated in various ways.

\begin{technique}[Star contraction] 
\label{star_primitive}
Let $G=(V,E)$ be a simple graph and $p \in (0,1]$ be a probability parameter (think of $p \in \widetilde \Theta(1/\delta(G))$).
\begin{enumerate}
\item
Define a set of ``center vertices'' $R$ where every vertex is put into $R$ independently at random with probability $p$.
\item
Define a set of ``star edges'' $X$ by doing the following for every vertex $v \notin R$: pick a neighbor $c \in R$ (if it exists) and put the edge $\{v,c\}$ into $X$. 
The set $X$ is a collection of star subgraphs centered at vertices in $R$.  
\end{enumerate}
Output the graph $G'$ obtained from $G$ by contracting all edges in $X$. 
\end{technique}
Note that in item~2 we do not specify how to pick a neighbor in $R$.  The rule for doing this will vary in our applications.  The nice thing about the star contraction 
framework is that no matter what rule is used here, the number of vertices in $G'$ will always be at most $|R|$ plus the number of vertices in $V \setminus R$ that 
have no neighbor in $R$.  By taking $p = \Theta(\log(n)/\delta(G))$, for example, with high probability all vertices will have a neighbor in $R$ and hence $G'$ will only have $O(n \log(n)/\delta(G))$ many vertices.  This leaves one only with the question of choosing a good rule to instantiate item~2 that does not choose an edge of non-trivial minimum cut
with too high a probability, and that can be efficiently performed in the query model of interest.

The most natural rule to instantiate item~2 of \cref{star_primitive} is to have each vertex in $V \setminus R$ independently and uniformly at random choose a neighbor in $R$.  We refer to this instantiation as 
\emph{uniform star contraction}.  The analysis of this case suffices for our algorithms in the quantum cut query model, the matrix-vector multiplication query model, and the streaming model.
As an example, in the matrix-vector multiplication and quantum cut query settings we can learn all the neighbors of a vertex with $\tOh(1)$ queries.
Hence, we can learn all edges incident on the center vertices $R$ with only $\tOh(|R|)$ queries, which then allows us to implement uniform star contraction.
Now if $\delta(G)$ is large (and we choose $p = \Theta(\log(n)/\delta(G))$ as above) then the query cost $\tOh(|R|)$ will be small, and we will take advantage of this.
This contrasts with the case of 2-out contraction, where in general there does not exist a small set so that all contracted edges are incident on this set (as an example, consider the case where $G$ is the complete graph).
Formally, we show the following theorem about uniform star contraction.

\begin{restatable}{theorem}{simplestar}
\label{thm:intro:simple star contraction}
Let $G=(V,E)$ be an $n$-vertex simple graph with $\mincut(G) < \delta(G)$. Then uniform star contraction with $p=\frac{1200 \ln n}{\delta(G)}$ gives $G'$ where 
\begin{enumerate}
\item $G'$ has at most $2400 n \ln(n)/\delta(G)$ vertices with probability 
at least $1-1/n^4$, and
\item $\mincut(G')=\mincut(G)$ with probability at least $2 \cdot 3^{-13}$.
\end{enumerate}
\end{restatable}

We give an overview of the proof here.
As mentioned, the number of vertices in the contracted graph $G'$ is at most $|R|$ plus 
the number of vertices in $V \setminus R$ that have no neighbor in $R$.  
By a Chernoff bound, with high probability $|R|$ will be at most twice its expectation, which is $\Theta(n \log(n)/\delta(G))$.
Further, the expected number of neighbors of any vertex $v$ among the centers $R$ is $\Omega(\log n)$. Thus, by a Chernoff bound plus a union bound, every 
vertex in $V \setminus R$ will have a neighbor in $R$ with high probability.
This argument nearly trivially bounds the number of vertices in $G'$.
In contrast, bounding the number of vertices in $G'$ is the most complicated part of the proof for the analog of \cref{thm:intro:simple star contraction} for 2-out contraction, although it must be noted the bound obtained there is better by a factor of $\log n$.  
Another nice property of star contraction is that each connected component of $G'$ has diameter 2, a property that is useful for designing algorithms in models of distributed 
computing.  Ghaffari et al.\  show that after 2-out contraction the average diameter of a component is $O(\log \delta(G))$ with high probability \cite[Lemma 5.1]{GNT20}.\footnote{Perhaps more 
comparable to our case, Ghaffari et al.\ also obtain a \emph{worst-case} upper bound of $O(\log n)$ on the diameters of the connected components of a contracted graph $G'$ with $O(n \log(n)/\delta(G))$ vertices obtained 
by only contracting a subset of the edges selected in a 2-out sample \cite[Remark 5.3]{GNT20}.}

For the second item of the theorem it is useful to first review the proof of the analogous statement for $2$-out contraction.  Let $C$ be a non-trivial 
minimum cut of $G$.   Let a \emph{random 1-out sample} of $G$ be 
the set of edges obtained by independently and uniformly at random selecting an edge incident to each vertex.  A $2$-out contraction is exactly the process of independently 
performing two random 1-out samples of $G$ and contracting all the selected edges.  The probability that we contract an edge of $C$ in performing a 2-out contraction is 
exactly the square of the probability that we select an edge of $C$ in a random 1-out sample.

For every vertex $v$ let $d(v)$ be the degree of $v$ and $c(v)$ be the number of $u$ such that $\{u,v\} \in C$.  Let $N(C) = \{v \in V : c(v) > 0\}$ be the 
set of vertices incident to $C$.  When we take a random 1-out sample of $G$, the probability that we do not choose an edge of $C$ is exactly
\begin{equation}
\label{eq:Cprod}
\prod_{v \in N(C)} \left( 1 - \frac{c(v)}{d(v)} \right) \enspace .
\end{equation}
At first it might seem that this probability could be very small.  The key to lower bounding it combines two observations:
\begin{align}
\frac{c(v)}{d(v)} &\le 1/2 \text{ for every } v \in N(C) 
\label{eq:half} \\  
\sum_{v \in N(C)} \frac{c(v)}{d(v)} &\le 2\frac{|C|}{\delta(G)} \le 2
\label{eq:sumis2}.
\end{align}
The first inequality follows from the fact that $C$ is a non-trivial cut, and if \cref{eq:half} did not hold then we could move $v$ to the other side and obtain a smaller cut. 
To derive \cref{eq:sumis2} we use the fact that $d(v) \ge \delta(G)$ and $|C| \le \delta(G)$.

How small can \cref{eq:Cprod} be under the constraints of 
\cref{eq:half} and \cref{eq:sumis2}?  In fact it is always at least $1/16$: some thought shows that the minimum of \cref{eq:Cprod} will be achieved 
at an extreme point of the set of constraints, which is obtained when 4 vertices have 
$c(v)/d(v) = 1/2$, thereby saturating both \cref{eq:half} and \cref{eq:sumis2}.\footnote{See \cref{clm:sample} for a proof.  One can alternatively obtain a looser bound by using the inequality $1-x \ge \exp(-x/(1-x))$ for $0<x<1$, as is done in \cite{GNT20}.}
The lower bound of $1/16$ is tight as can be seen by taking a non-trivial minimum cut of the cycle graph.
This completes the slick argument that indeed 2-out contraction will not contract an edge of a non-trivial minimum cut with constant probability.

Correctness of our algorithms based on star contraction is proven using analogs of \cref{eq:half} and \cref{eq:sumis2} (with slightly worse constants).
We again illustrate the correctness proof for the case of \emph{uniform} star contraction.
Let $d_R(v) = |E(v, R\setminus \{v\})|$ and $c_R(v) = |C \cap E(v, R\setminus \{v\})|$.
As an analog of \cref{eq:sumis2}, we directly show that for all $v \in V$
\begin{align}
\bbE_R \left [ \frac{c_R(v)}{d_R(v)} \;\bigg|\; d_R(v) > 0 \right] = \frac{c(v)}{d(v)} \enspace .\label{eq:expectation}
\end{align}
Therefore by linearity of expectation and \cref{eq:sumis2} we have $\bbE[\sum_{v: d_R(v)>0}  c_R(v)/d_R(v)] \le 2$, and by Markov's inequality the sum will not significantly exceed this quantity with constant probability.
To prove an analog of \cref{eq:half} we again use the fact that, with high 
probability, $d_R(v) = \Omega(\log n)$ for all $v$.
By \cref{eq:expectation} we also know that $\bbE[c_R(v)/d_R(v) \mid d_R(v) > 0] = c(v)/d(v) \le 1/2$.  
Thus for $c_R(v)/d_R(v)$ to exceed $2/3$ we 
must have $c_R(v) = \Omega(\log n)$ and $c_R(v)$ exceeding its expected value by a constant factor bigger than $1$.  We then again use a Chernoff 
bound to show that for each $v$ individually this does not happen with high probability, and finally apply a union bound over all~$v$.

\subsection{Matrix-vector multiplication and quantum cut query algorithm}
\label{sec:tech_mvp}
As a direct application of our uniform star contraction procedure we obtain an algorithm for computing the edge connectivity of a simple graph with $\tOh(\sqrt{n})$ quantum cut queries (i.e., \Cref{thm:intro:quantum result}) or 
matrix-vector multiplication queries to the adjacency matrix.
We sketch the algorithm here and postpone details to \cref{sec:quant_algo}.
The algorithm uses the following primitives (which we can run either on the original graph, a vertex-induced subgraph, or a vertex-induced subgraph with an explicit set of edges removed):
\begin{enumerate}
\item[P1.]
\label{P:nbhs}
{\em Find all neighbors of a vertex.}
This can be done with 1 matrix-vector multiplication query to the adjacency matrix (for a vertex $v$, query $A \chi_v$ with $\chi_v$ the standard basis vector corresponding to vertex $v$) or $O(\log n)$ quantum cut queries (this is implicitly shown in \cite{LeeSZ21} and made explicit in \cite[Corollary 11]{AuzaL21}).
\item[P2.]
{\em Construct a spanning forest.}
This can be done with $\polylog(n)$ matrix-vector multiplication queries \cite{AuzaL21} or $\polylog(n)$ quantum cut queries \cite{LeeSZ21}.
\item[P3.]
{\em Compute the minimum degree.} This takes 1 matrix-vector multiplication query (query for the matrix-vector product $A \one$, with $\one$ the all-ones vector) or $O(\sqrt{n})$ quantum cut queries (run quantum minimum finding on the vertex degrees).
\item[P4.]
{\em Compute a cut query.} This can clearly be done with 1 matrix-vector multiplication or quantum cut query.
\end{enumerate}
Uniform star contraction for a given parameter $p$ can be implemented with just the first primitive: (i) pick a random subset of vertices $R$ by including every vertex independently at random with probability $p$, (ii) for every vertex in $R$ learn all its neighbors, explicitly 
giving the set $\cut(R)$, and (iii) for every vertex $v$ not in $R$, select a random edge in $\cut(R)$ incident to $v$ (if it exists).
Contracting the resulting star graphs yields (the supervertices of) the contracted graph $G'$.
We only make queries in step (ii).  By primitive~P1, this can be done with $O(|R|)$ queries, and this is $O(np)$ in expectation.

We can now easily sketch our algorithm for computing the edge connectivity of the input graph~$G$ using the above primitives:
\begin{enumerate}
\item
Compute the minimum degree $\delta(G)$ using primitive P3.
\item
If $\delta(G) \leq \sqrt{n}$, find a sparse $\delta(G)$-edge connectivity certificate using primitive P2. 
Output the edge connectivity of the connectivity certificate.
\item
If $\delta(G) > \sqrt{n}$, do uniform star contraction with $p \in \Theta(\log(n)/\delta(G))$, resulting in a contracted multigraph $G'$ that has $\tOh(\sqrt{n})$ vertices 
with high probability.
Run the randomized algorithm from \cite{MukhopadhyayN20} (\cref{thm:MN20}) to compute $\mincut(G')$.
\end{enumerate}
Step 1.~can be implemented with $O(\sqrt{n})$ matrix-vector multiplication or quantum cut queries by~P3.  Step 2.~costs $\polylog(n)$ queries per spanning forest by P2, thus $\tOh(\delta(G)) \in \tOh(\sqrt{n})$ queries overall.
In step 3.~we have $|R| = O(\sqrt{n} \log n)$ with high probability, in which case the star contraction can be done with $O(\sqrt{n} \log n)$ queries by P1.  The algorithm of  \cite{MukhopadhyayN20} 
can compute the weight of a minimum cut in a weighted $N$-vertex graph with high probability after $\tOh(N)$ classical cut queries, thus by primitive P4 we can compute $\mincut(G')$ with $\tOh(\sqrt{n})$ queries.

\subsection{Randomized cut-query algorithm}
Finally we describe our randomized $O(n)$ cut query algorithm for edge connectivity.  It does not seem possible to achieve this result using 
uniform star contraction as we did in the quantum cut query and matrix-vector multiplication query case.  The reason is that a vertex in $V \setminus R$ could 
have up to $|R|$  many neighbors in $R$ and it is too expensive to learn all of these neighbors with cut queries.  Instead we use another variation on star contraction 
that we call \emph{sparse star contraction}.  We show that instead of choosing a uniformly random neighbor in $R$, we can instead first learn a bipartite subgraph 
between $V \setminus R$ and $R$ where each vertex in $V \setminus R$ has \emph{constant} degree.  We then do 1-out contraction by independently choosing, for each $v \in V \setminus R$,
a uniformly random neighbor in this sparse bipartite subgraph.  Our main technical contribution is to show that this process can be done while preserving a non-trivial minimum cut with constant probability.  
We call this technique sparse star contraction as the contraction is performed on a bipartite subgraph with only $O(n)$ edges.
To actually learn this sparse bipartite subgraph with $O(n)$ cut queries we use our
second main tool, which is the \emph{separating matrix} framework of Grebinski and Kucherov \cite{GK98,GK00}.  We next elaborate on sparse star contraction and the separating matrix framework in 
more detail.

\paragraph{Sparse star contraction.}
To put sparse star contraction into context we begin with a more general scenario.  We can imagine a general form of a randomized contraction algorithm 
that first learns a spanning subgraph $H = (V, E')$ of the input graph $G=(V,E)$, and then for each vertex $v \in V$ independently at random selects an edge incident to $v$ in $H$.  
Finally, the selected edges are contracted in the original graph $G$.  In our case it will further be useful to think of $H$ as a \emph{directed} graph, where we will choose a random 
\emph{outgoing} edge from each vertex.  This point of view gives us more control over which endpoints can contract an edge.
We make the following definition.
\begin{definition}[Directed subgraph, 1-out contraction]
Let $G =(V,E)$ be a simple graph.  We say that the directed graph $H = (V,A)$ is a \emph{directed subgraph} of $G$ if 
every arc $(u,v) \in A$ satisfies $\{u,v\} \in E$.  In a \emph{random 1-out sample} of $H$, we independently and uniformly at random choose an outgoing edge in $H$
from every vertex that has one.  In a \emph{random 1-out contraction} of $H$, we take a random 1-out sample of $H$ and output the graph $G'$ obtained by 
contracting the sampled edges in $G$.
\end{definition}
Note that in a directed subgraph $H$, for any edge $\{u,v\}$ of $G$, we can either have both arcs $(u,v),(v,u)$ in $H$, just one of them, or neither of them.
When we speak about doing 1-out contraction on an undirected subgraph it should be interpreted that all edges are oriented in both directions.

With this terminology, uniform star contraction corresponds to doing a random 1-out contraction on the directed subgraph $H$ which is the induced bipartite graph between $V \setminus R$ and $R$ with all edges directed from 
$V \setminus R$ to $R$.  

There are two properties that we want in a directed subgraph $H$.  The first is that after a random 1-out contraction on $H$ the contracted graph does not have too many vertices.  We can automatically 
guarantee this condition by working in the star contraction framework.
The second is that in taking a random 1-out sample of $H$ we do not have too high probability of selecting an edge of a non-trivial minimum cut.  We can precisely extract a sufficient 
condition that makes a directed subgraph $H$ ``good for contracting'' in this second sense.
\begin{definition}[$(\alpha,\beta)$-good for contracting]
\label{def:good_contraction}
Let $G =(V,E)$ be a simple graph and $C \subseteq E$.  Let $H=(V,A)$ be a directed subgraph of $G$.  For every $u \in V$ let $q_u = \Pr_{v: (u,v) \in A} [\{u,v \} \in C]$ if $|\{(u,v) \in A : v \in V\}| > 0$ and $q_u = 0$ otherwise.
We say that $H$ is $(\alpha,\beta)$-good for contracting with respect to $C$ if it satisfies the following two conditions
\begin{enumerate}
	\item \emph{max property}: $\max_u q_u \le \alpha$, and
	\item \emph{sum property}: $\sum_u q_u \le \beta$.
\end{enumerate}
An undirected subgraph of $G$ is $(\alpha,\beta)$-good for contracting if and only if its directed version where all edges are directed in both directions is.
\end{definition}
As an example, it follows from \cref{eq:half} and \cref{eq:sumis2} used in the correctness proof of 2-out contraction that $G$ itself is $(1/2,2)$-good for contracting for any non-trivial minimum cut $C$.  In \cref{cor:sample_cor} we show 
that if $H$ is $(\alpha, \beta)$-good for contracting with respect to $C$ then the probability we do not select an edge of $C$ in a random 1-out sample of $H$ is at least $(1-\alpha)^{\ceil{\beta/\alpha}}$.

In sparse star contraction, we again start out by choosing a random set $R$ by taking each vertex $v$ to be in $R$ with probability $p$, although we take $p = \Theta(\log(\delta(G))/\delta(G))$ to be slightly smaller 
than what we used before.  With constant probability 
 the number of vertices in $V \setminus R$ with no neighbor in $R$ will be $O(n/\delta(G))$, and $R$ itself will satisfy $|R| = O(n \log(\delta(G))/\delta(G))$.  Let $H$ be the induced bipartite subgraph between $V\setminus R$ and 
 $R$ with all edges directed from $V \setminus R$.  In uniform start contraction we do a random 1-out contraction on $H$.  For the randomized cut query algorithm we will learn a sparse subgraph $H'$ of $H$ that has the property that every $v \in V \setminus R$ 
 that has an outgoing edge in $H$ also has an outgoing edge in $H'$.  No matter what $H'$ we take with this property we are guaranteed that after 1-out contraction the resulting contracted graph will have $O(n \log(\delta(G))/\delta(G))$ 
 vertices.  Our main technical contribution (in particular \cref{lem:preserve_main}) shows that we can find such an $H'$ that is $(\alpha, \beta)$-good for contracting for $\alpha < 1$ and small constant $\beta$ that has \emph{constant} degree.  
 As $H'$ only has $O(n)$ edges, we can hope to learn it with $O(n)$ cut queries, and we show this can indeed be done using the separating matrix machinery, described next.

\paragraph{Separating matrices and Recover-$k$-From-All.}
The second key tool of our algorithm is the \emph{separating matrix} machinery.
This toolset is best described by first considering an immediate obstacle to our $O(n)$ cut query algorithm for edge connectivity: an $O(n)$ bound is not even known for the simpler problem of determining if a graph is connected.
Harvey gave a deterministic $O(n \log n)$ cut query algorithm for connectivity \cite[Theorem 5.10]{HarveyThesis}, and we are not aware of any better result in the 
randomized case.  Besides the fact that connectivity is a special case of edge connectivity, our algorithmic framework will also heavily rely on being able to efficiently find spanning forests to construct sparse $r$-edge connectivity certificates.  

Harvey's connectivity algorithm, which can also find a spanning forest, is an implementation of Prim's spanning forest algorithm in the cut 
query model.  This algorithm can equally well be implemented with a weaker oracle that simply reports whether or not there exists an edge between two disjoint sets $S$ and $T$ (this is known in the literature as a bipartite independent set oracle, see e.g.\ \cite{BeameHRRS20}). Interestingly, Harvey's algorithm is actually optimal if restricted to this type of queries.
Indeed, any deterministic algorithm that determines connectivity 
while making use of an oracle that returns 1 bit of information must make $\Omega(n \log n)$ queries.  This follows from the aforementioned deterministic $\Omega(n \log n)$ 2-party communication complexity lower bound for 
deciding if a graph is connected \cite{HajnalMT88}.\footnote{We do not know of any such superlinear lower bound for randomized algorithms making bipartite indpendent set queries or the randomized communication complexity of connectivity. The best known bounded-error randomized communication complexity lower bound of $\Omega(n)$ follows from a reduction from set-disjointness \cite[Corollary 7.4]{BabaiFS86} or the inner-product mod 2 function \cite[Theorem 1]{IvanyosKLSW12} on $\Theta(n)$ bit inputs.}

The key now to both our zero-error $O(n)$ cut query algorithm for finding a spanning forest and our randomized $O(n)$ cut query algorithm for edge connectivity is to make use of the fact that a cut query actually returns 
$\Omega( \log n)$ bits of information.  This power was first harnessed by Grebinski and Kucherov \cite{GK98,GK00} who studied a related, but more powerful, query known as an additive query.  
In this model, when the input is a simple $n$-vertex graph with adjacency matrix $A$ one can query two Boolean vectors $x,y \in \{0,1\}^n$ and receive the answer $x^T A y$.   Grebinski and Kucherov 
\cite{GK00} showed the surprising fact that one can learn an $n$-vertex simple graph with only $O(n^2/\log n)$ additive queries, achieving the information theoretic lower bound.   
The main tool in the proof of Grebinski and Kucherov is the use of \emph{separating matrices}: the existence of an $O(n/\log n)$-by-$n$ matrix $B$ such that $Bx \ne By$ for any two distinct $n$-dimensional Boolean 
vectors $x$ and $y$.  In \cref{lem:learn} we use the separating matrix framework of Grebinski and Kucherov to show that if $S,T$ are disjoint subsets of $V$ that are polynomially related in 
size and 
$d_T(v) \le \ell$ for every $v \in S$, then we can learn all edges between $S$ and $T$ with only $O(\ell |S|)$ cut queries.  

This fact is the heart of the subroutine Recover-$k$-From-All (\cref{alg:recover}) which plays a key role in both the spanning forest and edge connectivity algorithms.  The input to this algorithm is two disjoint subsets 
$S,T \subseteq V$ 
that are polynomially related in size with the promise that $d_T(v) \ge \ell$ for every $v \in S$.  Recover-$k$-From-All is a zero-error randomized algorithm that can then learn $\min\{k,\ell\}$ neighbors in $T$ for 
every vertex $v \in S$ and makes $O(k |S|)$ cut queries in expectation.  Recover-$k$-From-All is based on ideas from $\ell_0$ sampling (e.g.\ Theorem 2 of \cite{JowhariST11}), which is similarly used in the 
connectivity algorithm in the semi-streaming model by Ahn, Guha, and MacGregor \cite{AGM12}.
First we put vertices in $S$ into $O(\log n)$ buckets by putting together those vertices with similar values of $d_T(v)$.  For the bucket $B$ with degree around $r$ into $T$, 
we randomly subsample a set $T' \subseteq T$ by putting each vertex of $T$ into $T'$ with probability $2k/r$.  We call a vertex in the bucket ``caught'' if $d_{T'}(v)$ is close to 
its expectation (e.g., it is in $[k,8k]$).  Letting $B' \subseteq B$ be the set of caught vertices, we can then learn $E(B',T')$ with $O(k|B'|)$ cut queries.\footnote{We assume that $B',T'$ are polynomially related in size 
for this high level description.  Handling smaller $B'$ is a technicality postponed to the full proof.}
This is repeated on all buckets until all vertices have been caught.  As we expect to catch a constant fraction of the remaining vertices in a bucket with each iteration, and the complexity of an iteration scales with the number of remaining vertices, one can argue that the expected 
number of queries overall is $O(k n)$.  

\paragraph{Edge connectivity.}
Now that we described the main tools, we can describe the main algorithm.
To make the exposition simpler, we begin with explaining how star contraction and Recover-$k$-From-All can be used to give a randomized $O(n \log \log n)$ cut query algorithm for edge connectivity.  This algorithm is given in \cref{sec:classical almost linear}.  We then describe the additional trick needed to reduce the query complexity to $O(n)$, which is given in \cref{sec:main_algo}.
The basic algorithm essentially follows the same steps as used in the quantum cut query case. 
\begin{enumerate}
	\item Compute the minimum degree $\delta(G)$.
	\item Perform sparse star contraction:
	 \begin{enumerate}
	 \item Choose a set $R$ by taking each $v \in V$ to be in $R$ with probability $\Theta(\log(\delta(G))/\delta(G))$.  Let $H$ be the directed subgraph obtained by picking every edge between $V \setminus R$ and $R$ and directing it from $V \setminus R$ to $R$.
	 \item Use Recover-$k$-From-All with constant $k$ on $V \setminus R$ and $R$ to learn a sparse subgraph $H'$ of $H$.
	 \item Do a random 1-out contraction on $H'$ and let $G'$ be the resulting graph.
	 \end{enumerate}
	\item Compute the edge connectivity of $G'$, and output the minimum of this and $\delta(G)$.
\end{enumerate}
As in the proof of \cref{thm:intro:simple star contraction}, we can again argue that with constant probability (i) $H$ is $(\alpha,\beta)$-good for contracting for some $\alpha < 1$ and constant $\beta$ (more specifically, for $\alpha = 3/5$ and $\beta = 8$) with respect to a non-trivial minimum cut, and (ii) that only $O(n/\delta(G))$ vertices in $V \setminus R$ have no neighbor in $R$.
Now, however, it is too expensive to learn the entire subgraph $H$ as in the quantum cut query algorithm, or even to independently sample a uniformly random neighbor in $R$ for each $v \in V \setminus R$ within the $O(n)$ query budget.
Instead, we use Recover-$k$-From-All with constant $k$ to learn a sparse subgraph $H'$ of 
$H$, where $H'$ has an outgoing edge for every $v \in V \setminus R$ that has one in $H$.  Moreover, the outgoing neighbors of $v$ in $H'$ are learned from a random set of vertices, conditioned on this 
set having at least one and not too many neighbors of $v$.  This can be done with $O(n)$ queries.  We then do a random 1-out contraction on the 
explicitly known graph $H'$.  Letting $G'$ be the result of this contraction, we finally compute the edge connectivity of $G'$ and output the minimum of this and $\delta(G)$.  

We postpone describing how we compute the edge connectivity of $G'$ and instead focus on showing that $H'$ is still $(\alpha',\beta')$-good for contracting for some 
$\alpha' < 1$ and constant $\beta'$.  As we use a
constant $k$ in Recover-$k$-From-All, we only expect to find a constant number of neighbors of a particular vertex $v$.  We have to show that, even in this very small sample, not 
too high a fraction of neighbors are on the opposite side of a non-trivial cut from $v$, \emph{for all} vertices $v$ incident on the cut.  In this low probability sampling regime, a Chernoff bound 
can only upper bound the failure probability for a particular vertex by a constant, which is not good enough as we have to union bound over the possibly growing number of vertices incident 
on the cut.

Instead, in \cref{lem:preserve_main}, we show the following statement.  Let $v \in V \setminus R$ and $C$ be a non-trivial minimum cut of $G$.  
Let $R' \subseteq R$ be chosen by putting each vertex of $R$ into $R'$ with 
probability $p = 2k/d_R(v)$ conditioned on $d_{R'}(v) > 0$.  We have already mentioned the fact that $\bbE_{R'}[c_{R'}(v)/d_{R'}(v) \mid d_{R'}(v) > 0 ] = c_R(v)/d_R(v)$.  We show that as long as $k \ge 10$
\[
\Pr_{R'}\left[\frac{c_{R'}(v)}{d_{R'}(v)} \ge \frac{c_R(v)}{d_R(v)} + \frac{1}{10} \mid d_{R'}(v) > 0 \right] \le \frac{200}{k} \frac{c_R(v)}{d_R(v)} \enspace .
\]
We know that $\sum_{v: c_R(v) > 0} c_R(v)/d_R(v) \le 8$ since $H$ is $(3/5,8)$-good for contracting.
Hence, by relating the failure probability to a sum that is bounded, and taking $k$ to be a large constant, we can again use a union bound to argue that $H'$ satisfies the max property with 
$\alpha' = 7/10$ with constant probability.

After a random 1-out contraction on $H'$ the contracted graph $G'$ has $O(n \frac{\log(\delta(G))}{\delta(G)})$ vertices.  Ideally, however, we would like it to only have 
$O(n/\delta(G))$ vertices.  As we describe next, we can then compute $\mincut(G')$ with $O(n)$ queries by finding a sparse $\delta(G)$-edge connectivity certificate.  
To further reduce the size of the contracted graph, we use Recover-$k$-From-All with $k = \Theta(\log(\delta(G)))$ to learn a directed subgraph $H_2$ of $G[R]$ where 
all but $O(n/\delta(G))$ vertices have outdegree $h=\Omega(\log(\delta(G))$.  This requires $O(\log(\delta(G)) |R|) = O(n)$ cut queries.
As was done with $H'$ we can similarly argue that $H_2$ is $(\alpha,\beta)$-good for contracting 
for yet other $\alpha < 1$ and constant $\beta$.  We then do a 2-out contraction on $H_2$.  We can use a lemma of \cite[Lemma 2.5]{GNT20} to conclude that 2-out contraction on $H_2$ reduces 
the number of vertices in $G[R]$ by a factor of $h$ and thus it becomes $O(n/\delta(G))$.  As all but $O(n/\delta(G))$ vertices in $V \setminus R$ are connected to a vertex in $R$ this reduces 
the number of vertices in $G'$ overall to $O(n/\delta(G))$.  

\paragraph{Spanning forests and sparse edge connectivity certificates.}
In order to accomplish step~3.~of the algorithm we show that we can construct a sparse $r$-edge connectivity certificate in a contracted graph with $q$ vertices using 
$O(n + r q \log(n)/\log(q))$ cut queries.  This lets us construct a sparse $\delta(G)$-edge connectivity certificate of $G'$ with $O(n)$ queries when $G'$ has $O(n/\delta(G))$ vertices.
In the final part of this section we give an overview of the key ideas that go into this algorithm and the obvious prerequisite of constructing a spanning forest with $O(n)$ cut queries.

Our spanning forest algorithm follows the framework of Bor\r{u}vka's spanning forest algorithm as has been used in several works related to matrix-vector multiplication queries
\cite{AGM12,LeeSZ21,AuzaL21}.  The application here requires several additional tricks to stay within the $O(n)$ query budget.

The algorithm proceeds in rounds and maintains the 
invariant that in each round there is a paritition $S_1, \ldots, S_t$ of $V$ and a spanning tree for 
each $S_i$ in the partition.  Initially, each $S_i$ is just a single vertex. In each round, it performs the following two steps:
\begin{enumerate}
    \item For each $S_i$, it finds a vertex that has at least one neighbor outside $S_i$. We call such vertices \textit{active} vertices. Whether or not a vertex is active can be determined with a constant number of cut queries by computing $|E(v,\bar{S}_i)|$ 
    for $v \in S_i$. We go over each $S_i$ looking for an active vertex; once we find an active vertex in $S_i$ we move on to $S_{i+1}$. The vertices that are discovered to be \textit{inactive} are ignored for all future rounds of the algorithm.
    \item Next we randomly bipartition the set of connected components and use Recover-$k$-From-All with constant $k$ to learn, for each active vertices on one side, a neighbor on the other side.  
    As in the case of Bor\r{u}vka's algorithm, we then combine the components which are connected by edges we discovered and reduce the number of components by a constant fraction.
\end{enumerate}
Note that, across all iterations of step 1, we make at most $n$ many useless queries (i.e., queries where we find a vertex to be inactive). So we only need to account for the query complexity of step 2. Here we crucially use the fact that we can reduce the number of connected components by a constant factor to show that total number of cut queries required over all invocations of 
Recover-$k$-From-All is bounded by $O(n)$.

The next task is to extend the spanning forest algorithm to also construct sparse edge connectivity certificates.  For our application we will want to construct a sparse edge connectivity certificate of the contracted graph $G'$, which is an integer weighted graph with $q \in o(n)$ vertices.  The most natural idea would be to extend our spanning forest algorithm to construct a spanning forest of such a graph while making only $O(q)$ cut queries.  We could then directly construct a sparse $r$-edge connectivity certificate with $O(rq)$ cut queries by iteratively finding one spanning forest at a time.  Unfortunately, 
we do not know how to find a spanning forest in a contracted graph more efficiently.  The reason is that if the adjacency matrix of the graph has entries with magnitude $M$ this introduces an extra $\log(M+1)$ factor into the separating matrix bounds, which we 
cannot afford.  

Instead we revisit the spanning forest algorithm and further parallelize Bor\r{u}vka's algorithm by \emph{simultaneously} building the different forests of the sparse edge connectivity certificate.\footnote{
Nagamochi-Ibaraki \cite{NI92} also construct all spanning forests of the sparse edge connectivity certificate in parallel.  They iterate over each edge of the graph and place it in the first spanning forest in which it does not create a cycle.  We cannot afford to iterate over all edges and instead modify Bor\r{u}vka's algorithm to build the spanning forests in parallel.}  We also crucially 
make use of the fact that $G'$ is not an arbitrary integer weighted graph, but the contraction of a simple graph $G$, and that for our application we can afford an extra additive $O(n)$ term.  
The fact that $G'$ is a contraction of a simple graph allows us to keep using the separating matrix machinery on \emph{Boolean} matrices by working on appropriate submatrices of the adjacency matrix of $G$, and the extra $O(n)$ term is used to identify these submatrices.

Let $F_1, \cdots, F_r$ be the spanning forests that we want to compute, recalling that $F_i$ is spanning forest in the graph $G \setminus (\bigcup_{j < i} F_i)$. Initially, each $F_i$ is empty. As before, we use steps 1.~and 2.~to find edges to extend these spanning tree. However, the crucial difference is the following: We find these edges with respect to the connected components of the last tree $F_r$, and we add each of these edges into the spanning forest $F_i$ for the least value of $i$ where it does not create a cycle. It is not hard to see that the set of connected components $\{S^{(i)}_1, \cdots, S^{(i)}_{t_i}\}$ of $F_i$ for different $i \in [r]$ form a laminar family: $\{S^{(i+1)}_1, \cdots, S^{(i+1)}_{t_{i+1}}\}$ is a refinement of $\{S^{(i)}_1, \cdots, S^{(i)}_{t_i}\}$. However, we cannot expect that the number of connected components of $F_r$ will decrease by a constant factor in each round as before. We can however show that it happens within $O(r)$ rounds. This, together with a similar accounting of cut queries as before, leads to the following theorem.
\begin{theorem}[Informal version, see \cref{thm:intro-certificate}] 
\label{thm:intro-certificate1}
Let $G = (V,E)$ be an $n$-vertex simple graph, and let $G' = (V',E')$ be a contraction of $G$ with $q$ supervertices, for $q$ sufficiently large. There is a zero-error randomized algorithm that makes $O(n + r q \log(n)/\log(q))$ cut queries in expectation and outputs a sparse $r$-edge connectivity certificate for $G'$.
\end{theorem}

\subsection{Open questions}
Our work raises some open questions that concern both upper bounds and lower bounds.
\paragraph{Lower bounds.}
The tightness of a number of our algorithms hinges on a positive answer to the following questions.
    \begin{itemize}
          \item
        Can we show a lower bound of $\Omega(n\log n)$ for the randomized two-party communication complexity of edge connectivity?
        The current best known bound in the randomized case is $\Omega(n \log\log n)$ \cite{AssadiD21},
        while the deterministic communication complexity of this problem is known to be $\Omega(n\log n)$~\cite{HajnalMT88}.
        A positive answer to this question implies an $\Omega(n)$ lower bound on the randomized cut query complexity of edge connectivity, showing that \cref{thm:intro:classical result} is tight. On the flip side, a randomized algorithm for edge connectivity making $o(n)$ cut queries would imply a negative answer to this question.  It is reasonable to think that a lower bound of of $\Omega(n \log n)$ for the randomized two-party communication complexity should hold even for the simpler problem of deciding if graph is connected, and we conjecture this to be true. Proving this would resolve the randomized communication complexity of connectivity, which has remained open since the work of Babai, Frankl, and Simon \cite{BabaiFS86}.
        \item
        Does computing the minimum degree of a simple graph indeed require $\Omega(\sqrt{n})$ quantum cut queries?
        As mentioned before, quantum minimum finding gives a simple $O(\sqrt{n})$ upper bound.
        By a reduction from minimum degree to edge connectivity (\cref{appensec:reduction}), a positive answer would imply that our $\tOh(\sqrt{n})$ quantum cut query algorithm for edge connectivity is tight (up to polylogarithmic factors).
    \end{itemize}

\paragraph{Upper bounds.}
Our algorithms could give rise to new algorithms and upper bounds in a number of ways.
    \begin{itemize}
        \item
        \textbf{Weighted graphs:} Can we find a minimum cut in a \emph{weighted graph} with $O(n)$ cut queries?
        This would not violate any currently known lower bound, and it would improve on the $O(n \, \polylog\, n)$ cut query algorithm from \cite{MukhopadhyayN20}.
        Similarly, can we find a minimum cut in a weighted graph with $o(n)$ quantum cut queries or matrix-vector queries?
        \item
        \textbf{(Approximate) edge connectivity with $\polylog(n)$ queries:} We mentioned that a key bottleneck for edge connectivity with quantum cut queries is computing the minimum degree, which might require $\Omega(\sqrt{n})$ quantum cut queries.
        In contrast, we can \emph{approximate} the minimum degree with $\polylog(n)$ quantum cut queries \cite{LeeSZ21}.
        The possibility hence remains of \emph{approximating} the edge connectivity with $\polylog(n)$ quantum cut queries.
        We also note that the minimum degree can be computed with exactly 1 matrix-vector query to the adjacency matrix.
        Hence, computing the edge connectivity with $\polylog(n)$ matrix-vector queries is also an interesting open question.
        \item
        \textbf{Zero-error and deterministic:}
        Can we make our zero-error $O(n)$ cut query algorithm for connectivity \emph{deterministic}?
        Alternatively, could we make our randomized $O(n)$ cut query algorithm for edge connectivity \emph{zero-error}?
        Both improvements would lead to tight algorithms (see \cref{table:bounds}).
    \end{itemize}

\subsection{Organization} \label{sec:organization}
 
In \cref{sec:prelim}, we define and state the necessary preliminary results that we need to state the technical details in the subsequent sections. In  \cref{sec:direct appl}, we give a formal proof of the correctness of uniform star contraction and provide three direct applications: in \cref{sec:quant_algo}, we design efficient quantum cut query and matrix-vector multiplication query algorithms (thereby proving \Cref{thm:intro:quantum result}), in \cref{sec:streaming}, we provide one-pass semi-streaming algorithms for edge connectivity in the complete and random vertex arrival model, and in \cref{sec:sequential}, we show a linear-time sequential algorithm for edge connectivity for dense graphs. Next, in \cref{sec:spanning forest}, we show a zero-error algorithm for computing a spanning forest with $O(n)$ cut queries.  In \cref{sec:classical almost linear} we give a randomized $O(n \log \log n)$ cut query algorithm 
for edge connectivity which combines many of the ideas from the previous sections.  Finally, in \cref{sec:classical edge conn}, we add one additional trick to give a randomized $O(n)$ cut query algorithm for edge connectivity, thereby proving \Cref{thm:intro:classical result}.

\section{Preliminaries}\label{sec:prelim}

\paragraph{Notation (strings, sets and matrices).} 
For a string $x \in \{0,1\}^n$ we use 
$|x|$ for the number of ones in $x$.
For a positive integer $n$ we let 
$[n]=\{1,\ldots, n\}$.  For a set $S \subseteq [n]$ we use $\bar{S}$ for the complement of $S$ and $|S|$ for its cardinality.  Given a matrix $M$ of dimension $k$-by-$\ell$, we denote the $\ith$ row of $M$ as $M(i,:)$. Note that $M(i,:)$ is a vector of dimension $\ell$. Given a set of column indices $R \subseteq [\ell]$, we define $M(i,R)$ to be the subvector of $M(i,:)$ that has entries restricted to the indices in $R$. Clearly $M(i,R)$ has dimension $|R|$. For a subset $R \subseteq [n]$, we let $\chi_R \in \{0,1\}^n$ be the characteristic vector of $R$, that is, $\chi_R(i) = 1$ if $i \in R$ and $\chi_R(i)=0$ otherwise.

\paragraph{Notation (graphs).}
Let $V$ be a finite set and $V^{(2)}$ the set of all $2$-element subsets of $V$.  We represent a weighted graph $G$ by the triple $G = (V,E,w)$, where $E \subseteq V^{(2)}$ is the set of edges and 
$w : E \rightarrow \mathbb{R}_{>0}$ assigns a positive weight to each edge.  For a subset of edges $F \subseteq E$ we let $w(F) = \sum_{f \in F} w(f)$.  
When the weight of every edge is $1$, we call the graph \emph{unweighted}.  We will 
consider 
two kinds of unweighted graphs, multigraphs and simple graphs.  In a multigraph $E$ is 
allowed to be a multiset of $V^{(2)}$ while 
in a simple graph it is simply a subset.  
In both cases we drop the weight function and write the graph as $G = (V,E)$.  
The inputs to our algorithms will always be simple graphs, and multigraphs only arise 
by considering contractions of a simple graph, described below.
We typically denote the number of vertices $|V| = n$ and the number of edges $|E| = m$.  

Let $G = (V,E,w)$ be a weighted graph.  For disjoint sets $S,T \subseteq V$ we let $E_G(S,T) = \{e \in E : |e \cap S| = |e \cap T| = 1\}$.  As shorthand, we use $\cut_G(S) = E_G(S,\bar{S})$ for the \emph{cut} defined by $S$, the set of 
edges with exactly one endpoint in $S$.  We will drop the subscript when the graph is clear from context.  We let $\mincut(G)= \min_{\emptyset \ne S \subsetneq V} w(\cut_G(S))$ be the weight of a minimum cut of $G$.  
When $G$ is unweighted we call $\mincut(G)$ the \emph{edge connectivity} of $G$.
The degree of a vertex $v$ is denoted by $d(v) = |\cut(\{v\})|$ and the number of edges from $v$ into a subset $R$ is $d_R(v) = |E(\{v\}, R \setminus \{v\})|$.
Finally, the minimum degree of a graph is denoted by $\delta(G) = \min_{v \in V} d(v)$.

We will also consider directed graphs $H = (V,A)$ where $A \subseteq \{(u,v) : u, v \in V, u \ne v\}$ is set of directed edges or arcs.  Directed graphs in this work will always arise by taking a (subgraph of a) simple 
graph and orienting the edges, possibly in both directions.  We will use analogous notations with arrows on top of them for directed graphs.  For disjoint sets $S,T$ we let $\overrightarrow{E}_H(S,T) = \{(u,v) \in A : u \in S, v \in T\}$ 
for the set of arcs directed from $S$ to $T$
and use the shorthand $\cutd_H(S) = \overrightarrow{E}_H(S,\bar{S})$ for the set of arcs leaving $S$.  Again we drop the subscript when the graph is clear from context.

Finally, in our algorithms we will look at \emph{contractions} of simple graphs.  For a simple graph $G = (V,E)$ and a subset of edges $F \subseteq E$, the multigraph $G' = (V',E')$ formed by \emph{contracting} the edges $F$ in $G$ 
is defined as follows.  $V'$ is the set of connected components of $F$.  We will sometimes refer to the vertices of $G'$ as \emph{supervertices} as they are sets of 
vertices of $G$.  For $S,T \in V'$ the number of edges between $S,T$ in $G'$ is $|E_G(S,T)|$.  We will sometimes instead view a contraction of a simple graph as an integer weighted graph, as a cut query cannot distinguish between these representations.

\paragraph{Useful bounds.}
Next, we define some useful algebraic bounds which we use in the analysis of our algorithms.

\begin{claim}
\label{clm:concave}
Let $b_1, \ldots, b_t$ satisfy $\sum_{i=1}^t b_i = n$ and $b_i \ge 1$ for all $i=1, \ldots, t$.  Then 
\[
\sum_{i=1}^t \frac{b_i}{\log(2b_i)} \le \frac{3n}{\log(2n/t)} \enspace .
\]
\end{claim}

\begin{proof}
Let $\euler =2.718\ldots$ be Euler's constant.
On the interval $[1, \infty)$ the function $x/\log(\euler^2 x)$ is concave and satisfies 
$x/\log(2x) \le 3 x/\log(\euler^2 x)$.  Therefore by Jensen's inequality
\begin{align*}
\sum_{i=1}^t \frac{b_i}{\log(2b_i)} &\le 3 \sum_{i=1}^t \frac{b_i}{\log(\euler^2 b_i)} \\
&\le \frac{3n}{\log(\euler^2 n/t)} \enspace. \qedhere
\end{align*}
\end{proof}

\begin{claim}
\label{clm:concave2}
Let $b_1, \ldots, b_k$ satisfy $\sum_i b_i = n$.  Then $\sum_{i=1}^k \log(b_i) \le k\log(n/k)$.
\end{claim}

\begin{proof}
By concavity of $\log(x)$ we have 
\[
\frac{1}{k} \sum_{i=1}^k \log b_i \le \log \left( \sum_{i=1}^k b_i/k \right) = \log(n/k) \enspace,
\]
which gives the claim.
\end{proof}

We will frequently make use of the following versions of the Chernoff bound.
\begin{lemma}
\label{lem:chernoff}
Let $X_1, \ldots, X_n$ be independent random variables taking values in $\{0,1\}$.  Let $X = \sum_{i=1}^n X_i$ and $\mu = \bbE[X]$.  Then
\begin{align}
\Pr[X \le (1-\delta) \mu] &\le \exp(-\delta^2 \mu/2) \text{ for any } 0 \le \delta \le 1
\label{eq:chernoff_below} \\
\Pr[X \ge (1+\delta) \mu] &\le \exp(-\delta^2 \mu/(2 + \delta)) \text{ for any } 0 \le \delta
\label{eq:chernoff_above}
\end{align}
In particular, if $\mu \ge 10$ then $\Pr[ \mu/2 \le X \le 2\mu] \ge 1/2$.
\end{lemma}

\subsection{Edge connectivity certificates}
Here we mention a few known results regarding so-called connectivity certificates (or, in our case, edge connectivity certificate), which we define next. 

\begin{definition}[sparse edge connectivity certificate]
Let $G = (V,E)$ be an $n$-vertex unweighted multigraph.  A sparse $r$-edge connectivity certificate for $G$ is a subgraph $\hat H = (V,F)$ with at most $rn$ edges 
and with the property that $|\cut_{\hat H}(S)| = \min\{r, |\cut_G(S)|\}$ for any $S \subseteq V$.
\end{definition}

A famous result of Nagamochi and Ibaraki \cite{NI92} gives a recipe for computing a sparse edge connectivity certificate efficiently by packing edge disjoint spanning forests.
\begin{theorem}[\cite{NI92}]
\label{lem:NI}
Let $G = (V,E)$ be an unweighted multigraph with $m$ edges and for $i \in [r]$ let $F_i$ be a spanning forest of $(V, E \setminus \bigcup_{j=1}^{i-1} F_j)$.  Then 
$\hat H = (V, \bigcup_{j=1}^r F_j)$ is a sparse $r$-edge connectivity certificate for $G$. In addition, this sparse certificate can be constructed by a deterministic algorithm in time $O(m)$.
\end{theorem}

\subsection{Cut query primitives}\label{sec:cut query prelims}

In this section, we state a few primitives in the cut query model that we use in various places of our algorithms. Most of these results appear in similar guise in \cite{RubinsteinSW18, MukhopadhyayN20}. We mention them here for completeness.

\begin{proposition}[Claim 5.1 in \cite{MukhopadhyayN20}]
\label{prop:3}
Let $G = (V,E,w)$ be a weighted graph and 
$S,T \subseteq V$ be disjoint sets.  The quantity $w(E(S,T))$ can be computed with 3 cut queries.
\end{proposition}
\begin{proof}
Let $A$ be the adjacency matrix of $G$, i.e.\ $A(u,v) = w(\{u,v\})$.
Then
\[
w(E(S,T))
= \chi_S^T A \chi_T
= \frac{1}{2} \left(\chi_{\bar S}^T A \chi_S + \chi_{\bar T}^T A \chi_T - \chi_{\overline{S \cup T}}^T A \chi_{S \cup T} \right) \enspace .
\]
The proposition follows by noting that $\chi_{\bar R}^T A \chi_R = w(\cut(R))$ can be evaluated with a single cut query, 
for any $R \subseteq [n]$.
\end{proof}

The following two corollaries follow easily from \Cref{prop:3}.

\begin{corollary} \label{cor:bip-query}
Let $G = (V,E,w)$ be a weighted graph.
For disjoint subsets $S,T \subset V$, let $G' = (S \cup T, E(S,T),w')$ be the induced bipartite subgraph between $S$ and $T$, where 
$w'$ is the weight function $w$ restricted to $E(S,T)$.  Let $A_{G'} \in \R^{|S| \times |T|}$ be the \emph{bipartite} adjacency matrix of $G'$. We can simulate a query of the form $x^T A_{G'} y$ for $x \in \{0,1\}^{|S|}$ and $y \in \{0,1\}^{|T|}$ with 3 cut queries to $G$.
\end{corollary}
\begin{proof}
Let $S_x \subseteq S$ and $T_y \subseteq T$ denote the supports of $x$ and $y$, respectively.
Then the corollary follows by noting that $x^T A_{G'} y = w(E(S_x,T_y))$ and using \cref{prop:3}.
\end{proof}

\begin{corollary}
Let $G = (V,E)$ be an unweighted multigraph. Let subsets $V' \subseteq V$ and $X \subseteq E$ be given explicitly, and define the graph $G' = (V',E')$ by setting $E' = \{\{u,v\} \in E \setminus X \mid u,v \in V'\}$.
Then we can simulate a cut query to $G'$ with 3 cut queries to $G$.
\end{corollary}
\begin{proof}
Consider $\cut_{G'}(S)$ in $G'$ for a vertex set $S \subset V'$. It is easy to observe that $|\cut_{G'}(S)| = |E_G(S, V'\setminus S) \setminus X|$. The subscripts $G$ and $G'$ clarify the corresponding graph associated with the edge set. We can compute $|E_G(S, V'\setminus S)|$ by making 3 cut queries to $G$ using \cref{prop:3}. As we know the set $X$ explicitly, we can compute $|E_G(S, V'\setminus S) \cap X|$ without queries. Subtracting the latter from the former we get the value of $|\cut_{G'}(S)|$.
\end{proof}

\begin{proposition}[Randomized binary search]
\label{prop:learn_random_neighbor}
Let $G=(V,E)$ be a simple graph. Let $v\in V$ and let $R\subseteq V\backslash \set{v}$. There is a randomized algorithm that, if $d_R(v)>0$, can output a uniformly random neighbor of $v$ in $R$ with $O(\log n)$ cut queries.
\end{proposition}
\begin{proof}
This is a simple variation on \cite[Corollary 2.2]{RubinsteinSW18} which we describe for completeness.

Use a single cut query to check whether $v$ has a neighbor in $R$.
If so, do the following recursively, until we end up with a single neighbor of $v$: split $R$ into two sets $R_1$ and $R_2$ with $|R_1| = \lceil |R|/2 \rceil$ and $|R_2| = \lfloor |R|/2 \rfloor$, and learn the number of neighbors $d_{R_1}(v)$ and $d_{R_2}(v)$ of $v$ in $R_1$ resp.~$R_2$.  Each of these can be computed with 3 cut queries by \cref{prop:3}.
Now pick $R_1$ with probability $d_{R_1}(v)/(d_{R_1}(v)+d_{R_2}(v))$ and otherwise pick $R_2$.
If the chosen set has a single vertex then return that vertex, otherwise recurse on the chosen set.
\end{proof}

Finally, a much more involved primitive is the following randomized cut query algorithm for finding a minimum cut in a weighted graph by Mukhopadhyay and Nanongkai \cite{MukhopadhyayN20}.
\begin{theorem}[{\cite[Section 5.1]{MukhopadhyayN20}}]
\label{thm:MN20}
Let $G = (V,E,w)$ be a weighted graph with $n$ vertices.  There is a randomized algorithm that computes $\mincut(G)$ with high probability after $O(n \log^8 n)$ cut queries.\footnote{The bound is stated as $\tOh(n)$ in \cite{MukhopadhyayN20} but 
we need a concrete exponent in the polylog term to make our algorithm explicit.  We take $\log^8(n)$ as a conservative over-estimate; the true cost of the algorithm in \cite{MukhopadhyayN20} is likely smaller.}
\end{theorem}

\section{Star contraction and direct applications} \label{sec:direct appl}
In this section we give the full proof of the uniform star contraction theorem stated in the introduction (\cref{thm:intro:simple star contraction}). We then derive several direct applications of it.

\subsection{$(\alpha,\beta)$-good for contracting}
We prove some preliminaries to set the scene.
First, we verify that if a directed subgraph $H$ is $(\alpha,\beta)$-good for contracting for $C$ (\cref{def:good_contraction}) then in taking a random 1-out sample of $H$ we do not have 
too high a probability of selecting an edge of $C$.  Let $G =(V,E)$ be a simple graph and $C \subseteq E$.  Let $H=(V,A)$ be a directed subgraph of $G$.  For every $v \in V$ let $q_v = \Pr_{u: (v,u) \in A} [\{u,v \} \in C]$ if $|\{(v,u) \in A : u \in V\}| > 0$ and $q_v = 0$ otherwise.  The probability we \emph{do not} choose an edge of $C$ when taking a 
random 1-out sample on $H$ is exactly
\[
\prod_{v \in V} ( 1 - q_v ) \enspace.
\]
We can lower bound this probability via the next proposition.  A very similar statement is given in \cite[Lemma 2.7]{GNT20}; we give an alternative 
analysis that improves the bound by a constant factor.
\begin{proposition}
\label{clm:sample}
Let $n$ be a positive integer, $0 \le \alpha < 1$, and $\beta \ge 1$.  Define
\begin{equation*}
    \begin{aligned}
        F(\alpha,\beta) = \ &\underset{x \in \R^n}{\text{minimize}}&& \prod_i (1-x_i) \\
        &\text{subject to} && \sum_i x_i = \beta,\\
        &&& 0 \le x_i \le \alpha \enspace .
    \end{aligned}
\end{equation*}
Then $F(\alpha,\beta) \ge (1-\alpha)^{\ceil{\beta/\alpha}}$.
\end{proposition}

\begin{proof}
We instead analyze $\ln(F(\alpha,\beta))$ whose objective function is $\sum_i \ln(1-x_i)$.
The problem then becomes the minimization of a continuous concave function over a compact convex set $K$.  By the Krein-Milman theorem 
\cite[Section 3.23]{Rudin}, $K$ is the closed convex hull of its extreme points and therefore by Jensen's inequality a global optimum will occur at an extremal point. It is 
clear that $F(\alpha,\beta) \ge F(\alpha,\beta')$ for $\beta \le \beta'$.  Therefore we instead lower bound $F(\alpha,\beta')$ with $\beta' = \alpha \ceil{\beta/\alpha}$.  
The extremal points in the set $0 \le x_i \le \alpha, \sum_i x_i = \beta'$ have $\beta'/\alpha$ non-zero entries which are all equal to $\alpha$.  Plugging this in gives the bound.
\end{proof}

\begin{corollary}
\label{cor:sample_cor}
If $H$ is $(\alpha,\beta)$-good for contracting with respect to $C$ then the probability an edge of $C$ is not selected in taking a random 1-out sample of $H$ is 
at least $(1-\alpha)^{\ceil{\beta/\alpha}}$.
\end{corollary}

Finally, we note a simple property of an $(\alpha,\beta)$-good for contracting subgraph.
\begin{proposition}
\label{prop:remove_out}
Let $G = (V,E)$ be a simple graph, $C \subseteq E$ 
a set of edges, and $H$ a directed subgraph of $G$ that is $(\alpha,\beta)$-good for contracting with 
respect to $C$.  For $v \in V$ if we form the 
graph $H'$ from $H$ by removing all outgoing edges of $v$ in $H$ then $H'$ is also $(\alpha,\beta)$-good for 
contracting with respect to $C$.
\end{proposition}

With these preliminaries in hand to set the scene, we next give the proof of the uniform star contraction theorem (\cref{thm:intro:simple star contraction}).
There a subset $R \subseteq V$ is chosen randomly and $H$ is taken to be $(V, \cutd(V \setminus R))$, where $\cutd(V \setminus R)$ denotes the edges in $\cut(V \setminus R)$ directed 
from $V\setminus R$ to $R$. 

\subsection{Full proof of uniform star contraction} \label{sec:simple-star-contr}
In this section we prove \cref{thm:intro:simple star contraction}, which for convenience is restated here.
\simplestar*

We first show two propositions that will help to prove \cref{thm:intro:simple star contraction}.  In uniform star contraction we do a random 1-out contraction on the 
graph $H=(V, \cutd(V \setminus R))$ where $R$ is chosen randomly by putting each vertex into $R$ with probability $p$.  As in doing 1-out contraction on $H$ 
we choose a neighbor of every vertex in $V\setminus R$ that has one, the total number of vertices in $G'$ can be upper bounded by $|R|$ plus the number of 
vertices in $V \setminus R$ that have no neighbor in $R$.  The next proposition bounds the probability that these values are much larger than their expectations.
\begin{proposition}
\label{prop:sizeR}
Let $G = (V,E)$ be an $n$-vertex graph.  
Randomly choose a set of vertices $R$ by putting each vertex into $R$ independently at random with probability $p$.  
Then 
\begin{enumerate}
    \item $\Pr_R[|R| \ge 2pn] \le \exp(-pn/3)$.
    \item $\Pr_R[d_R(v) \le 0.9 p d(v)] \le \exp(-p d(v)/200)$ for any $v \in V$.
\end{enumerate}
\end{proposition}

\begin{proof}
Both items follows directly from a Chernoff bound.
For the first item we use \cref{eq:chernoff_above} and for the second \cref{eq:chernoff_below}.
\end{proof}
\cref{prop:sizeR} will handle item~1 of \cref{thm:intro:simple star contraction}.  We next show that $H=(V, \cutd(V \setminus R))$ is 
$(2/3,8)$-good for contracting with probability at least $2/3$.  This will imply item~2 by \cref{cor:sample_cor}.   We single out showing that $H$ has the ``sum property'' (item~2 of \cref{def:good_contraction}) with 
constant probability in 
the next proposition.  We go ahead and prove a slightly more general statement than is needed for 
\cref{thm:intro:simple star contraction}, but which will be 
used in the $O(n)$ randomized cut query algorithm for edge 
connectivity.  It is also interesting to note that this statement holds for any sampling probability $p$.

\begin{proposition}
\label{prop:expectation}
Let $G = (V,E)$ be a simple $n$-vertex graph and let $C \subseteq E$.  Let $0 < f \le g \le n$ be positive integers.  Choose a set $R$ by putting each vertex of $V$ into $R$ independently at random with probability $p$.  Then
for any $v \in V$
\begin{equation}
\bbE_R\left[ \frac{c_R(v)}{d_R(v)}  \;\bigg|\; f \le d_R(v) \le g \right] = \frac{c(v)}{d(v)} \label{expectation} \\
\end{equation}
\end{proposition}

\begin{proof}
Let $d$ be the degree of $v$ and $c$ be the number of edges of $C$ incident to it.  The proposition is equivalent to the following 
purely probabilistic statement.  Let $X_1, \ldots, X_c, Z_1, \ldots, Z_{d-c}$ be independent and identically distributed Bernoulli random variables that are $1$ with probability $p$.  
$X_i=1$ represents the event that the $\ith$ edge of $C$ incident to $v$ is selected, and $Z_i=1$ the event that the $\ith$ non-edge of $C$ incident to $v$ is selected.  
Then $X = \sum_{i=1}^c X_i$ is the random variable for the total number of edges of $C$ incident to $v$ selected and $Y = X + \sum_{i=1}^{d-c}Z_{i}$ is the random 
variable for the total number of edges incident to $v$ selected.
We want to show $\bbE[X/Y \mid f \le Y \le g] = c/d$. 

For $0 < b \le d$ let us first compute $\bbE[X/Y \mid Y = b] = (1/b) \bbE[X \mid Y = b]$.  We claim that $\bbE[X \mid Y = b] = cb/d$.  By 
linearity of conditional expectation, $\bbE[X \mid Y = b] = \sum_{i=1}^c \bbE[X_i \mid Y = b]$.  As each $X_i$ and $Z_j$ are identically distributed, $\bbE[X_i \mid Y=b]$ is independent of $i$ 
and also equal to $\bbE[Z_j \mid Y=b]$ for $j = 1, \ldots, d-b$.  Thus $b = \sum_{i=1}^c \bbE[X_i \mid Y=b] + \sum_{j=1}^{d-c} \bbE[Z_j \mid Y=b] = d \cdot \bbE[X_i \mid Y = b]$ 
for any $i$.  This implies $\bbE[X_i \mid Y = b] = b/d$ and so $\bbE[X \mid Y = b] = cb/d$.

As $\bbE[X \mid Y = b] = cb/d$ for any $0 < b \le d$, we directly obtain $\bbE[X/Y \mid f \le Y \le g] = c/d$ for any integers $0 < f \le g$.  
\end{proof}

\begin{lemma}
\label{lem:good_for_contracting}
Let $G = (V,E)$ be a simple graph and $C$ be a non-trivial minimum cut of $G$.  Choose a set $R$ by putting each $v \in V$ into 
$R$ independently at random with probability $p \ge 1200 \ln(n)/ \delta(G)$.  Then the directed subgraph $H = (V, \cutd(V \setminus R))$ is 
$(2/3,8)$-good for contracting with respect to $C$ with probability at least $2/3$.
\end{lemma}

\begin{proof}
As $C$ is a non-trivial minimum cut we know that
$\sum_{v: c(v) > 0} c(v)/d(v) \le 2$ (see \cref{eq:sumis2} 
and following discussion).  Together 
with \cref{prop:expectation}, 
linearity of expectation, and Markov's inequality gives 
\[
\Pr_R \left[ \sum_{v: c_R(v) > 0} \frac{c_R(v)}{d_R(v)} \ge 8 \right] \le 1/4 \enspace,
\]
showing the ``sum property'' of \cref{def:good_contraction} holds except with probability $1/4$.  Note that we did not have to use the
fact that $p \ge 1200 \ln(n)/ \delta(G)$ here, this holds for any $p$.

It remains to show the ``max property'' of \cref{def:good_contraction}, where we do use the assumption that $p$ is not too small.
For any $v \in V$ we have $d_R(v) \ge 0.9 p d(v)$ except with probability at most $\exp(-6\ln n)$ by item~2 of \cref{prop:sizeR}.  
Thus except with probability $n^{-5}$ this will hold 
for all $v$ by a union bound.  Let us add this to our error probability and assume this holds in the rest of the proof.  
The expected value of $c_R(v)$ is $p c(v)$ and we 
know that $c(v)/d(v) \le 1/2$ because $C$ is non-trivial.  Thus $\bbE[c_R(v)] \le p d(v)/2$ and to have $c_R(v)/d_R(v) > 2/3$ we must have $c_R(v) > 1.2 \bbE[c_R(v)]$ and 
$c_R(v) > 720 \ln n$, the latter because we are in the case $d_R(v) \ge 0.9 \cdot 1200 \ln n$.  A Chernoff bound thus gives the probability this 
happens is at most $n^{-8}$.  Hence by a union bound this will be true for all $v$ except with probability at most $n^{-7}$.  

This shows that $H = (V, \cutd(V \setminus R))$ is $(2/3,8)$-good for contracting with respect to $C$ except with probability at most $1/4+n^{-5} + n^{-7} \le 1/3$ over the choice of $R$.
\end{proof}

With these tools in hand, we can now prove 
\cref{thm:intro:simple star contraction}.

\begin{proof}[Proof of \cref{thm:intro:simple star contraction}]
Let us first show item~1.  The number of vertices in $G'$ can be upper bounded by $|R|$ plus 
the number of vertices in $V \setminus R$ that have no neighbor in $R$.  By item~1 of \cref{prop:sizeR}  
$|R|$ will be at most twice its expected value, which is $2400 n\ln(n)/\delta(G)$, except with probability at most $n^{-400}$.  By item~2 of \cref{prop:sizeR} the probability 
that a vertex has no neighbor in $R$ is at most $\exp(-6 \ln n)$.  Thus by a union bound, except with probability $n^{-5}$, every 
vertex in $V \setminus R$ will have a neighbor in $R$.  Both items will 
hold with probability at least $1-1/n^4$.  This completes the proof of item~1.

For item~2, let $C$ be a non-trivial cut realizing $\mincut(G)$.  With the choice of $p$ in the theorem we know that $H = (V, \cutd(V \setminus R)$ is $(2/3,8)$-good for contracting with respect to $C$ with 
probability at least $2/3$ over the choice of $R$ by \cref{lem:good_for_contracting}.  
When this happens, we do not select an edge of $C$ in doing a random 1-out sample of $H$ with probability at least $3^{-12}$ by \cref{cor:sample_cor}.  Not selecting an edge of $C$ to contract 
implies that in the contracted graph $G'$, for every super-vertex $S$, all vertices in $S$ will lie on the same side of the cut $C$.  This in turn gives $\mincut(G) = \mincut(G')$.
Thus overall item~2 of the theorem holds with 
probability at least $2 \cdot 3^{-13}$.  
\end{proof}

\subsection{Quantum cut queries and matrix-vector multiplication queries} \label{sec:quant_algo}

In this subsection we consider randomized algorithms that can make matrix-vector multiplication queries to the adjacency matrix of a graph 
and quantum algorithms with cut queries to a graph.  In the 
former, if $A$ is the adjacency matrix of an $n$-vertex simple graph $G$, one can query a vector $x \in \{0,1\}^n$ and receive the answer $Ax$.  

As described in \cref{sec:tech_mvp}, our proof works for any query model that has the following primitives, which we now state more formally.  We say 
that an algorithm that can perform these operations in the stated costs has \direct access to the graph, which stands for ``Minimum Degree and Cut Product.''  The 
origin of the name is explained below in \cref{rem:MDCP}.
\begin{definition}[\direct access]
Let $G=(V,E)$ be a simple graph.  We say 
a query model has \direct access to $G$ if it can 
execute the following query 
operations with high probability in the specified 
cost.
\begin{enumerate}
    \item Minimum degree queries: 
    One can query $\mindeg(G)$ and receive as answer the minimum degree of $G$. The cost is 
    $O(\sqrt{n} \log n)$.
    \item Neighborhood queries: For any vertex $v \in V$, the characteristic vector $\nbh(v) \in \{0,1\}^n$ of $v$ can be computed with cost $O(\log n)$.
    \item Spanning forest queries: For any subset $E' \subseteq E$ known to the algorithm one can compute a spanning forest $\spf(G,E')$ of 
    $G'=(V, E\setminus E')$.  The cost is 
    $O(\log^6 n)$.
    \item Cut queries: For any $S \subseteq V$ one can compute $|\cut(S)|$ with cost $O(1)$.
\end{enumerate}
\end{definition}

The formulation of the spanning forest query may seem unusual.  It is formulated in this way to allow the computation of sparse edge connectivity certificates using the technique of Nagamochi-Ibaraki given in \cref{lem:NI}.  Further, both matrix-vector multiplication queries and quantum cut queries allow one to compute a spanning forest in a polylogarithmic number of queries, as shown later in this section (\cref{lem:simulate_direct}).

\begin{proposition}
\label{prop:direct_cert}
Let $G = (V,E)$ be an $n$-vertex simple graph and $0 < r \le n$ an integer. There is a deterministic algorithm with \direct access to $G$ that computes a sparse $r$-edge connectivity certificate for $G$ with cost $O(r \log^6 n)$.
\end{proposition}

\begin{proof}
This follows directly from the construction of an $r$-sparse edge connectivity certificate by Nagamochi and Ibaraki \cite{NI92} given in \cref{lem:NI}. Let $F_0 =\emptyset$.  For $i=1, \ldots r$ one computes $F_i = \spf(G, \bigcup_{j=0}^{i-1} F_j)$ with a query of cost $O(\log^6 n)$.  At stage $i$ the edges in $\bigcup_{j=0}^{i-1} F_j$ are all known to the algorithm from previous queries so this is a valid query.  The total cost is $O(r \log^6 n)$ and a sparse $r$-edge connectivity certificate is given by $\bigcup_{j=1}^{r} F_j$.
\end{proof}

\begin{algorithm}[!htbp]
\caption{\direct query algorithm for edge connectivity}
\label{alg:simple_ec}
 \hspace*{\algorithmicindent} \textbf{Input:} 
 \direct query access to a simple graph $G$ \\
 \hspace*{\algorithmicindent} \textbf{Output:} $\mincut(G)$
\begin{algorithmic}[1]
\State $\delta \gets \mindeg(G)$.
\label{line:min_degree}
\If{$\delta < \sqrt{n}$} 
    \State Compute a sparse 
    $\sqrt{n}$-edge connectivity certificate $F$ of $G$. 
    \State Output $\mincut(F)$.
    \label{line:small_deg}
\Else 
    \State $\best \gets \infty$
    \ForAll{$i=1, \ldots, \ceil{100 \log n}$} \label{line:outer_for}
        \State $R \gets \emptyset$
        \label{line:initR}
        \ForAll{$v \in V$}
            Put $v$ in $R$ at random with probability $1200 \ln(n)/\delta$.
        \EndFor
        \If{$|R| > 2400n \ln(n)/\delta$}     Abort.
            \label{line:abort}
        \EndIf
        \ForAll{$c \in R$} 
            Query $\nbh(c)$.   \label{line:Rnbh}
        \EndFor
        \State $X \gets \emptyset$
        \ForAll{$v \in V \setminus R$}
            Choose a neighbor $u \in R$ of $v$ uniformly at random (if it exists) and add $\{u,v\}$ to $X$. 
            \label{line:sample}
        \EndFor
        \State Let $G'$ be the multigraph formed from $G$ by contracting all edges in $X$.
        \label{line:contract}
        \If{$G'$ has more than $2400n \ln(n)/\delta$ vertices}     Abort.
            \label{line:abort2}
        \EndIf
        \If{$\mincut(G') < \best$} $\best \gets \mincut(G')$.
        \label{line:MN}
        \EndIf
    \EndFor
    \State Output the minimum of $\delta$ and $\best$. 
\EndIf
\end{algorithmic}
\end{algorithm}

\begin{theorem}\label{thrm:quant_algo}
Let $G$ be an $n$-vertex simple graph. 
There is a randomized algorithm with \direct access to $G$ that outputs the 
edge connectivity of $G$ with high probability and has total query cost of $\tOh(\sqrt{n})$.
\end{theorem}

\begin{proof}
The algorithm is given in 
\cref{alg:simple_ec}.  We go through the 
steps to describe in more detail their implementation and give their cost.

In the first step we compute the minimum degree $\delta$ of the graph, which can be done with cost $O(\sqrt{n}\log n)$.  The rest of the algorithm breaks down into two cases depending on $\delta$.

If $\delta < \sqrt{n}$ then 
$\mincut(G) < \sqrt{n}$ as well.  Thus by
definition, if $F$ is a sparse $\sqrt{n}$-edge connectivity certificate for $G$ we will have $\mincut(G) = \mincut(F)$ and the algorithm will output correctly in line~\ref{line:small_deg}.  
The cost of computing a sparse $\sqrt{n}$-edge connectivity certificate is $O(\sqrt{n} \log^6 n)$ by
\cref{prop:direct_cert}.  Thus overall 
in the low degree case the query cost is 
$O(\sqrt{n} \log^6 n)$.

Let us now consider the case where 
$\delta \ge \sqrt{n}$.  We first 
describe 
the implementation of each step in the for loop beginning on line~\ref{line:outer_for}
and its cost.  
In each iteration of the for loop
we do uniform star contraction with 
$p = 1200 \ln(n)/\delta$ as detailed 
in lines \ref{line:initR}--\ref{line:contract}. By item~1 of 
\cref{thm:intro:simple star contraction} 
the probability that we abort in 
line~\ref{line:abort} or line~\ref{line:abort2}
is $n^{-4}$, thus the probability 
we abort in any iteration of the for loop is $O(n^{-3})$.  We add this to our error 
bound and henceforth assume this does not 
happen.  The cost of 
line~\ref{line:Rnbh} is thus $O(n \log^2(n)/\delta)$ as we are assured in 
this step that 
$|R| \le 2400 \ln(n)/\delta$.  

Once we know all neighbors of all vertices in $R$, for each $v \in V \setminus R$ we can compile a list $B(v)$ of its neighbors in $R$.  We then use these lists to execute line~\ref{line:sample}.  

Next we argue that \direct access to $G$ gives us cut query access to the multigraph $G'$ constructed in line~\ref{line:contract}.  We can compute
the connected components of the set of edges $X$.  Say that this gives the 
partition $\Pcal = \{A_1, \ldots, A_t\}$ of $V$.  Then $\Pcal$ is the vertex set 
of $G'$ and by \cref{thm:intro:simple star contraction} we have 
$|\Pcal| = O(n \log(n)/\delta)$ with 
high probability.  To execute the 
cut query $|\cut_{G'}(T)|$ for $T \subseteq \Pcal$, note that
$|\cut_{G'}(T)| = |\cut_G(S)|$ 
where $S = \bigcup_{A \in T} A$.  Thus 
with \direct access to $G$ we can answer cut queries to $G'$.  Therefore 
to execute line~\ref{line:MN} we can run 
the randomized cut query algorithm of \cite{MukhopadhyayN20} quoted in \cref{thm:MN20} to compute
$\mincut(G')$ with high probability. The cost of this step is $O(n \log^{9}(n)/\delta)$.  Thus over all $O(\log n)$ 
iterations of the for loop the total cost is 
$O(n \log^{10}(n)/\delta)$.  As $\delta \ge \sqrt{n}$ the total cost in this case is 
$O(\sqrt{n} \log^{10}(n))$.

Now let us argue correctness in the high 
degree case.  Even when item~2
of \cref{thm:intro:simple star contraction} does not hold in the 
star contraction, the graph
$G'$ is always a contraction of $G$ and thus 
$\mincut(G') \ge \mincut(G)$.  As we compute 
$\mincut(G')$ correctly with high probability 
in line~\ref{line:MN}, with high probability 
we have $\best \ge \mincut(G)$ and so the algorithm 
will output correctly if $\mincut(G) = \delta$.  
If $\mincut(G) < \delta$ by \cref{thm:intro:simple star contraction} we will have $\mincut(G') = \mincut(G)$ with probability at least $2\cdot 3^{-13}$ in each 
iteration of the for loop.  Thus in this case, as 
we repeat the star contraction $100 \log n$ times
and take the minimum result, with high probability we will have $\best = \mincut(G)$, and the algorithm will 
output correctly.
\end{proof}

Next we show that with both quantum cut queries and 
matrix-vector multiplication queries to the adjacency matrix of $G$ we can simulate \direct access to $G$.

\begin{lemma}
\label{lem:simulate_direct}
Let $G$ be a simple graph.
A quantum algorithm with cut query access to $G$ and a randomized algorithm with matrix-vector multiplication queries to the adjacency matrix of $G$ both have \direct access to $G$.
\end{lemma}

\begin{proof}
Let $A$ be the adjacency matrix of $G = (V,E)$.  For two vectors $x,y \in \R^n$ let $x \circ y \in \R^n$ be their entrywise product, and 
let $\one$ be the $n$-dimensional all one vector.
Lee, Santha, and Zhang \cite{LeeSZ21} use a generalization of the Bernstein-Vazirani algorithm \cite{BV97} 
to show that for any $x \in \{0,1\}^n$ a quantum algorithm can with certainty compute $Ax \circ (\one-x)$ with $O(\log n)$ cut queries.
This power was stated explicitly in \cite[Corollary 11]{AuzaL21}.  One can also 
clearly compute $Ax \circ (\one -x)$ with one matrix-vector 
query to the adjacency matrix.  We postpone item~1 for now and first discuss 
items~2 and~3 together for the two models
in terms of $Ax \circ (1-x)$ queries.

For item~2 note that we can learn the neighborhood 
of $v$ via the query $Ae_v \circ (\one-\chi_{\{v\}})$.

For item~3, Auza and Lee \cite[Theorem 8]{AuzaL21},
in simplifying and quantitatively improving the results of \cite{LeeSZ21}, show that one can compute a spanning forest of $G$ with high probability with 
$O(\log^5 n)$ many $Ax \circ (\one-x)$ queries.  
Say that one explicitly knows a set of edges $E' \subseteq E$.
Let $B$ be the adjacency matrix of the graph $(V,E')$ and $A'$ be the adjacency matrix of the 
graph $G' = (V, E\setminus E')$.  Thus $A' = A-B$ 
and we can compute 
$A' x \circ (\one -x) = A x \circ (\one -x) - Bx \circ (\one -x)$ from $A x \circ (\one-x)$ as $Bx \circ (\one-x)$ is known 
explicitly.  This allows us to also compute 
spanning forests of $G'$ with high probability 
in $O(\log^5 n)$ many $A x \circ (\one -x)$ queries 
as well.

For item~4 it is clear that a quantum cut query 
algorithm can compute a classical cut query. Since $(\one-x)^T A x = (\one-x)^T (A x \circ (\one-x))$, one can also compute a cut query with one $A x \circ (\one-x)$ query.

Item~1 is where the arguments diverge.  It is not 
obvious how to compute the minimum degree with 
$o(n)$ queries of the form $Ax \circ (\one-x)$.  
However, the minimum degree can be computed with 
one matrix-vector multiplication query as all degrees are given by $A \one$.

For the quantum case, note that the degree of one 
vertex can be computed with a single classical cut 
query.  Thus we can find the minimum degree with 
high probability in $O(\sqrt{n} \log n)$ quantum 
cut queries using the quantum minimum finding routine
of D\"{u}rr and H{\o}yer \cite[Theorem 1]{DurrH96}.
\end{proof}

\begin{remark}
\label{rem:MDCP}
As can be seen from the proof of \cref{lem:simulate_direct}, to simulate \direct access it suffices to be able to
\begin{enumerate}
	\item Compute $\mindeg(G)$ with high probability in cost $O(\sqrt{n} \log(n))$.
	\item Compute $Ax \circ (\one -x)$ with high probability in cost $O(\log n)$ for any $x \in \{0,1\}$, where $A$ is the adjacency matrix of $G$.  We call this a \emph{cut product}.
\end{enumerate}
This is a smaller set of primitives that can still be used in \cref{thrm:quant_algo} to give a $\tOh(\sqrt{n})$ cost algorithm for edge connectivity, and is the origin of the name 
Minimum Degree and Cut Product.  We chose to define 
\direct access with a more verbose but less mysterious set of primitives for greater clarity.
\end{remark}

\subsection{One-pass semi-streaming algorithms} \label{sec:streaming}
Next, we consider applications of star contraction for edge connectivity computation in various settings of the streaming model of computation. Specifically, we consider the following settings.
\begin{enumerate}
    \item \textit{Explicit vertex arrivals}, in which vertices appear in an arbitrary order, along with all edges incident to previously seen vertices.
    \item  \textit{Complete vertex arrivals}, in which vertices appear in an arbitrary order, along with all incident edges.
    \item \textit{Random vertex arrivals}, in which vertices appear in a random order, along with all edges incident to previously seen vertices.
\end{enumerate}

For a related model of streaming computation where the edges arrive in arbitrary order, \cite{Zelke11} showed a $\Omega(n^2)$ space lower bound against any one-pass randomized algorithm that correctly computes edge connectivity. We observe that this lower bound construction also gives a lower bound of the same strength for the explicit vertex arrival setting.\footnote{This does not follow as a black box reduction. Rather, the edge stream admitted by the lower bound construction of \cite{Zelke11} is actually an explicit vertex arrival stream. Hence the same argument provides an $\Omega(n^2)$ space lower bound.}
We include a proof sketch in \cref{appensec:streaming}.

\begin{restatable}{observation}{streamreduction}[Follows from \cite{Zelke11}]
\label{obs:lb_vertex_arrival}
Any one-pass streaming algorithm computing the edge connectivity of a simple graph in the explicit vertex arrival setting requires $\Omega(n^2)$ memory.
\end{restatable}

Next, we employ uniform star contraction to prove that the random vertex arrival setting allows to circumvent the aforementioned lower bound.
In contrast, we discuss in \cref{remark:2-out-stream} that it is not clear how to use the related 2-out contraction technique for this purpose.
Formally, we prove the following.

\begin{theorem}
\label{thrm:random_vertex_arrival}
There is a one-pass streaming algorithm, using $\tOh(n)$ memory, that given a simple graph $G=(V,E)$ in the random vertex arrival setting, computes the edge connectivity of $G$ with high probability.
\end{theorem}

\begin{proof}
We run in parallel $\ceil{\log (n)}$ independent instances of an algorithm, each of which uses a different estimate $d = 2^\ell$ for the minimum degree $\delta(G)$, with $\ell = 0,1,2,\dots,\ceil{\log(n)}-1$.
Each algorithm aborts if it uses more than $\tOh(n)$ memory, and we will show that if $d = 2^\ell$ is such that $d \leq \delta(G) < 2 d$ then with high probability the corresponding algorithm will not abort and will correctly output $\mincut(G)$.
Since we know $\delta(G)$ exactly by the end of the stream (we can keep track of all degrees with $\tOh(n)$ memory), we can filter out the correct outcome at the end of the algorithm.

In the remainder we describe the algorithm for the value $d$ satisfying $d \leq \delta(G) < 2 d$.
It will be clear that the algorithms for different estimates are independent, and hence can be run in parallel.
In a single pass, the algorithm will perform uniform star contraction on $G$.
Simultaneously, it will construct a sparse $2d$-edge connectivity certificate on the contracted graph and compute the edge connectivity of this certificate.
Finally, in order to boost the constant success probability of uniform star contraction, we will run $r \in \Theta(\log n)$ parallel repetitions of this.
This last step will require some care, as every instance uses the same randomness from the input stream, and this needs to be done appropriately
to ensure independence.

The key idea to simulate a single implementation of star contraction in the random vertex arrival model is the following: because the vertices arrive in a random order, we can select the first $\Theta(n \log(n)/d)$ vertices as the set of centers $R$, and put all other vertices in $V \setminus R$. Because each vertex $v \in V\setminus R$ comes with all edges incident to $R$, we can for each such vertex $v$ sample a uniform and independent neighbor in $R$ in a single pass, thereby performing uniform star contraction.
Dealing with $r$ parallel repetitions requires a slightly more complicated approach, as we need to ensure independence and hence cannot reuse the same set of centers.
Nonetheless, we can still assume that all sampled centers come at the start of the stream.
This is captured by the following lemma, whose proof we postpone to \cref{app:parallel-sampling}.

\begin{restatable}{lemma}{sampling}
\label{lem:sampling}
There is a sampling procedure that operates within $\tOh(n)$ space and, given a stream $S$ of $n$ vertices, outputs
$Y_{1},\ldots,Y_{r}\subseteq [n]$ after reading the first $|\bigcup_{i=1}^{r}Y_{i}|$ vertices.
The distribution $D$ on $(Y_{1},\ldots,Y_{r})$ defined by the procedure admits the following property. For every
$R_{1},\ldots,R_{r}\subseteq [n]$:
\[  \Pr_{S\sim S_{n}, (Y_{1},\dots,Y_{r})\sim D}[Y_{1}=R_{1}, \dots, Y_{r}=R_{r}] = \prod_{i=1}^{r}\Pr_{X_{i}\sim B([n],p)}[X_{i} = R_{i}]
\enspace .\]
\end{restatable}

In each independent repetition of star contraction, we want to sample a subset $R$ by choosing every vertex with probability
$p=\frac{1200 \ln n}{d}$ under the additional condition that $|R|\leq 2pn$. 
This is done by applying \cref{lem:sampling} to obtain the subsets $R_{1},\dots,R_{r}$, where $R_{i}$ is the subset that should be used
in the $\ith$ repetition.
Let $B=\bigcup_{i=1}^{r}R_{i}$.
By the properties of the sampling procedure, the subsets are generated after reading the first $|B|$ vertices
from the stream.
Then, the $\ith$ repetition uses $R_{i}$ as the set of centers and
performs uniform star contraction for the \emph{subsequent} vertices in the input stream (i.e., those that
follow the first $|B|$ vertices).
For each of these, we choose a uniformly random edge towards $R_i$ and contract it in $G'_{i}$. All vertices in $B$, and
the subsequent vertices in the input stream with no edge towards $R_{i}$, are simply kept
in the contracted graph $G'_{i}$. We will argue that the number of vertices in $G'_{i}$ is still
$\tOh(n/d)$ with high probability, even though (as opposed to the original uniform
star contraction) we never contract any of the vertices of $B$ in $G'_{i}$.
In parallel, we build a $2d$-edge connectivity certificate $F^i_1 \cup \cdots \cup F^i_{2d}$ of $G_i'$, where the $F^i_k$'s are
edge disjoint spanning forests of $G_i'$ (as in the Nagamochi-Ibaraki certificate, \Cref{lem:NI}).
As an invariant throughout the stream, we will have that $F^i_1 \cup \cdots \cup F^i_{2d}$ is a $2d$-edge connectivity certificate of
the (potentially contracted) subgraph $G_i'$ seen so far.
Since this subgraph will only have $\tOh(n/d)$ vertices throughout the stream, the certificate will only contain $\tOh(n)$ edges.

We summarize the $\ith$ parallel repetition in full detail. The set $R_{i}$ can be accessed after reading the first
$|B|$ vertices from the input stream, and it is computed together with all sets $R_{1},\dots,R_{r}$
using \cref{lem:sampling} globally, outside of the $\ith$ repetition itself.
Recall that we abort each repetition as soon as its memory usage exceeds $\tOh(n)$,
in which case we set the $\ith$ outcome to be $\lambda_i = \infty$.
\begin{enumerate}
\item
Initialize $F^i_1,\dots,F^i_{2d}$ as empty forests and set $R_i = \emptyset$.
Initialize a mapping $r_i:V\to V$ to be the identity (through the stream this will keep track of the contracted vertices).
\item
For the $\jth$ vertex arrival $v$ with edges $e_1,...,e_\ell$ between $v$ to previously seen vertices, the following is done:
 \begin{enumerate}
    \item \textbf{Uniform star contraction:}
    If $j \leq |B|$, do nothing.
    If $j > |B|$, pick a uniformly random center $w$ from the center neighborhood $N_{R_i}(v)$ (if it exists) and set $r_i(v) = w$. This amounts to contracting the edge $\{v,w\}$. For each $e_t$ among $e_1,\dots,e_\ell$, except for the contracted edge which is discarded, change the endpoints of $e_t=\{v,u\}$ to be $\{r_i(v),r_i(u)\}$, discard any self loops. At the end of the stream, the vertices with the same $r_i(\cdot)$ values constitute a vertex in $G_i'$.
    \item \textbf{Maintaining of $2d$-edge connectivity certificate:}
    For each (relabelled) incident edge $e_{t}$ among $e_1,\dots,e_{\ell}$, add $e_t$ to $F^i_k$ where $k$ is the minimal index for which $F^i_k\cup \set{e_t}$ contains no cycles. If there is no such $k$, discard the edge.
 \end{enumerate}
\item
If the repetition did not abort by the end of the stream, we compute the edge connectivity of the connectivity certificate $\lambda(G'_i)$ and set $\lambda_i = \lambda(G'_i)$. Note that $\lambda_i \geq \lambda(G)$.
\end{enumerate}
Finally, we combine the $r$ parallel repetitions by outputting $\min\{\delta(G),\lambda_1,\dots,\lambda_r\}$.

\paragraph{Analysis.}
As mentioned before, it suffices to prove correctness and a $\tOh(n)$ memory bound only for the algorithm that has an
estimate $d$ such that $d \leq \delta(G) < 2d$. By \cref{lem:sampling}, a run of the whole algorithm for such an estimate
is equivalent to $r\in\Theta(\log n)$ independent repetitions of a variant of the uniform star contraction with $p=\frac{1200 \ln n}{d}$,
except that vertices from $B$ remain in the contracted graph.

A single repetition can be analysed as follows. Sample each vertex with probability $p=\frac{1200 \ln n}{d}$
to obtain the set of centers $R$.
Then, for some set of vertices $B$ such that $R\subseteq B$ and $|B|\leq 2rpn$, proceed as follows.
For any vertex in $[n]\setminus B$, choose a uniform random edge towards $R$ (if it exists) and contract it.
This results in a contracted graph $G'$.
By item~2 of \cref{prop:sizeR} and $d\leq d(v)$,
the probability that a vertex in $[n]\setminus B$ has no neighbor in $R$ is at most
$n^{-6}$. Hence, by a union bound, $|G'| = |B|$ except with probability at most $n^{-5}$ over the choice of $R$.
Next, we want to lower bound the probability that $\lambda(G)=\lambda(G')$,
assuming that $\lambda(G')<\delta(G)$.
Let $C$ be a non-trivial minimum cut of $G$. 
By \cref{lem:good_for_contracting}, $H = (V, \cutd(V \setminus R)$ is $(2/3,8)$-good for contracting with respect to $C$ with 
probability at least $2/3$ over the choice of $R$.
We perform a random 1-out contraction on a subgraph of $H$,
so the probability of not contracting an edge of $C$ is at least as large as when performing a random
1-out contraction on the whole $H$, which is at least $3^{-12}$ by \cref{cor:sample_cor} over the choice of $R$.
Thus $\lambda(G')=\lambda(G)$ with probability at least $2/3\cdot 3^{-12}$ over the choice of $R$.

By union bound, the overall probability of error is at most the sum of probabilities that some repetition aborts due to
using too much memory and the probability that $\lambda(G_{i}')>\lambda(G)$ holds for every $i=1,\ldots,r$,
both over the choice of $R_{1},\ldots,R_{r}$.
By item~1 of \cref{prop:sizeR},  $|R_{i}| \leq 2rpn$ except with probability $n^{-400}$,
thus $|B|\leq 2rpn$ except with probability $r\cdot n^{-400}$.
For each $i=1,\ldots,r$ we have  $|G_{i}'| = |B|$ except with probability at most $n^{-5}$.
Overall, no repetition aborts except with probability $r\cdot n^{-400}+r\cdot n^{-5}=O(n^{-4})$.
For each $i=1,\ldots,r$ independently, we have $\lambda(G_{i}')=\lambda(G)$ with probability at least $2/3\cdot 3^{-12}$ over
the choice of $R_{i}$. Thus, at least one repetition correctly determines $\lambda(G)$ except
with probability at most $(1-2/3\cdot 3^{-12})^{\Theta(\log n)}=n^{-\Omega(1)}$.
\end{proof}

We can prove a similar result for the \textit{complete vertex arrival} setting, in which the vertices arrive in an arbitrary order with all edges incident on them. 
The proof follows along the same lines, but
is simpler in this case because we can randomly sample sets 
of centers offline before the stream begins. 
As a vertex arrives with all its edges, we 
can immediately randomly choose an edge 
incident on the set of centers to implement 
star contraction.  In parallel we also 
construct a sparse edge connectivity certificate.  Due to the similarities 
with \cref{thrm:random_vertex_arrival}, 
we move the proof to \cref{app:vertex-arrival}.

\begin{restatable}{theorem}{complete}
\label{thrm:complete_vertex_arrival}
There is a one-pass streaming algorithm, using $\tOh(n)$ memory, that given a simple graph $G=(V,E)$ in the complete vertex arrival setting, computes the edge connectivity of $G$ with high probability.
\end{restatable}

\begin{remark} \label{remark:2-out-stream}
While we cannot rule it out, it is not obvious how to prove either \cref{thrm:random_vertex_arrival} or \cref{thrm:complete_vertex_arrival} using 2-out contraction \cite{GNT20}.  We crucially use two 
features of star contraction in the proofs.  The first is that when we see a vertex in the stream we can immediately choose its edges to be contracted, which allows us to contract the graph 
``on-the-go.''  In the random vertex arrival model 
we cannot randomly choose 2 edges incident to a vertex when the vertex arrives, as at that point we have not seen all of its neighbors.  To naively implement 2-out this means contracting edges 
has to be delayed until the end of the stream which prohibits constructing the sparse edge connectivity certificate with $\tOh(n)$ memory.

The difficulty of using 2-out contraction to obtain the complete vertex arrival result is more subtle.  Here we do see all incident edges to a vertex when the vertex arrives, thus we can immediately 
select two random incident edges to contract.  However, in this case it is not clear that the (partly) contracted graph throughout the stream has no more than $\tOh(n/\delta(G))$ components.  This 
property is needed to ensure that we can keep a sparse $\delta(G)$-edge connectivity certificate of the contracted graph throughout the stream using a memory of size only $\tOh(n)$.  
For star contraction, this issue is resolved by a second key feature of star contraction that we make use of in our proofs. At any point in the stream, the contracted graph on the vertices seen so far has size at most the number of 
centers, and thus its size can be bounded by $\tOh(n/\delta(G))$ in the branch of the computation with the correct degree estimate.
\end{remark}

\subsection{Linear time sequential algorithm for slightly dense graphs} \label{sec:sequential}
In this section, we give another illustration on how to use the star contraction algorithm.
Similar to the approach in the state-of-the-art $O(m+n\log^2 n)$ algorithm \cite{GNT20} for computing the edge connectivity of a simple graph, we use it to obtain a simple $O(m+n \,\mathrm{polylog}(n))$ sequential algorithm. Formally, we prove the following theorem. 

\begin{theorem}\label{thrm:sequential}
There is a randomized algorithm with running time $O(m+n\log^3 n)$ that computes the edge connectivity of a simple graph $G=(V,E)$ with success probability at least $2/3$.
\end{theorem}

The proof of this theorem is achieved via an application of uniform star contraction as in \cref{thm:intro:simple star contraction}.
The naive implementation requires $O(m)$ time by sampling vertices independently into $R$ in $O(n)$ time, and then choosing a random neighbor in $R$ for each $v\not\in R$ in $O(\sum\limits_{v\in V} d(v))=O(m)$ time.
We combine this with the connectivity certificate algorithm from \cref{lem:NI}, whose implementation in the sequential setting also requires $O(m)$ time \cite{NI92}, and the following recent result.

\begin{theorem}[\cite{GawrychowskiMW20}]\label{thrm:GMW_min_cut}
There is an algorithm that with high probability computes the weight of a minimum cut in a weighted graph $G=(V,E,w)$ with $n$ vertices and $m$ edges in time $O(m\log ^2 n)$.
\end{theorem}

\begin{proof}[Proof of \cref{thrm:sequential}]
The algorithm proceeds as follows: Given $G=(V,E)$, compute $\delta(G)$ in $O(m)$ time. 
If $\lambda(G) \geq \delta(G)$ then this corresponds to the edge connectivity.
Now assume $\lambda(G) < \delta(G)$.
We perform uniform star contraction as in \cref{thm:intro:simple star contraction} on $G$ to obtain $G'$ with $O(\frac{n\log n}{\delta(G)})$ vertices, and with constant probability we have that $\mincut(G)=\mincut(G')$. 
Construct in time $O(m)$ a sparse $\delta(G)$-edge connectivity certificate (\cref{lem:NI}) $G''$ of $G'$ with $O(n\log n)$ edges. Finally, apply \cref{thrm:GMW_min_cut} on $G''$ to compute $\lambda(G'') = \lambda(G')$ with witness $S \subseteq V$. Since $|E(G'')|=O(n\log n)$ this takes $O(n\log ^3 n)$ time. If we output $\min\{\lambda(G'),\delta(G)\}$ then this yields the correct output with constant probability.
We can boost the success probability to above $2/3$ by repeating this a constant number of times, and outputting the smallest value achieved.
\end{proof}

\section{Finding a spanning forest with $O(n)$ cut queries} \label{sec:spanning forest}

In this section we describe an $O(n)$ cut query algorithm for constructing a spanning forest of a simple graph. This proves that the cut query complexity of graph connectivity is $O(n)$, which was not known before. We also describe a variation of the algorithm for constructing a connectivity certificate, which is a key building block of the edge connectivity algorithm described in the next section.

Let us first describe a simple algorithm to find a spanning forest of a graph using $O(n \log n)$ cut queries given by Harvey \cite[Theorem 5.10]{HarveyThesis}.  For the application 
to finding a sparse edge connectivity certificate it will be useful 
to define the algorithm more generally to work on a contraction of a graph.  
\begin{lemma}[Simple spanning forest algorithm]
\label{lem:simple}
Let $G = (V,E)$ be an $n$-vertex simple graph.
Let $G'$ be a contraction of $G$ with $q$ many supervertices, which are given explicitly as the partition $\Pcal = \{A_1,\ldots, A_t\}$ of $V$.  
There is a deterministic algorithm that outputs a set of edges $F \subseteq E$ that form a spanning forest of $G'$ and makes $O(q \log n)$ cut queries to $G$.
\end{lemma}

\begin{proof}
We follow the plan of Prim's spanning forest algorithm.
We begin at an arbitrary supervertex $A$ of $G'$ and initialize $F = \emptyset$.
We want to find an edge $\{u,v\} \in E$ with $u \in A$ and $v \notin A$.
We first identify a vertex $u \in V$ which has an edge leaving $A$ by doing binary search with queries of the form $|E(S,\bar{A})|$ for $S \subseteq A$.  This takes at most $\log n$ many cut queries.  Then we want to find one of the neighbors $v \in \bar{A}$ of $u$ by doing binary search asking queries of the form $|E(v,S)|$ for $S \subseteq \bar{A}$.  We add $\{u,v\}$ to $F$ and let $G''$ be the graph $G'$ with supervertex $A$ merged with the supervertex containing $v$.  Let $A'$ be the 
name of this new supervertex.  We then repeat this procedure on $G''$ and $A'$.  We keep repeating this procedure until we find a supervertex with no outgoing edges. This supervertex represents a connected component in $G'$.  If the supervertex is not all of $V$ then we arbitrarily choose another supervertex of the graph and begin the procedure again.  There are at most $q$ iterations and each iteration 
costs $O(\log n)$ cut queries, thus the total number of cut queries is $O(q \log n)$.  
\end{proof}

An important insight now is that in the proof of \cref{lem:simple} we did not use the full power of cut queries.  For the binary search routines in the proof we might as well have used \emph{bipartite independent set queries}, which return just a single bit telling if $|E(S,T)|$ is zero or positive for two disjoint sets $S$ and $T$. Notice that the (deterministic) algorithm is optimal for that type of queries. Indeed, by the $\Omega(n \log n)$ deterministic communication complexity lower bound for connectivity \cite{HajnalMT88}, any deterministic algorithm for connectivity must make $\Omega(n \log n)$ such 1-bit queries.

Crucially, the situation is different with cut queries, which return $\Omega(\log n)$ bits of information per query in a simple graph. Taking 
advantage of this additional information suggests an avenue towards saving 
a $\log n$ factor over the simple algorithm, and this is 
the approach we follow to give a zero-error randomized algorithm to compute a spanning forest of a simple graph with $O(n)$ cut queries.  

\subsection{Separating matrices and learning buckets}
The fact that a cut query returns $\Omega(\log n)$ bits allows us to use the remarkable result that there is a matrix $A \in \{0,1\}^{k \times n}$ with $k = O(n/\log n)$ such that one can recover any Boolean vector $x \in \{0,1\}^n$ 
given the product $Ax$.  At a very high level, this $\log n$ factor reduction in the size of $k$ over the obvious bound is the key that allows us to save a $\log n$ factor in the 
spanning forest computation.
Such a matrix $A$ is called a separating matrix and formally defined next.

\begin{definition}[Separating matrix]
A $k$-by-$n$ Boolean matrix $B$ is called a separating matrix for the set $S \subseteq \{0,1,\ldots, d\}^n$ if for all $x,y \in S$ with $x \ne y$ it holds that $Bx \ne By$.  
\end{definition}

\begin{theorem}[{\cite[Theorem A.1]{GK98}, \cite[Theorem 1]{GK00}}]
\label{thm:separating}
There exists a $k$-by-$n$ separating matrix for 
\begin{enumerate}
	\item The set $\{0,1,\ldots, d\}^n$ with $k \le 8\ceil{\log(d+1)}n/\log(2n)$.
	\item  The set $S_\ell = \{x \in \{0,1\}^n : |x| \le \ell\}$ with $k = O(\ell \log(2n)/\log(2\ell))$.
\end{enumerate}
\end{theorem}
Grebinski and Kucherov \cite[Theorem 5]{GK00} use separating matrices to show that if $A$ is the adjacency matrix of a simple $n$-vertex graph $G$ with maximum degree $\ell$, then one can learn $G$ with $O(\ell n)$ queries 
of the form $x^TAy$ for Boolean vectors $x,y \in \{0,1\}^n$.  We will use separating matrices for a very similar application, although we focus on bipartite graphs where the left hand side has bounded degree $\ell$, and want to express the 
complexity in terms of the number of vertices on the left hand side.  To this end, let $M$ be an $m$-by-$n$ Boolean matrix and suppose that every row has at most $\ell$ ones.  Suppose that we have $x^TMy$ query access to 
$M$ for Boolean vectors $x,y$.  One should imagine $M$ being a bipartite adjacency matrix of a bipartite subgraph of $G$, in which case we can simulate $x^TMy$ with cut queries to $G$, and think of
$\ell$ as constant or slowly growing in a typical application.  The next lemma shows that if $m$ and $n$ are polynomially related then we can learn $M$ with $O(\ell m)$ queries.  A very similar theorem is shown by 
Grebinski and Kucherov \cite[Theorem 4]{GK00} for the case $m=n$.

\begin{lemma}
\label{lem:learn}
Let $M \in \{0,1\}^{m \times n}$ be an $m$-by-$n$ matrix with at most $\ell$ non-zero entries per row.  There is a deterministic algorithm that learns $M$ with $O\left(\frac{\ell m \log(2n)}{\log(2m)} \right)$ many queries of the form $x^T M y$ with Boolean vectors $x,y$.
\end{lemma}

\begin{proof}
Let $Y$ be a separating matrix with $O(\ell \log(2n)/\log(2\ell))$ rows for the set $S_\ell \subseteq \{0,1\}^n$ of Boolean vectors with at most $\ell$ non-zero entries, which exists by item~2 of \cref{thm:separating}.  
From $M Y^T$ we can recover $M$.  
Every column of $M Y^T$ has integer entries of magnitude at most $\ell$.  Let $X$ be a separating 
matrix with $8m \ceil{\log(\ell+1)}/\log(2m) $ rows for the set $\{0,1,\ldots, \ell\}^m$ which exists by item~1 of \cref{thm:separating}.  
We can recover $M Y^T$ from $X M Y^T$.  Putting it together we can compute $X M Y^T$ with a number of $x^T M y$ queries of order
\[
\frac{\ell m \log(2n) \ceil{\log(\ell+1)}}{\log(2m)\log(2\ell)} \le \frac{\ell m \log(2n)}{\log(2m)} \enspace . \qedhere
\]
\end{proof} 

Using this lemma we can describe a key subroutine LearnBucket$[M](r,k)$ that will be used both in the spanning forest algorithm and in the edge connectivity algorithm.  In this algorithm we have 
oracle access to a matrix $M$ via $x^T M y$ queries for Boolean vectors $x$ and $y$.  The oracle access to $M$ is indicated by having $M$ in brackets in the call of the algorithm.  The promise 
is that every row of $M$ has between $r$ and $2r$ many ones---in our applications this arises from ``bucketing'' together vertices with similar degrees, hence the name.  The number $k$ is a parameter 
indicating how many ones we want to learn from each row of $M$---the algorithm will learn $\min\{k,r\}$ many ones from each row.  In the application to finding a spanning forest we just need to find a single one 
in every row and we take $k$ to be a large constant.
In the edge connectivity algorithm we will also apply LearnBucket where $k$ is growing.
For the application to edge connectivity, we will require that the found neighbors are selected in a sufficiently random fashion, as recorded in the ``further'' statement of the theorem.  This statement 
is not needed for the application to finding a spanning forest and can be skipped on a first reading.

\begin{algorithm}[!htbp]
\caption{LearnBucket$[M](r,k)$}
\label{alg:learnBucket}
 \hspace*{\algorithmicindent} \textbf{Input:} $x^T M y$ query access to a Boolean matrix $M \in \{0,1\}^{m \times n}$, a natural number $r$ with the promise that 
 all rows of $M$ have at least $r$ and at most $2r$ ones, and a parameter $k$. \\
 \hspace*{\algorithmicindent} \textbf{Output:} The output consists of a list $Z[i]$ for each $i \in [m]$ where $M(i,Z[i][j]) = 1$ for all $i,j$ and each $Z[i]$ has at least 
 $\min\{k,r\}$ many elements.
\begin{algorithmic}[1]
\State $B = [m]$
\While{$B$ is non-empty}
	\State  Choose $Q \subseteq [n] $ by putting each $a \in [n]$ into $Q$ independently with probability $q=\min\{\frac{2k}{r},1\}$.\label{line:chooseQ}
	\For{$j \in B$} 
		\State $\ones(j) \gets \chi_{\{j\}}^T M \chi_{Q}$. \Comment{$\ones(j)$ is number of ones in $M(j,Q)$}
		 \label{line:count} 
	\EndFor
	\State Set $K\gets \{ j \in B : \min\{r,k\} \le \ones(j) \le 8k\}$. \label{line:K}
	\If{$|K| > 0$}
	    \State Learn the submatrix $M(K,Q)$ by \cref{lem:learn} and populate $Z[i]$ for all $i \in K$.
    \label{line:learn}
    	\State $B \gets B \setminus K$.
    \EndIf
\EndWhile
\State Return all lists $Z[i]$.
\end{algorithmic}
\end{algorithm}

\begin{lemma}
\label{lem:learnBucket}
Let $m,n$ be positive integers and $M \in \{0,1\}^{m \times n}$ be a Boolean matrix where every row has at least $r$ and at most $2r$ ones.  Let $k \ge 10$ and 
$\ell = \min\{r,k\}$.  
Suppose we can query $x^T M y$ for any $x \in \{0,1\}^m, y \in \{0,1\}^n$.  
There is a zero-error randomized algorithm, LearnBucket$[M](r,k)$ given in \cref{alg:learnBucket}, that makes 
\[
O\left(m + \frac{km \log(n)}{\log(2m)} \right)
\] 
queries in expectation and for each $i \in [m]$ outputs a list $Z[i]$ such that 
$M(i,Z[i][j]) = 1$ for all $i \in [m]$ and $j$, and each $Z[i]$ contains at least $\ell$ many elements.  Let $d(i)$ be the number of ones in row $i$.  Further, $Z[i]$ contains all 
the ones of $M(i,:)$ contained in a set $Q$ chosen by putting each $j \in [n]$ into $Q$ independently at random with probability $q \ge \min\{2k/d(i),1\}$,
conditioned on $M(i,Q)$ having at least $f$ and at most $g$ ones, where  $0 < f \le q d(i)/2$ and $g \ge 2 q d(i)$.
\end{lemma}

\begin{proof}
The algorithm is given by \cref{alg:learnBucket}.
We first consider the (trivial) case where $2k \ge r$, in which case the sampling probability $q =1$.  In this case, in the first iteration of the while loop 
$K = [m]$ and we learn the entire matrix $M$ deterministically via \cref{lem:learn}.  The cost of this is $O(2rm \log(n)/\log(2m)) = O(k m \log(n)/\log(2m))$ 
and we learn at least $r$ ones in each row as desired. 

Now consider the case that $2k < r$, in which case the goal is to learn the positions of $k$ ones in every row of $M$.
We first show correctness.  At the start of the while loop $B = [m]$, and we only remove a row index
$j$ from $B$ if it is in the set $K$ processed in line~\ref{line:learn}.  On this 
line we deterministically learn the 
entire submatrix $M(K,Q)$.  Further we are guaranteed that each row of $M(K,Q)$ has at least $k$ ones by 
the definition of $K$ on line~\ref{line:K}.
Thus when $j$ is removed from $B$ we are guaranteed that we have learned the positions of at least $k$ ones in row $j$.  This process continues 
until $B$ is empty, thus the algorithm is correct with zero error.

Let us now argue about the complexity.  In any iteration of the while loop, for each $j \in B$ the expected value of $\ones(j)$ is in the 
interval $[2k, 4k]$.  
Therefore by a Chernoff bound (\cref{lem:chernoff}) using the fact that $k \ge 10$, for any $j \in B$ we have 
$k \le \ones(j) \le 8k$ with probability at least $1/2$. Thus $\bbE[|K|] \ge |B|/2$, and letting $b_i$ be a random variable for the size of $B$ 
at the start of the $\ith$ iteration of the while loop
we have $\bbE[b_{i+1} \mid b_i = s] \le s/2$.  From this it follows that the expected number of iterations of the while loop is at most 
$2(\log(m) + 1)$ (see \cite[Theorem 3]{DoerrJW12}).  
The fact that $\bbE[b_{i+1} \mid b_i = s] \le s/2$ also implies $\bbE[b_{i+1}] \le \bbE[b_i]/2$ and so the expected number of queries from 
line~\ref{line:count} is 
\[
\bbE\left[ \sum_i b_i \right]
= \sum_i \bbE[b_i] \le \sum_i \frac{m}{2^i} \le  2m \enspace .
\]

We next turn to queries made in line~\ref{line:learn}.  By \cref{lem:learn} the number of queries in an execution of this line is $O(k|K| \log(|Q|)/\log(2|K|)) \in O(k|K| \log(n)/\log(2|K|))$.
Let $Y$ be a random variable for the 
number of times line~\ref{line:learn} 
is executed.  As the sum of $|K|$ over all executions of line~\ref{line:learn} is at most $m$, we must have 
$Y \le m$. Further $\bbE[Y] \le 2\log m+1$ as it is at most the total expected number of iterations of the while loop. By \cref{clm:concave}, when 
$Y = t$ the overall number queries from line~\ref{line:learn} is $O(km \log(n)/\log(2m/t))$.  Thus the expected number of queries from line~\ref{line:learn}
overall is of order
\begin{equation}
\label{eq:learn_sum}
\sum_{t=1}^{m} \frac{km \log(n)}{\log(2m/t)} \Pr[Y=t] 
\le \sum_{t=1}^{\sqrt{m}} \frac{2km \log(n)}{\log(2m)} \Pr[Y=t] + km \log(n) \Pr[Y > \sqrt{m}] \enspace .
\end{equation}
By Markov's inequality
$\Pr[Y > \sqrt{m}] \le 2 (\log(m) + 1)/\sqrt{m}$.
Hence we can upper bound \cref{eq:learn_sum} by
\[
\frac{2 k m \log(n)}{\log(2m)} + \frac{2(\log(m) +1)}{\sqrt{m}} km \log(n)
\in O\left(\frac{km\log(n)}{\log(2m)}\right). 
\]

The ``further'' statement follows from the definition of $Q$ on line~\ref{line:chooseQ}.  This set is taken by putting each $i \in [n]$ into $Q$ independently at random 
with probability $q = \min\{2k/r,1\} \ge \min\{2k/d(i),1\}$ as $d(i) \ge r$.  $Z[i]$ contains exactly the ones in $M(i,Q)$ for the first 
such $Q$ chosen that has at least $f=\min\{k,r\} \le q d(i)/2$ and at most $g=8k \ge 2q d(i)$ ones by line~\ref{line:learn}.
\end{proof}

LearnBucket is restricted in that it requires the matrix to have approximately uniform row sums.
In the next algorithm, Recover-$k$-From-All, we use LearnBucket as a subroutine to locate the position of $\min\{k,d\}$ ones in each row of a matrix with the weaker promise that every row has at least 
$d$ ones.  With respect to its goal, Recover-$k$-From-All is very similar to the primitive RecoverOneFromAll introduced by \cite{AuzaL21} in the study of connectivity algorithms with matrix-vector 
multiplication queries.  RecoverOneFromAll was in turn inspired by the Recover primitive used by \cite{AssadiCK21} for connectivity algorithms with linear and OR queries.

We follow the same algorithmic plan used by \cite{AuzaL21} in RecoverOneFromAll, which is 
to count the number of ones in each row, bucket rows together with similar number of ones, and 
then operate on each bucket 
separately. The main difference is in the implementation of 
learning the position of ones for each row in a bucket.  We use separating matrices for this in LearnBucket while the technique in \cite{AuzaL21} is based on combinatorial group testing algorithms.

\begin{algorithm}[!htbp]
\caption{Recover-$k$-From-All$[M](k)$}
\label{alg:recover}
 \hspace*{\algorithmicindent} \textbf{Input:} $x^T M y$ query access to a Boolean matrix $M \in \{0,1\}^{m \times n}$ and a parameter $k$. \\
 \hspace*{\algorithmicindent} \textbf{Output:} Let $d$ be the minimum number of ones in a row of $m$.  For $\ell = \min\{d,k\}$ the output is a list $Z[i]$ for each $i \in [m]$ 
 such that $Z[i]$ has at least $\ell$ elements and $M(i,Z[i][j]) = 1$ for all $i \in [m]$ and $j$.
\begin{algorithmic}[1]
\For{$j \in [m]$}
	\State $d(j) \gets \chi_{\{j\}}^T M \one.$  \Comment{$d(j)$ is the number of ones in row $j$} \label{line:degree}
	\label{line:count_deg}
\EndFor
\State $d \gets \min_j d(j)$.
\For{$a = 0$ to $\ceil{\log(n/d)}$}
	\State $B_a \gets \{ j \in [m] : d(j) \in [d2^a, d2^{a+1})\}.$ \label{line:bucket}
	\State $Z_a \gets$ LearnBucket$[M(B_a,:)](d2^a,k)$. \label{line:call_learn_bucket}
\EndFor
\State Output all adjacency lists $Z_a$.
\end{algorithmic}
\end{algorithm}

\begin{lemma}
\label{cor:learn}
Let $M \in \{0,1\}^{m \times n}$ be a Boolean matrix where every row has at least $d > 0$ ones.  Let $k \ge 10$ and $\ell = \min\{k,d\}$.  
Suppose we can query $x^T M y$ for any $x \in \{0,1\}^m, y \in \{0,1\}^n$.  
There is a zero-error randomized algorithm that outputs a list $Z[i]$ with at least $\ell$ elements for each $i \in [m]$ satisfying $M(i,Z[i][j])=1$ for all $i \in [m]$ and $j$, and makes
\[
O \left( m +  \frac{km\log(n)}{\log(2m/\log(n))} \right)
\]
queries in expectation.  Let $d(i)$ be the number of ones in row $i$.  Further, $Z[i]$ contains all 
the ones of $M(i,:)$ contained in a set $Q$ chosen by putting each $j \in [n]$ into $Q$ independently at random with probability $q \ge \min\{2k/d(i),1\}$,
conditioned on $M(i,Q)$ having at least $f$ and at most $g$ ones, where  $0 < f \le q d(i)/2$ and $g \ge 2 q d(i)$.
\end{lemma}

\begin{proof}
The algorithm is given in \cref{alg:recover}.  In line~\ref{line:count_deg} we compute the number of ones in each row of the matrix and then bucket 
the vertices accordingly in line~\ref{line:bucket}.  Thus in the call to LearnBucket$(d2^i,k)$ for those rows in $B_i$ the promise that each row 
has number of ones in $[d2^i, d2^{i+1})$ will hold and LearnBucket will return the positions of $\min\{d2^i,k\} \ge \min\{d,k\}$ ones for each row of $B_i$ by \cref{lem:learnBucket}.  
This shows correctness.

Let us now examine the complexity.  There are $m$ queries made in line~\ref{line:degree}.  The rest of the queries are made in the for loop.
In the execution of the for loop on $B_a$ we make $O(|B_a| + k |B_a| \log(n)/\log(2|B_a|))$ queries in expectation by \cref{lem:learnBucket}.  Thus by \cref{clm:concave} the 
total number of queries is at most $O \left (m + km \frac{\log(n)}{\log(2m/\log(n))} \right)$.

The ``further'' statement follows immediately from the ``further'' statement of \cref{lem:learnBucket}.
\end{proof}

\subsection{Spanning forest algorithm}
Now we are ready to describe a zero-error randomized algorithm to compute a spanning forest of a simple $n$-vertex graph with $O(n)$ cut queries in expectation.  Compared to the simple 
spanning forest algorithm, the first high level idea is to switch from a Prim style spanning forest algorithm to one based on Bor\r{u}vka's algorithm, which is known to work well in parallel settings.  Here the basic task it to find an outgoing edge from each of the
connected sets $S_1, \ldots, S_t$.  The second idea is to use \cref{cor:learn} to do this in parallel and save a $\log n$ factor compared to the naive sequential computation.
\spanning*

\begin{proof}
We will follow Bor\r{u}vka's spanning forest algorithm.  The algorithm proceeds in rounds and maintains the 
invariant that in each round there is a paritition $S_1, \ldots, S_t$ of $V$ and a spanning tree for 
each $S_i$ in the partition.  Initially, each $S_i$ is just a single vertex.  

In a generic round the goal is to find an outgoing edge from each $S_i$ that is not already a connected component.  To help with this, 
we will label every vertex $v \in V$ as $\act$ or $\inact$.  Initially, all vertices are marked $\act$.  If in any round we learn that $v \in S_i$ has no edge 
going outside of $S_i$ then we mark it as $\inact$.  An inactive vertex is not useful to the algorithm because it will not have an edge leaving its component in any 
future round of the algorithm.  We will similarly call a set $S_i$ $\act$ if and only if it contains an active vertex, and $\inact$ otherwise.  A set that is $\inact$ is a connected component.

Once we have found an outgoing edge from each $S_i$ that has one, we select a subset of these edges that is cycle free with respect to the partition $S_1, \ldots, S_t$.
These edges are used to merge the corresponding sets of the partition and update the spanning trees accordingly.
If $t'$ sets among $S_1, \ldots, S_t$ are $\act$, then the cycle free subset of edges will have size at least $t'/2$, and every edge added reduces the number of active sets by at least $1$.  It follows that the number of $\act$ sets decreases by a factor of at least two in each round.  We will crucially use this geometric decrease in our analysis of the algorithm.  Once the number of $\act$ sets falls below 
$n/\log(n)$ we switch to the simple spanning forest algorithm from \cref{lem:simple} to finish finding a spanning forest.

We now formally describe the actions of the algorithm in a generic round where we have sets $S_1, \ldots, S_t$ and a spanning tree for each $S_i$.  There
are two main steps to a round.

\paragraph{Step 1: For each active $S_i$ find a $v \in S_i$ that has a neighbor outside of $S_i$.}  
Let $t$ be the number of sets that were $\act$ at the end of the previous round, and say without loss of generality these are the sets $S_1, \ldots, S_t$.
For each $i=1,\ldots,t$ we do the following. We query $|E(v,\bar{S}_i)|$ for each $\act$ vertex $v \in S_i$ until we find a vertex with $|E(v,\bar{S}_i)|>0$.
In such case we mark $v$ as the \emph{representative} of $S_i$ and move on to $S_{i+1}$ without any further queries in $S_i$.  For all vertices in $S_i$ with $|E(v,\bar{S}_i)|=0$ we mark $v$ as $\inact$. If all vertices in $S_i$ become $\inact$ then $S_i$ becomes $\inact$: it is a connected component and we do not need to process it in future rounds.  

As we only make queries in $S_i$ until we find an $\act$ vertex, the total number of queries in a round is $O(t + w)$, where $w$ is the number of vertices that become $\inact$ in the round.  These vertices will never be queried again, so the term for $\inact$ vertices will only contribute $O(n)$ queries over all 
the rounds.

\paragraph{Step 2: Learn an outgoing edge from a constant fraction of the representatives.}
Let $t'$ be the number of $\act$ sets after Step 1 (so we have already determined that $S_{t'+1}, \ldots, S_t$ are connected components). Equivalently, $t'$ denotes the number of representatives found in 
step~1, and say without loss of generality these are from the sets $S_1, \ldots, S_{t'}$.

For $i=1, \ldots, t'$ we color each $S_i$ red or blue independently at random with equal probability.  All vertices in $S_i$ are given the color of $S_i$.
For each red representative we then count how many neighbors it has colored blue.  This can be done with $O(t')$ cut queries.  Let $W$ be the set of red representatives 
that have a blue neighbor, and consider the submatrix $M$ of the adjacency matrix with rows labeled by elements of $W$ and columns labeled by vertices colored blue.  
By \cref{cor:learn} with $k = 10$, for every element of $W$ we can learn the name of a blue neighbor with $O(t' + t' \log(n)/\log(2t'/\log n))$ cut queries.

Now let us compute the expected number of components at the end of the round.
In expectation, $1/2$ of the representatives labeled red will have a neighbor colored blue.  We learn one edge crossing the red-blue cut 
from each red representative.  This set of edges is necessarily cycle free with respect to $S_1, \ldots, S_{t'}$.  Thus by this process 
in expectation we will find a cycle free set of edges of size at least $t'/4$.  As any cycle free set of edges is of size at most $t'$, this means that by a reverse Markov inequality we will find a cycle free set of edges of size at least $t'/8$ with probability at least $1/7$.  Hence with probability at least $1/7$ the number of $\act$ sets at the start of the next round is at most $7t'/8$.

\paragraph{Total number of queries.}
When there are $t$ active sets remaining at the end of the previous round, then we have seen that the number of queries made in the current round is $O(t + t \log(n)/\log(2t/\log(n))) \in O(t \log(n)/\log(2t/\log(n)))$, plus a term which is $O(n)$ over the course of the algorithm.
We have also argued that with probability at least $1/7$ the number of active components in the following round is at most $7t/8$.

Let $T(t)$ denote the expected number of queries made by the algorithm starting from when there are $t$ active sets remaining.  As $T(t)$ is monotonically increasing in $t$, we have that
\[
T(t) \le f(t) + \frac{6}{7} T(t) + \frac{1}{7} T(7t/8) \enspace,
\]
where $f(t) = O(t\log(n)/\log(2t/(\log n)) )$. Equivalently, $T(t)/7 \le f(t) + T(7t/8)/7$, and letting $c = 7/8$ we have that for any $j$ we can bound
\[
\frac{1}{7} T(n) \le f(n) + f(cn) + f(c^2 n) + \dots + f(c^j n) + \frac{1}{7} T(c^{j+1} n).
\]
Now notice that we switch to the algorithm in \cref{lem:simple} once the number of active sets falls below $n/\log(n)$.
In that case the remaining query complexity is $O(n)$ and hence $T(s) \in O(n)$ for~$s \leq n/\log(n)$.
So it remains to bound $f(n) + \dots + f(c^j n)$ for $c^j \geq 1/\log(n)$.
By the definition of $f(t)$ this is of order
\[
n \log(n) \left( \frac{1}{\log(2n/\log n)} + \frac{c}{\log(2 c n/\log n)} + \cdots + \frac{c^j}{\log(2c^j n/\log n)} \right).
\]
Using that $c^j \geq 1/\log(n)$, for $0 \leq i \leq j$ we can bound all denominators by
\[
\log(2 c^i n/\log n)
\geq \log(2n) - 2\log(\log(n))
\in \Omega(\log(n)).
\]
This gives the bound $f(n) + \dots + f(c^j n) \in O(n (1 + c + \dots + c^j)) \in O(n)$.
\end{proof}

\subsection{Edge connectivity certificate} \label{subsec:conn-certificate}
We can easily use the spanning forest algorithm from last section to construct a sparse $r$-edge connectivity certificate by following the Nagamochi-Ibaraki approach of packing spanning forests. This would require $O(nr)$ cut queries. In our edge connectivity algorithm, however, we will first do a star contraction on the input graph. This yields a contracted multigraph with significantly fewer vertices (say $q \ll n$), and we would like to construct an $r$-edge connectivity certificate with only $O(r q)$ queries. This would easily follow from modifying the spanning forest algorithm from the last section to find a spanning forest of a multigraphs with $q$ vertices using $O(q)$ cut queries. However, it is not clear whether this is possible.\footnote{In particular, the separating matrix machinery encounters additional logarithmic 
factors in working with the non-Boolean adjacency 
matrix of a weighted graph, which seem hard to avoid.}

In the following theorem we show that it is nevertheless possible to obtain a sparse $r$-edge connectivity certificate for a $q$-vertex contraction of an $n$-vertex simple graph efficiently, namely with $O(n + rq \log(n)/\log(q))$ cut queries. A key idea, as in the 
classic sequential algorithm of Nagamochi-Ibaraki, 
is to build the $r$ spanning forests in parallel.

\begin{theorem}[Formal version of \Cref{thm:intro-certificate1}]
\label{thm:intro-certificate}
Let $G = (V,E)$ be an $n$-vertex simple graph, and let $G' = (V',E')$ be a contraction of $G$ with $q$ supervertices for $q \ge \log^{2+\eps}(n)$ for some $\eps > 0$. There is a zero-error randomized algorithm that makes $O(n + r q \log(n)/\log(q))$ cut queries in expectation and outputs a sparse $r$-edge connectivity certificate for $G'$.
\end{theorem}

\begin{proof}
We will make use of \cref{lem:NI} and find $F_1, \ldots, F_r$ such that $F_i$ is a spanning forest of $(V', E' \setminus \bigcup_{j=1}^{i-1} F_j)$.
We will follow the algorithm from \cref{thm:spanning} to find these $r$ spanning forests in parallel.  

The algorithm proceeds in rounds.  We maintain the invariant that each $F_i$ is a collection of trees $F_1^{(i)}, \ldots, F_{t_i}^{(i)}$ in the 
graph  $G'_i = (V', E' \setminus \bigcup_{j=1}^{i-1} F_j)$. We let $S_1^{(i)}, \ldots, S_{t_i}^{(i)}$ be the partition of $V'$ induced by the connected components of the trees in $F_i$.  
Each $F_i$ is initialized to be empty, and thus corresponds to the trivial partition of $V'$ by sets of size one.   A key property that we maintain is that the partitions form a \emph{laminar} family: the partition $S_1^{(i+1)}, \ldots, S_{t_{i+1}}^{(i+1)}$ is a refinement of the partition $S_1^{(i)}, \ldots, S_{t_i}^{(i)}$.
This property means that if a (super-)vertex $U \in S_j^{(r)}$ has no edge leaving $S_j^{(r)}$, then $U$ will not have an edge leaving any of $S_j^{(r-1)},\dots,S_j^{(1)}$ either.

The adjacency matrix of the contracted graph $G'$ is no longer Boolean. To still take advantage of separating matrices as in \cref{lem:learn}, we will actually operate on the vertices of $V$ instead of the (super-)vertices of $V'$. To aid in this we use the notation $T_j^{(i)} = \bigcup_{U \in S_j^{(i)}} U$, for all $i=1,\ldots, r$ and $j$.

We initialize all $v \in V$ as $\act$.  If at some point we discover that $v \in T_j^{(r)}$ satisfies $|E(v, \overline{T}_j^{(r)})| = 0$ then we change $v$ to $\inact$. Indeed, by the aforementioned laminar property we know that $|E(v, \overline{T}_j^{(i)})| = 0$ for all $i \leq r$ as well, and so it will not have an outgoing edge with respect to any of the forests.

We again proceed in rounds, until the number of active components in $F_r$ has decreased by a factor of $\Omega(\log n)$. 
Let $G_i'$ be the graph $G'$ with supervertices contracted according to the edges in $F_i$, for $i=1, \ldots, r$ at this point of the algorithm.  
By the laminar property, the number of supervertices in each of these contracted graphs is at most $t_r \in O(q/\log n)$.
We first complete finding a spanning forest of $G_1'$ using the simple spanning forest algorithm \cref{lem:simple} with $O(q)$ queries.  We then remove the edges found in completing the spanning forest of $G_1'$ from $G_2'$ and complete finding a 
spanning forest for $G_2'$ via the simple spanning forest algorithm with $O(q)$ queries.  We continue in this way removing previous edges found and finding spanning forests for each $G_i'$ for $i=3, \ldots, r$ to finish finding a sparse $r$-edge connectivity 
certificate with $O(rq)$ more queries.

Let us now describe a generic round $k$ of the algorithm, where we have partitions $S_1^{(i)}, \ldots, S_{t_i}^{(i)}$ for $i=1,\ldots, r$.  In round $k$
we will simulate queries to the graph $G$ where all edges already a part of $F_1, \ldots, F_r$ are removed.  We let $E_k$ denote this set of edges.
\paragraph{Step 1.}  This step is very similar to step~1 of the algorithm in \cref{thm:spanning}.
For $j=1, \ldots, t_r$ we query $|E_k(v, \overline{T}_j^{(r)})|$ for $v \in T_j^{(r)} \cap \act$.  If $|E_k(v, \bar{T}_j^{(r)})| = 0$ then $v$ becomes $\inact$; if 
$|E(v, \bar{T}_j^{(r)})| > 0$ then $v$ becomes the \emph{representative} of $T_j^{(r)}$ and we move on to $T_{j+1}^{(r)}$ without any further queries in $T_j^{(r)}$.
The number of queries in this step is $O(t_r + w)$ where $w$ is the number of vertices that become $\inact$ in this round.  Again, over all rounds the contribution to the number of queries from vertices becoming inactive is $O(n)$.

\paragraph{Step 2.} 
Let $t_r' \le t_r$ be the number of representatives found in the previous step, and let us assume that these are representatives for the 
sets $T_1^{(r)}, \ldots, T_{t_r'}^{(r)}$.  For $j=1, \ldots, t_{r}'$ we color each $T_j^{(r)}$ red or blue independently at random with equal probability, and give all vertices inside it the same color.  
For each red representative $v \in T_j^{(r)}$ we query its number of blue neighbors and let $W$ be the set of all red representatives where this 
number is positive.  Consider the submatrix $M$ of the adjacency matrix of $G$ whose rows are labeled by elements of $W$ and columns are 
labeled by blue vertices.  By \cref{cor:learn} with $k = 10$ for every element of $W$ we can learn the name of a blue neighbor with 
$O( t_r' + t_r' \log(n)/\log(2t_r'/\log(n))) = O(t_r' \log(n)/\log(2t_r'/\log(n)))$ 
cut queries.

Via this process we learn $|W|$ edges.   We add each of these edges into the spanning forest $F_i$ for the \emph{least} value of $i$ where it does not create a cycle.  
As this set of edges is necessarily cycle free with respect to $F_r$, 
\emph{all} of the edges can be inserted somewhere, and so the \emph{total} number of sets in the $r$ forests goes down by at least $|W|$.  
The expected size of $|W|$ is at least $t_r'/4$ as we expect half of the representatives to be red, and at least half of these to have a neighbor that is blue.  

\paragraph{Total number of queries.}
The number of queries in a round depends on the number of components in the last spanning forest $F_r$.  Apart from the queries made discovering 
inactive vertices, which we know is $O(n)$ over the course of the entire algorithm, the number of queries made in a round is $O( t_r' \log(n)/\log(2t_r'/\log(n)))$.
Thus we must analyze how $t_r'$ decreases over the course of the algorithm.

\begin{claim}
\label{clm:decrease}
The number of active components in $F_r$ decreases by a factor of $1/2$ after $16r$ rounds with probability at least $1/2$.
\end{claim} 

\begin{proof}
Fix a round $k$, and suppose that at the start of round $k$ the number of active components of $F_r$ is $\alpha$.  Note that then the total number of components 
over all $F_1, \ldots, F_r$ is at most $r \alpha$ by the laminar property of these components.  We define two random variables at round $k+i$.  Let $Q_i$ be the random variable denoting the number of active components of $F_r$ at the start of round $k+i$, and let $W_i$ be the random variable denoting the number of edges found in round $k+i$. The key fact we need from the preceding discussion is that $\bbE[W_i] \ge \bbE[Q_i]/4$.  

The total expected number of edges we find after $\ell$ rounds is 
\[
\bbE\left[ \sum_{i=1}^\ell W_i \right] = \sum_{i=1}^\ell \bbE[W_i] \ge \sum_{i=1}^\ell \bbE[Q_i]/4 \enspace .
\]
As every found edge decreases the number of components over all $F_1, \ldots, F_r$ by one, and there are at most $r \alpha$ components in total, this expectation 
is upper bounded by $r\alpha$.  As $\bbE[Q_i]$ is a non-increasing function with $i$ it must therefore be the case that $\bbE[Q_i] \le \alpha/4$ for all $i \ge 16r$.  
Therefore, by Markov's inequality $\Pr[Q_i \ge \alpha/2] \le 1/2$ for all $i \ge 16r$.
\end{proof}

We can now bound the total number of queries from the $O( t_r' \log(n)/\log(2t_r'/\log(n)))$ terms in a similar way as we did in \cref{thm:spanning}.  
Let $T(s)$ be the cost of this term over the course of the algorithm starting from when $t_r' = s$.  Let $f(s) = s \log(n)/\log(2s/\log(n))$ be the 
round cost.  Then by \cref{clm:decrease} we have $T(s) \le 16r f(s) + T(s)/2 + T(s/2)/2$.  This means $T(s) \le 32r f(s) + T(s/2)$.
With $c = 1/2$ and for any $j$, the quantity $T(q)$ is hence of the order
\[
32r q \log(n) \left( \frac{1}{\log(2q/\log(n))} + \frac{c}{\log(2cq/\log(n))} + \cdots  + \frac{c^j}{\log(2c^j q/\log(n))} \right) + T(c^{j+1} q) \enspace .
\]
Once $c^{j+1} \leq 1/\log(n)$ we switch to the simpler algorithm, and so $T(c^{j+1} q) \in O(qr)$.
Hence it remains to bound the preceding sum for $c^j > 1/\log(n)$, in which case we can bound $\log(2c^i q/\log(n)) \geq \log(2q/\log^2(n)) \in \Omega(\log(q))$ for all $i \leq j$ because by assumption 
$q \geq \log^{2+\varepsilon}(n)$ for some $\varepsilon > 0$. The sum then becomes $O(r q \log(n)/\log(q))$, finalizing the proof.
\end{proof}

For the edge connectivity algorithm we will make 
use of a Monte Carlo version of \cref{thm:intro-certificate}, which we state here for 
reference.
\begin{corollary}
\label{cor:monte_carlo}
Let $G = (V,E)$ be an $n$-vertex simple graph, and let $G' = (V',E')$ be a contraction of $G$ with $q$ supervertices where $q \ge \log^{2+\eps}(n)$ for some $\eps > 0$.  Let $r \le n$ 
be a positive integer.
There is a randomized algorithm that makes $O(n + rq \log(n)/\log(q))$ cut queries and with probability $99/100$ outputs a sparse $r$-edge connectivity certificate for $G'$ and 
otherwise outputs FAIL.
\end{corollary}

\section{Edge connectivity with $O(n \log \log n)$ cut queries} 
\label{sec:classical almost linear}

\subsection{Sparse star contraction}
In our edge connectivity algorithm with quantum cut or matrix-vector multiplication queries we used \emph{uniform star contraction}---we randomly chose a set of center vertices $R$ 
by taking each vertex with probability $p = \Theta(\log(n)/\delta(G))$, 
and considered the bipartite directed subgraph $H = (V, \cutd(V\setminus R))$.  For every vertex in $V \setminus R$ we then independently at random chose an outgoing edge in $H$ and 
contracted the set of 
selected edges.  For these algorithms we could afford to learn the entire subgraph $H$ when $\delta(G) \ge \sqrt{n}$ within the desired $\tOh(\sqrt{n})$ query bound.  

In the randomized cut query model it is too expensive to learn $H$ entirely.  With our main tool for learning a bipartite graph, \cref{lem:learn}, we expect to spend $O(n\ell)$ queries 
to learn $H$, where $\ell$ is the maximum degree of a vertex in $V \setminus R$.  To achieve our goal of an $O(n)$ cut query algorithm, therefore, 
we would like to work with a directed subgraph $H$ where the left hand side has \emph{constant} degree.  

To get an $H$ where vertices on the left hand side have constant degree in expectation using star contraction we would have to take $p = \Theta(1/\delta(G))$.  
In this case, however, in expectation $\Omega(n)$ vertices would have no neighbor in $R$ at all, thus doing 1-out contraction on $H$ would not greatly reduce the number of vertices in the contracted graph.

The solution in this section is to perform \textit{sparse star contraction} that uses two different sampling probabilities.  First we randomly choose a set of center vertices $R$ 
by taking each vertex with a slightly larger probability $p = \Theta(\log(\delta(G))/\delta(G))$ and letting $H = (V, \cutd(V\setminus R))$.  By a Chernoff bound, with constant probability 
now only $O(n/\delta(G))$ vertices in $V\setminus R$ have no outgoing edge in $H$.  Let $S \subseteq V \setminus R$ be the set of vertices with positive outdegree in $H$.  
We then find a subgraph $H' = (S \cup R, A)$ of $H$ where every vertex in $S$ has an outgoing edge, but the maximum degree of a vertex in $S$ is \emph{constant}.  
In doing a random 1-out contraction on $H'$, the resulting contracted graph $G'$ will still only have $|R| + O(n/\delta(G)) \in O(n \log(\delta(G))/\delta(G))$ vertices, 
and intuitively we can hope to learn such an $H'$ with only $O(n)$ cut queries as the left hand side has constant degree.  

The tricky part of doing this is to ensure that $H'$ is still $(\alpha,\beta)$-good for contracting with respect to a non-trivial minimum cut for some $\alpha < 1$ and constant $\beta$.  
We find the graph $H'$ by using Recover-$k$-From-All (\cref{alg:recover}), taking $k$ to be a large constant, to learn $k$ neighbors of every vertex in $S$.  These learned edges 
define the graph $H'$.  To review, what happens in Recover-$k$-From-All is that we first bucket the vertices in $S$ into buckets with similar degrees in $H$.  For a bucket with degree approximately 
$r$, we then run LearnBucket (\cref{alg:learnBucket}) which samples a subset of $R' \subseteq R$ by selecting each vertex of $R$ with probability $2k/r$.  In expectation, each vertex in the bucket has a \emph{constant} number of 
neighbors in $R'$, as $k$ is a constant.  With cut queries we can easily check which vertices in the bucket were successfully ``caught'', where a vertex is caught if its number of neighbors in $R'$ is within a constant 
factor of $2k$, its expectation.  
For all the vertices caught we then learn all their neighbors in the subsample and add these edges to $H'$.  We then repeat this routine until all vertices in the bucket are caught.  

From the point of view of a single vertex $v \in S$ with degree $r$ in $H$, its neighbors in $H'$ will be its neighbors in a random subset $R'$ of $R$, where each vertex of $R$ is taken with 
probability $p$, \emph{conditioned} on $d_{R'}(v)$ being close to its expectation.  We need to show that this process 
does not select too high a fraction of edges from a non-trivial minimum cut even when $p = 2k/r$ for a large constant $k$. Specifically, we want to upper bound the probability that $c_{R'}(v)/d_{R'}(v) \ge c(v)/d(v) + 1/10$.  
This requires a different proof than we used in \cref{lem:good_for_contracting}
where $p = \Omega(\log(n)/\delta(G))$.  With $p$ this large we can argue by a Chernoff bound that with high probability $d_{R'}(v) = \Omega(\log n)$ and then again by a Chernoff bound that the probability 
that $c_{R'}(v)$ is both $\Omega(\log n)$ and greatly exceeds its expectation is negligible.  When $p = 2k/r$ with constant $k$, such an argument would only upper bound the probability that 
$c_{R'}(v)/d_{R'}(v) \ge c(v)/d(v) + 1/10$ by an absolute constant.  This is not good enough for us because the number of vertices incident on a non-trivial minimum cut can be 
$\Omega(\delta(G))$, so this does not allow us to use a union bound.

The key to our proof is to show that for $p = 2k/r$ we can upper bound the probability that $c_{R'}(v)/d_{R'}(v) \ge c(v)/d(v) + 1/10$ by a small constant times $\frac{c(v)}{kd(v)}$.  By relating the failure probability to $c(v)/d(v)$ and taking 
$k$ to be a large enough constant, we can again use a union bound since we know that $\sum_{v \in N(C)} c(v)/d(v) \le 2$.  We prove this in \cref{lem:preserve_main}.  Then in \cref{lem:bucket-rand-v2} we formally verify 
that (a small modification of) Recover-$k$-From-All has the required properties needed to show that $H'$ is indeed $(\alpha,\beta)$-good for contracting. We prove these lemmas in the next subsection before giving 
a randomized algorithm for edge connectivity making $O(n \log \log n)$ cut queries in \cref{sec:easy_algo}.  To get down to $O(n)$ cut queries one more trick is needed, which is postponed to \cref{sec:main_algo}.
 
 \subsection{Preparatory lemmas}
\begin{restatable}{lemma}{chebyshev}
\label{lem:preserve_main}
Let $G = (V,E)$ be a simple $n$-vertex graph and let $C \subseteq E$.  Let $v \in N(C)$ and $k \ge 10$.
Choose a set $R$ by putting each vertex of $V$ into $R$ independently at random with probability $p\ge 2k/d(v)$.  Let $0< f \le pd(v)/2$ and $g \ge 2pd(v)$.  Then 
\[
\Pr_R\left[ \frac{c_R(v)}{d_R(v)} \ge \frac{c(v)}{d(v)} + \frac{1}{10} \;\Big|\; f \le d_R(v) \le g \right] \le \frac{200}{k} \frac{c(v)}{d(v)} \enspace. 
\]
\end{restatable}
The proof is deferred to \cref{appen:proof_of_lemma}.  At a high-level, the idea of the proof is the following.
We already computed $\bbE_R[ \frac{c_R(v)}{d_R(v)} \mid f \le d_R(v) \le g] = c(v)/d(v)$ in \cref{prop:expectation}.  To prove 
\cref{lem:preserve_main} we also 
compute $\bbE_R[ \frac{c_R(v)^2}{d_R(v)^2} \mid f \le d_R(v) \le g]$. This allows us to upper bound the variance of 
$c_R(v)/d_R(v)$ by $2c(v)/(kd(v))$ in \cref{prop:var3}.   We then obtain
\cref{lem:preserve_main}
by Chebyshev's inequality.

The next lemma summarizes the state of affairs after 
choosing 
the set $R$ of centers using 
$p = \Theta(\log(\delta(G))/\delta(G))$.
\begin{lemma}
\label{lem:Rgood}
Let $G = (V,E)$ be an $n$-vertex simple graph with minimum degree $d \ge 5 \cdot 10^6$ and let $C$ be a non-trivial minimum cut of $G$.  Let $p = \frac{10^5 \log(d)}{d}$ and 
choose a set $R$ by putting each $v \in V$ into $R$ independently at random with probability $p$.  With probability at least $2/3$ over the choice of $R$ the 
following conditions will simultaneously hold
\begin{enumerate}
	\item $|R| < \frac{3 \cdot 10^5 n \log(d)}{d}$.
	\item $|\{v \in  V : d_R(v) \le 5 \cdot 10^4 \log(d) \}| \le \frac{n}{10^3 d}$.
	\item The graphs $H = (V,\cutd(V \setminus R))$ and $G[R]$ are $(3/5,8)$-good for contracting with respect to~$C$.\footnote{The fact that $G[R]$ is good for contracting will only be used in the $O(n)$ algorithm in the next section.}
\end{enumerate}
\end{lemma}

\begin{proof}
We will upper bound the probability that each item \emph{does not} happen.  A union bound will then give the lemma.
\paragraph{Item~1} The expected size of $R$ is $10^5 n \log(d)/d$.  As the elements of $R$ are chosen independently we can apply a Chernoff 
bound to see that the probability that $|R| \ge 3 \cdot 10^5 n \log(d)/d$ is at most $\exp(-10^5 n \log(d)/d) < 10^{-3}$.  

\paragraph{Item~2} For $v \in V$ we have $\bbE_R[d_R(v)] \ge 10^5 \log(d)$.  As the elements of $R$ are chosen independently we can apply a Chernoff 
bound to see that the probability over $R$ that $d_R(v) \le 5 \cdot 10^4 \log(d)$ is at most $\exp(-10^4 \log(d))$.  Therefore the expected number of $v$ with 
$d_R(v) \le 5 \cdot 10^4 \log(d)$ is at most $n/(10^6 d)$, and by Markov's inequality item~(2) holds except with probability at most $10^{-3}$.

\paragraph{Item~3} As $C$ is a non-trivial minimum cut, we know that $c(v)/d(v) \le 1/2$ for every $v \in V$.  As we sample with probability $p = \frac{10^5 \log(d)}{d} \ge \frac{10^5 \log(d)}{d(v)}$ 
we can apply \cref{lem:preserve_main} with $k = 10^5 \log(d)/2$ to obtain $\Pr_R[c_R(v)/d_R(v) \ge 3/5 \mid d_R(v) > 0] \le \frac{400}{10^5 \log(d)} \frac{c(v)}{d(v)}$ for any $v \in V$.  Thus as 
$\sum_{c \in N(C)} c(v)/d(v) \le 2$, 
by a union bound the probability that any $v$ violates this is at most $800/(10^5 \log(d)) \le 10^{-3}$. 

Since $\bbE[c_R(v)/d_R(v) \mid d_R(v) > 0]=c(v)/d(v)$ by \cref{prop:expectation}, the probability 
\[
\sum_{v: c_R(v) > 0} c_R(v)/d_R(v) \ge 8
\]
is at most $1/4$ by Markov's inequality.  This shows that both $H$ and 
$G[R]$ are $(3/5, 8)$-good for 
contracting with respect to $C$ except with probability at most $1/4 + 10^{-3}$.

Summing the three failure probabilities, overall the failure probability is at most $3\cdot 10^{-3} + 1/4 < 1/3$, giving the lemma.
\end{proof}

To learn neighbors in $R$ of vertices in $V \setminus R$ 
we will use the next lemma.  This lemma describes a worst-case version of the 
algorithm Recover-$k$-From-All (\cref{alg:recover}) that was 
used in the spanning forest algorithm.  However, we 
need to make some further observations about this algorithm,
namely that neighbors are learned in a sufficiently 
random way that we are able to apply \cref{lem:preserve_main}.

\begin{lemma} \label{lem:bucket-rand-v2}
Let $G = (V,E)$ be an $n$-vertex simple graph and 
$C \subseteq E$.  Let $h \ge 10$ be an integer.  Let $S,T \subseteq V$ be disjoint subsets such that $|S| \ge |T|^{1/3}$ 
and $d_T(v) \ge h$ for all $v \in S$.  Suppose that $H = (S \cup T, \overrightarrow{E}(S,T))$ is $(\alpha,\beta)$-good for contracting with respect to $C$.
There is a randomized algorithm that makes $O(h|S|)$
cut queries and with probability at most $1/100$ outputs FAIL, and otherwise explicitly outputs a graph $H' = (S \cup T, A)$ with $A \subseteq \Ed(S,T)$ where
every vertex in $S$ has outdegree at least $h$ in $H'$ and that with probability at least $1/10 + 200\beta/h$
is $(\alpha + 1/10, 10\beta)$-good for contracting with respect to $C$.
\end{lemma}

\begin{algorithm}[!htbp]
\caption{WC-Recover-$k$-From-All$[G](S,T,k)$}
\label{alg:wcrecover}
 \hspace*{\algorithmicindent} \textbf{Input:} Cut query access to a simple graph $G =(V,E)$, two disjoint subsets $S,T \subseteq V$ with $d_T(v) \ge h$ for all $v \in S$, and a parameter $k$. \\
 \hspace*{\algorithmicindent} \textbf{Output:} The adjacency list $Z$ of a directed graph $H' = (S \cup T, A)$ with $A \subseteq \Ed(S,T)$ such that all vertices in $S$ have outdegree 
 at least $\min\{k,h\}$ in $H'$.
\begin{algorithmic}[1]
\State Let $A$ be the adjacency matrix of $G$ and $M = A(S,T)$.
\State Run Recover-$k$-From-All$[M](k)$ (\cref{alg:recover}) and terminate with FAIL if it makes more than 100 times its expected number of queries.  Otherwise output the 
adjacency list $Z$ returned.
\end{algorithmic}
\end{algorithm}

\begin{proof}
We run a worst-case query complexity version of Recover-$k$-From-All \cref{alg:recover}, which we call WC-Recover-$k$-From-All$[G](S,T,k)$ (\cref{alg:wcrecover}), with $k = h$. 
The original Recover-$k$-From-All is stated as a zero-error algorithm and a bound is given 
on its expected number of queries.  Here we want a worst-case bound on the number of queries, so we set a clock 
on Recover-$k$-From-All and terminate, outputting FAIL, if it makes more than 100 times its expected number of queries.  
The probability that this happens is at most $1/100$.  As $S,T$ are disjoint, we can compute $x^TMy$ for any 
$x \in \{0,1\}^{|S|}, y \in \{0,1\}^{|T|}$ with 3 cut queries 
to $G$ by \cref{cor:bip-query}.  Thus by \cref{cor:learn} the number of queries made is 
$O(h|S|)$, using the assumption $|S| \ge |T|^{1/3}$.  

Now suppose that Recover-$k$-From-All terminates within 100 times 
its expected number of queries and let $Z$ be the adjacency list returned.  This adjacency list defines the directed graph $H'$.  By \cref{cor:learn} 
every list $Z[v]$ has at least $k$ items,
thus as we take $k=h$ every vertex in $S$ has outdegree at least $h$ in $H'$.  It remains to show that 
$H'$ is $(\alpha + 1/10, 10\beta)$-good for contracting with respect to $C$ with probability at least $1/10 + 200\beta/h$.

By the ``further'' statement of \cref{cor:learn} the elements in $Z[v]$ are the neighbors of $v$ in a set $Q$ chosen by placing 
each vertex $u \in T$ into $Q$ with probability at least $\min\{2k/d_T(v),1\}$, conditioned on $Q$ having at least $f$ and at most $g$ 
neighbors of $v$, for $0 < f \le qd_T(v)/2$ and $g \ge 2qd_T(v)$.  By \cref{prop:expectation} we know that
\[
\bbE_{Q}\left[\frac{c_{Q}(v)}{d_{Q}(v)} \mid f \le d_{Q} \le g \right] = \frac{c_T(v)}{d_T(v)} \enspace.
\]
As $H$ is $(\alpha,\beta)$-good for contracting, by linearity of expectation and Markov's inequality
we therefore have that except with probability at most $1/10$ over the choice of $Q$
\[
\sum_{v \in S} \left[ \frac{c_{Q}(v)}{d_{Q}(v)} \right] \le 10 \beta \enspace.
\]

Further as $h \ge 10$ and $Q$ satisfies the hypotheses of \cref{lem:preserve_main}, we can invoke this lemma 
to obtain
\[
\Pr_{Q}\left[ \frac{c_Q(v)}{d_Q(v)} \ge \frac{c_T(v)}{d_T(v)} + \frac{1}{10} \mid f \le d_Q(v) \le g \right] \le \frac{200}{h}
\frac{c_T(v)}{d_T(v)} \enspace .
\]
This will hold for all $v \in S$ except with probability $200 \beta/h$ by a union bound.  This shows that $H'$ is $(\alpha + 1/10, 10\beta)$-good for contracting with respect to $C$ except with probability $1/10 + 200 \beta/h$.
\end{proof}

\subsection{Algorithm and correctness}
\label{sec:easy_algo}
We are now ready for the main result 
of this section.
\begin{theorem}
\label{thrm:linear_query_min_cut_easy}
There is a randomized algorithm that computes the
edge connectivity of a simple graph $G$ with 
probability at least $2/3$ 
after $O(n \log \log n)$ cut 
queries.  If $\delta(G) > \log^{10}(n)$ then only $O(n)$ cut queries are needed.
\end{theorem}

\begin{algorithm}[htbp]
\caption{Randomized $O(n \log \log n)$ cut query edge connectivity algorithm}
\label{alg:randomized-easy}
 \hspace*{\algorithmicindent} \textbf{Input:} Cut query access to a simple graph $G=(V,E)$ with adjacency matrix $A$. \\
 \hspace*{\algorithmicindent} \textbf{Output:} With constant probability outputs the edge connectivity of $G$.
\begin{algorithmic}[1]
\State Compute the minimum degree $d$ of $G$.
\label{line:mindeg-easy}
\If{$d < 5 \cdot 10^6$} 
    \label{line:smalldeg-easy}
	\State Compute a sparse $d$-edge connectivity certificate via \cref{cor:monte_carlo}.
	\State Return the edge connectivity of this certificate.
\EndIf
\State Choose a random set $R$ by putting each vertex in $R$ independently with probability~$p = 10^5 \log(d)/d$. \label{line:random-centers}
\State For each $v \in V$ compute $d_R(v) = |E(v,R\setminus\{v\})|$.
\If{$|R| \ge 3 \cdot 10^5 n \log(d)/d$ {\bf or} $|\{v \in V : d_R(v) \le 5 \cdot 10^4 \log(d)\}| > n/(10^3 d)$} \label{line:badif-easy}
	\State Return FAIL.
\EndIf  
\State $S \gets (V \setminus R) \cap \{v : d_R(v) > 5\cdot 10^4 \log(d)\}$.
\State Run WC-Recover-$k$-From-All$[A(S,R)](5 \cdot 10^3)$.  If this returns FAIL then return FAIL, otherwise let $H'$ be the output. \label{line:wcrecover-easy}
\State Do a random 1-out contraction on $H'$ and let $G'$ be the resulting multigraph. \label{line:1-out-easy}
\If{$d \le \log^{10} n$}
	\State Find a sparse $d$-edge connectivity certificate $F$ of $G'$ via \cref{cor:monte_carlo}.
	\label{line:certificate_packing-easy}
	\State Compute a bipartition $(Y, \overline Y)$ of $V'$ corresponding to a minimum cut of $F$.
\Else
	\State Compute a bipartition $(Y, \overline Y)$ of $V'$ corresponding to a minimum cut of $G'$ via the algorithm of \cite{MukhopadhyayN20} (\cref{thm:MN20}).
	\label{line:MNec-easy}
\EndIf
\State $W \gets \cup_{W_i \in Y} W_i$. \label{line:W-easy}
\State Return $\min\{d, |\cut(W)|\}$. \label{line:output-easy}
\end{algorithmic}
\end{algorithm}

\begin{proof}
The claim follows from \cref{alg:randomized-easy}.
In the first step with $n$ cut queries we compute the minimum degree $d$.  In line 
\ref{line:smalldeg-easy} we then handle 
the small degree case.  As $\lambda(G) \le d$,
the edge connectivity of a sparse $d$-edge 
connectivity certificate of $G$ will equal 
$\lambda(G)$.  This step succeeds with 
probability at least $99/100$ by 
\cref{cor:monte_carlo} and takes $O(nd)=O(n)$ 
cut queries.  

We now assume we are in the $d \ge 5 \cdot 10^6$ case.
As can be seen on line~\ref{line:output-easy}, the output of the algorithm is the 
minimum of $d$ and $|\cut(W)|$ for a subset $W \subseteq V$.  
Thus if $d = \mincut(G)$ the algorithm will always correctly return $d$.

Let us therefore focus on the case 
$\mincut(G) < d$, and 
let $C$ be a fixed non-trivial minimum cut of $G$.  
We randomly choose a set $R$ by putting each $v$ into $R$ with probability $p = 10^5 \log(d)/d$.  Note that $p < 1$ as we have already handled the small 
$d$ case.  $R$ will satisfy the conditions of 
\cref{lem:Rgood} with respect to $C$ with probability at 
least $2/3$.  We condition on this good event happening for the rest of the proof.  In particular, 
items~1 and~2 of \cref{lem:Rgood} mean that we will not fail in line~\ref{line:badif-easy}.

Let $S = (V \setminus R) \cap \{v : d_R(v) > 5\cdot 10^4 \log(d)\}$.  By item~3 of \cref{lem:Rgood} 
and \cref{prop:remove_out} we know that $H = (S \cup R, \Ed(S,R))$ is $(3/5,8)$-good for 
contracting with respect to $C$.
On line \ref{line:wcrecover-easy} we run the algorithm WC-Recover-$k$-From-All on the sets $S$ and $R$ with $k = 5 \cdot 10^3$.  By \cref{lem:bucket-rand-v2} this takes $O(n)$ cut queries.  
Further, by the same lemma, with probability at least $99/100$ this algorithm will not fail, in which case it outputs a directed graph $H' = (S \cup R , A)$ with $A \subseteq \Ed(S,R)$, where 
every vertex in $S$ has outdegree at least $1$.  Further, as $H$ is $(3/5,8)$-good for contracting, by \cref{lem:bucket-rand-v2} $H'$ will be $(7/10,80)$-good for contracting except 
with probability at most $1/10 + 1600/(5 \cdot 10^3) \le 1/2$.  Thus overall the algorithm has succeeded up to this point with probability at least $(2/3)\cdot (99/100) \cdot (1/2) \ge 3/10$.

On line~\ref{line:1-out-easy} we do a random 1-out contraction on $H'$.  As $H'$ is $(7/10,80)$-good for contracting with respect to $C$, by \cref{cor:sample_cor} we do not 
contract an edge of $C$ with probability at least $(1/5)^{110}$.  In this case we will have $\mincut(G') = \mincut(G)$.

Let us compute the number $N$ of supervertices in $G'$.  This is at most $|R| + n/(10^3 d) \le 4 \cdot 10^5 n \log(d)/d$ because for every vertex in $S$ we have contracted an edge 
connecting it to a vertex in $R$, since every vertex in $S$ has an outgoing edge in $H'$.   Therefore if $d \geq \log^{10} n$, the number of vertices in $G'$ is $O(n/\log^9 n)$ and we can run the 
minimum cut algorithm of \cite{MukhopadhyayN20} (\cref{thm:MN20}) on $G'$ on 
line \ref{line:MNec-easy} to compute $\mincut(G')$ with $O(N \log^8 N) = O(n)$ cut queries.  This algorithm succeeds with high probability.  

If $d \leq \log^{10} n$, then we find sparse $d$-edge connectivity certificate in line \ref{line:certificate_packing-easy} using the algorithm from \cref{cor:monte_carlo}.  This algorithm correctly outputs a 
sparse $d$-edge connectivity 
certificate with probability $99/100$ and otherwise outputs FAIL.  The number of cut queries is $O(n + n \log(n) \log(d)/\log(n \log(d)/d)) = O(n \log d) = O(n \log \log n)$.  

Thus with probability at least $(3/10) \cdot (1/5)^{110} \cdot (99/100)$ we will have $\mincut(G) = |\cut(W)|$ for the set $W$ defined on line~\ref{line:W-easy}.  
Therefore by repeating the whole algorithm a sufficiently large constant number of times and outputting the minimum of $|\cut(W)|$ over all sets $W$ produced we can output the edge connectivity 
with probability at least $2/3$.  The cut query complexity is dominated by line~\ref{line:certificate_packing-easy} and is $O(n \log \log n)$.  In the case $d > \log^{10}(n)$ we avoid doing this step 
and only make $O(n)$ cut queries.
\end{proof}

\section{Edge connectivity with $O(n)$ cut queries} \label{sec:classical edge conn}
\label{sec:main_algo}
The bottleneck in the algorithm from the previous section  is that the contracted graph $G'$ 
had $\Omega(n \log(\delta(G))/\delta(G))$ vertices.  We would like to get it down to $O(n/\delta(G))$ so that we can compute a sparse $\delta(G)$-edge connectivity 
certificate of $G'$ with $O(n)$ queries by
\cref{cor:monte_carlo}. 

The bound on the 
number of vertices in $G'$ resulted 
because we had to choose a set $R$ of 
$\Theta(n \log(\delta(G))/\delta(G))$ 
centers in order to ensure that a sufficient number of
vertices had a neighbor in $R$.
Note, however, that we did not contract any edges \emph{inside} the induced subgraph $G[R]$.  As $R$ was chosen 
randomly, however, each vertex in $R$ has $\Theta(\log(\delta(G)))$ neighbors 
in $R$ in expectation.  As all but $O(n/\delta(G))$
vertices are connected to a vertex in $R$, we could further reduce the number of vertices in $G'$ by contracting edges in $G[R]$.

We could potentially do this via another round of star 
contraction inside $G[R]$: as each 
vertex in $R$ has $\Theta(\log(\delta(G)))$ neighbors in $R$ in expectation, we could randomly sample $R' \subseteq R$ by 
taking each vertex of $R$ to be in $R'$ with probability 
$p' = \log \log(\delta(G))/ \log \delta(G)$.  The expected size 
of $R'$ is $n \log \log (\delta(G))/\delta(G)$ and with 
constant probability all but $O(n/\delta(G))$ many $v \in R$
have a neighbor in $R'$.  
Following this idea through can 
give an $O(n \log \log \log n)$ cut query algorithm for edge connectivity.\footnote{In fact, repeatedly applying the same argument yields a query complexity $O(n \log \log \cdots \log n)$ for any constant number of $\log$'s.}
To actually get the number of vertices down to $O(n/\delta(G))$, we follow a different approach based on \emph{2-out contraction} rather than star contraction.

\subsection{2-out contraction on the centers}
We will take advantage of the following lemma shown by Ghaffari, Nowicki, and Thorup about the number of vertices in a graph after a random 2-out contraction.
\begin{lemma}[{\cite[Lemma 2.5]{GNT20}}]
\label{lem:gnt_original}
Let $G = (V,E)$ be a simple $n$-vertex graph with minimum degree $\ell$.  
Independently for each $v \in V$ 
choose two outgoing edges $\{v, u_1\}, \{v,u_2\}$ 
uniformly at random and add them to a set $X$.
Then with high probability the 
graph $(V,X)$ has $O(n/\ell)$ connected 
components.  
\end{lemma}
A nice quality of this lemma is that it can also be applied to a \emph{subgraph} of $G$.  In other words,
if we learn a subgraph $H = (V,E')$ of $G=(V,E)$ such that all vertices in $V$ have degree at least $h$ in $H$ then 
by doing 2-out contraction restricted to edges of $H$, we can still reduce the number of vertices 
in the corresponding contraction of $G$ to $O(n/h)$. 

This will be our approach with the induced subgraph $G[R]$.  If we could learn a subgraph $H$ of $G[R]$ where every vertex has degree $\Omega(\log \delta(G))$, then by doing 2-out 
contraction on $H$ we could reduce the number of vertices in the contraction of $G[R]$ by a $\log \delta(G)$ factor, i.e.\ down to $O(n/\delta(G))$.  Furthermore, using the algorithm WC-Recover-$k$-From-All with $k = \log \delta(G)$
we can hope to learn such a subgraph with $O(n \log^2(\delta(G))/\delta(G)) = O(n)$ cut queries.

A direct obstacle to this plan is that $G[R]$ can have $\Omega(n/\delta(G))$ many vertices with $o(\log\delta(G))$ neighbors in $R$. 
Luckily, we can deal with this by slightly generalizing \cref{lem:gnt_original}. We show that if there is a degree threshold $h$ such 
that only $O(n/h)$ vertices have degree less than $h$, then after 2-out contraction the contracted graph still has only $O(n/h)$ supervertices.

While this appropriately reduces the size of the contracted graph, a second obstacle remains: we have to ensure that we do not contract any edge of a non-trivial minimum cut.  
This is again where the randomness properties of WC-Recover-$k$-From-All shown in \cref{lem:bucket-rand-v2} come in handy.  Using this lemma we will show that we can
explicitly learn a directed subgraph $H = (R,A)$ of $G[R]$ such that (i) all but $O(n/\delta(G))$ vertices have degree at least $h$ in $H$, and (ii) $H$ is $(\alpha, \beta)$-good 
for contracting with respect to a non-trivial minimum cut for some $\alpha < 1$ and constant $\beta$.  We show how to do this in \cref{lem:2outforR} below.

As we view $H$ as a directed graph, we also need to generalize \cref{lem:gnt_original} to this case, where we sample only 
from outgoing edges, not incoming ones.  Both generalizations are captured in the following lemma.
The proof follows the original proof of \cite{GNT20} with minor modifications, and we defer it to \cref{appen:contraction}.
Similar to \cref{lem:gnt_original}, this lemma also applies when $H = (V,A)$ is a subgraph of a larger graph $G = (V,E)$.
\begin{restatable}[{cf.\ \cite[Lemma 2.5]{GNT20}}]{lemma}{gntcontract}
\label{lem:num_con}
Let $H = (V,A)$ be an $n$-vertex 
directed graph such that all but 
$\tau$ vertices have out-degree 
at least $\ell \ge 4$.  Independently for each $v \in V$ 
choose two outgoing edges $(v, u_1), (v,u_2)$ 
uniformly at random and add them 
to a set $X$.  Then with high probability the 
graph $(V,X)$ has at most $\tau + 2n/\ell$ weakly connected 
components.  
\end{restatable}

Next we give the algorithm based on WC-Recover-$k$-From-All that we will use to build the directed subgraph $H$ of $G[R]$ on which we will 
perform the 2-out contraction.  

\begin{lemma}
\label{lem:2outforR}
Let $G = (V,E)$ be an $n$-vertex simple 
graph and $C \subseteq E$.  Suppose that $G$ is $(\alpha, \beta)$-good for contracting with respect to $C$.
Let $h \ge \max\{1500 \beta, 35\}$ 
be such that for all but $\tau$ vertices in $v \in V$ it holds that $d(v) \ge 4 h$.
There is a randomized algorithm that makes 
$O(h n)$ cut queries and with probability at most $3/100$ outputs FAIL, and otherwise outputs a directed subgraph $H$ of $G$ 
where all but $\tau+n/h$ vertices in $V$ have outdegree at least $h$ in $H$.  Further, $H$ is 
$(\alpha + 1/5,100\beta)$-good for contracting with respect to $C$ with probability at least $1/2$.  
\end{lemma}

\begin{algorithm}[!htbp]
\caption{LearnSubgraph$[G](h)$}
\label{alg:learn_subgraph}
 \hspace*{\algorithmicindent} \textbf{Input:} Cut query access to a simple graph $G =(V,E)$ with adjacency matrix $A$ and a parameter $h$ such that for all but $\tau$ vertices in $v \in V$ 
 it holds that $d(v) \ge 8h$.  \\
 \hspace*{\algorithmicindent} \textbf{Output:} A directed subgraph $H$ of $G$ where all but $\tau + n/h$ vertices have outdegree $\ge h$.
\begin{algorithmic}[1]
\State Randomly partition $V$ into two sets $V_1$ and $V_2 = V \setminus V_1$ by putting each vertex 
independently at random into $V_1$ with probability $1/2$ and otherwise into $V_2$. 
\State $V_1' \gets \{v \in V_1: d_{V_2}(v) \ge h\}, V_2' \gets \{v \in V_2: d_{V_1}(v) \ge h\}$.
\If{$|V \setminus (V_1' \cup V_2')| >  \tau+ n/h$} return FAIL. \label{line:first_fail}
\EndIf
\State Run WC-Recover-$k$-From-All$[A(V_1',V_2)](h)$ and WC-Recover-$k$-From-All$[A(V_2',V_1)](h)$.  If either call returns FAIL then 
return FAIL.  Otherwise, let $Z_1$ and $Z_2$ be the outputs.
\State Return the directed graph $H$ defined by the concatenation of $Z_1$ and $Z_2$.
\end{algorithmic}
\end{algorithm}

\begin{proof}
The algorithm is given in \cref{alg:learn_subgraph}.  Let us first check the probability that we fail on line~\ref{line:first_fail}.
We first randomly partition $V$ into two sets $V_1$ and $V_2 = V \setminus V_1$. 
Let $V_1' = \{v \in V_1: d_{V_2}(v) \ge h\}$ and $V_2' = \{v \in V_2: d_{V_1}(v) \ge h\}$.
If $v \in V_1$ has $d(v) \ge 4h$ then the probability that $v$ is not in $V_1'$ is at most $\exp(-h/4)$ 
by a Chernoff bound (\cref{eq:chernoff_below}).  The same is true for any $v \in V_2$, therefore the expected number of vertices 
with degree at least $4h$ that are not in $V_1' \cup V_2'$ is at most $n \exp(-h/4)$.  By Markov's inequality 
therefore we have $|V \setminus (V_1' \cup V_2')| \le  \tau+ n/h$ except with probability $h \cdot \exp(-h/4) \le 1/100$ as $h \ge 35$.  
Thus the probability that we fail on line~\ref{line:first_fail} is at most $1/100$.  Checking this condition can be done with $O(n)$ 
cut queries as we can compute $d_{V_2}(v)$ with a constant number of cut queries, and likewise for $d_{V_1}(v)$.

By \cref{prop:expectation} we have $\bbE[c_{V_2}(v)/d_{V_2}(v) \mid d_{V_2}(v) > 0] = c(v)/d(v)$ for every $v \in V_1$ and $\bbE[c_{V_1}(v)/d_{V_1}(v) \mid d_{V_1}(v) > 0] = c(v)/d(v)$ 
for every $v \in V_2$.  Thus by Markov's inequality we have 
\begin{equation}
    \sum_{v \in N(C) \cap V_1'} \frac{c_{V_2}(v)}{d_{V_2}(v)} + \sum_{v \in N(C) \cap V_2'} \frac{c_{V_1}(v)}{d_{V_1}(v)} \le 10 \beta
    \label{eq:sum_prop}
\end{equation}
except with probability at most $1/10$.  

Further, by \cref{lem:preserve_main} as we sample with probability $p = 1/2$ we can apply \cref{lem:preserve_main} with $k = h$ together 
with a union bound to obtain that except with probability at most $200\beta/h$ we have
\begin{equation}
\label{eq:max_prop}
\begin{aligned}
    \frac{c_{V_2}(v)}{d_{V_2}(v)} &\le \frac{c(v)}{d(v)} + \frac{1}{10} \text{ for all } v \in V_1' , \\
     \frac{c_{V_1}(v)}{d_{V_1}(v)} &\le \frac{c(v)}{d(v)} + \frac{1}{10} \text{ for all } v \in V_2' \enspace .
\end{aligned}
\end{equation}
To summarize, \cref{eq:sum_prop} and \cref{eq:max_prop} show that the graph $F = (V, \Ed(V_1', V_2) \cup \Ed(V_2', V_1))$ is $(\alpha+1/10, 10 \beta)$-good for contracting with respect to $C$ 
except with probability at most $1/10 + 200\beta/h$.  We now condition on this good event that $F$ is $(\alpha+1/10, 10 \beta)$-good for contracting.

The goal now is to learn $h$ neighbors 
in $V_2$ of every vertex in $V_1'$, and vice 
versa, which we do by running WC-Recover-$k$-From-All$[A(V_1',V_2)](h)$ and  
WC-Recover-$k$-From-All$[A(V_2',V_1)](h)$.  The total number of queries is $O(hn)$ by \cref{lem:bucket-rand-v2}.  If either call outputs FAIL, then we abort and output FAIL, which happens with probability at most $2/100$.  
We now condition on both of these calls being successful and let $Z_1$ and $Z_2$ be the adjacency lists returned.  
The directed subgraph $H$ is defined by the concatenation of $Z_1$ and $Z_2$.   
When these calls do not fail, every vertex of $H$ with positive outdegree has outdegree at least $h$, as each list in $Z_1,Z_2$ has at least $h$ neighbors by \cref{lem:bucket-rand-v2}.  Thus the number of vertices with zero outdegree in $H$ is at most 
$|V \setminus (V_1' \cup V_2')| \le \tau +  n/h$ assuming we did not FAIL in line~\ref{line:first_fail}.  In summary, with probability at most $3/100$ the algorithm outputs FAIL, and otherwise it always 
returns a directed subgraph $H$ where all but at most $\tau +  n/h$ vertices have outdegree at least $h$.  Further, item~3 of
\cref{lem:bucket-rand-v2} together with the fact that $F$ is $(\alpha+1/10, 10\beta)$-good for contracting tells us that $H$ is $(\alpha + 1/5,100\beta)$-good for contracting with respect to $C$ 
except with probability at most $1/10 + 200 \beta/h$.  Thus overall, $H$ will be $(\alpha + 1/5,100\beta)$-good for contracting with respect to $C$ except with probability at most $2/10 + 400 \beta/h \le 1/2$.
\end{proof}

\subsection{Algorithm and correctness}
We are now ready to give a randomized $O(n)$ cut query algorithm for edge connectivity.

\begin{restatable}{theorem}{LinearQueryMinCut}
\label{thrm:linear_query_min_cut}
There is a randomized algorithm that computes the
edge connectivity of a simple graph with 
probability at least $2/3$ after $O(n)$ cut 
queries.
\end{restatable}

\begin{algorithm}[!htbp]
\caption{Randomized $O(n)$ cut query edge connectivity algorithm}
\label{alg:nquery}
 \hspace*{\algorithmicindent} \textbf{Input:} Cut query access to a simple graph $G=(V,E)$ with minimum degree $d < \log^{10}(n)$. \\
 \hspace*{\algorithmicindent} \textbf{Output:} With constant probability output the edge connectivity of $G$.
\begin{algorithmic}[1]
\State Compute the minimum degree $d$ of $G$.
\label{line:mindeg}
\If{$d < 5 \cdot 10^6$} 
    \label{line:smalldeg}
	\State Compute a sparse $d$-edge connectivity certificate via \cref{cor:monte_carlo}.
	\State Return the edge connectivity of this certificate.
\EndIf
\State Choose a random set $R$ by putting each vertex into $R$ independently at random with probability $p = 10^5 \log(d)/d$.
\State For each $v \in V$ compute $d_R(v) = |E(v,R\setminus\{v\})|$. \label{line:degree-to-R}
\If{$|R| \ge 3 \cdot 10^5 n \log(d)/d$ OR $|\{v \in V : d_R(v) \le 5 \cdot 10^4 \log(d)\}| > n/(10^3 d)$} \label{line:badif}
	\State Return FAIL.
\EndIf  
\State Let $S = (V \setminus R) \cap \{v : d_R(v) > 5\cdot 10^4 \log(d)\}$ and $A$ be the adjacency matrix of $G$.
\State Run WC-Recover-$k$-From-All$[A(S,R)](5 \cdot 10^3)$.  If this returns FAIL then return FAIL, otherwise let 
$H'$ be the output. \label{line:wcrecover}
\State Run LearnSubgraph$[G[R]](h)$  (\cref{alg:learn_subgraph}) with $h = 5 \cdot 10^4 \log(d)$. 
If this returns FAIL then return FAIL, otherwise let 
$H$ be the output.  \label{line:2out}
\State Take a random 1-out sample of $H'$ and a random 2-out sample of $H$ and contract all selected edges in $G$.  Let $G'$ be 
the resulting multigraph.
\State Return FAIL if $G'$ has more than $n/(10^3d) + 3|R|/h$ vertices. \label{line:contract_too_many}
\State Find a sparse $d$-edge connectivity certificate $F$ of $G'$ via \cref{cor:monte_carlo}.
\label{line:certificate_packing}
\State Compute a bipartition $(Y, \overline Y)$ of $V'$ corresponding to a minimum cut of $F$.
\State $W \gets \bigcup_{W_i \in Y} W_i$. \label{line:W}
\State Return $\min\{d, |\cut(W)|\}$. \label{line:output}
\end{algorithmic}
\end{algorithm}

\begin{proof}
We can restrict to the case that the minimum degree $\delta(G) \leq \log^{10} n$ because an $O(n)$ cut query algorithm for the case of larger degree 
is already handled by \cref{thrm:linear_query_min_cut_easy}.  The algorithm is given in \cref{alg:nquery}.  As argued in the proof of \cref{alg:randomized-easy}, the 
algorithm will always return correctly when the edge connectivity is achieved by a trivial cut.  Let us therefore analyze the case that the edge connectivity 
is achieved by a non-trivial cut $C$.

The algorithm is identical to \cref{alg:randomized-easy} until line~\ref{line:2out}.  From the proof of 
\cref{thrm:linear_query_min_cut_easy}, at this point of the algorithm with probability at least $3/10$ we will be in the state where
\begin{itemize}
	\item The set $R$ satisfies the three conditions of \cref{lem:Rgood}.
	\item The call to WC-Recover-$k$-From-All did not fail, and the returned graph $H'$ is $(7/10,80)$-good for contracting with respect to $C$.
\end{itemize}
Next, on line~\ref{line:2out} we run \cref{alg:learn_subgraph} on $G[R]$ with $h = 5 \cdot 10^3 \log(d)$.  This takes $O(n \log^2(d)/d) = O(n)$ cut queries by 
\cref{lem:2outforR}.  Note that by item~2 of \cref{lem:Rgood}, 
at most $n/(10^3 d)$ vertices in $R$ have $d_R(v) < 8h < 5 \cdot 10^4 \log(d)$, thus we can take $\tau = n/(10^3d)$ in \cref{lem:2outforR}.  Further, by item~3 of \cref{lem:Rgood} $G[R]$ 
is $(3/5,8)$-good for contracting with respect to $C$.  With $\beta =8$ our choice of $h$ satisfies $h \ge 1500 \beta$ as we are in the 
case $d \ge 5 \cdot 10^6$.  Thus we are in a position to apply \cref{lem:2outforR},
which tells us that with probability at least $1/2$ the directed subgraph $H$ of $G[R]$ returned by the algorithm will be $(7/10, 800)$-good for contracting with respect to $C$.  Let us assume this is the case, and 
the probability the algorithm reaches this good state is at least $(3/10)(1/2) = 3/20$.
Further the number of vertices in $R$ with outdegree less than $h$ in $H$ is at most $n/(10^3 d) + |R|/h$.

As argued in the proof of \cref{alg:randomized-easy}, the probability we do not select an edge of $C$ in taking a 1-out sample of $H'$ is at least $(1/5)^{110}$.  As $H$ is $(7/10,800)$-good for contracting with respect to $C$, 
by \cref{clm:sample} the probability that we do not select an edge of $C$ in taking a random 2-out sample of $H$ is at least $(3/10)^{2286}$.  We apply 
\cref{lem:num_con} with degree threshold $h$ to see that with high probability (for concreteness say $99/100$) the number of vertices in the contraction of $G[R]$ by the edges in the 2-out sample will be 
the number of vertices with outdegree $< h$, which is at most $n/(10^3d) + |R|/h$, plus $2|R|/h$.  In particular, we do not fail in line~\ref{line:contract_too_many} with probability at least $99/100$.
Overall, we are now in the good case that all steps of the algorithm have been successful with probability at least $(3/20) \cdot (1/5)^{110} \cdot (3/10)^{2286} \cdot (99/100)$, and $G'$ has at most $7n/d$ vertices.

Finally, we find a sparse $d$-edge connectivity certificate $F$ of $G'$.  By \cref{cor:monte_carlo} this succeeds with probability at least $99/100$ and takes 
$O(n + n \log(n)/\log(n/d))$ queries, which is $O(n)$ overall as $d < \log^{10} n$.  Hence with probability at least $(5/24) \cdot (1/5)^{110} \cdot (3/10)^{2286} \cdot (99/100)^2$ 
we will correctly output the edge connectivity on line~\ref{line:output}.  As we never output a value that is less than the edge connectivity, we can repeat the whole algorithm 
a sufficiently large but constant number of times and output the minimum of the values returned to boost the success probability to $2/3$.
\end{proof}

\section*{Acknowledgements}
This project has received funding from the European Research Council (ERC) under the European Union's Horizon 2020 research and innovation programme under grant agreement No 715672. Danupon Nanongkai and Sagnik Mukhopadhyay were also supported by the Swedish Research Council (Reg. No. 2015-04659 and 2019-05622). Troy Lee is supported in part by the Australian Research Council Grant No: DP200100950.
Pawe\l{} Gawrychowski is partially supported by the Bekker programme of the Polish National Agency for Academic Exchange (PPN/BEK/2020/1/00444).

\bibliography{biblio.bib}

\appendix

\section{Proof of \cref{lem:preserve_main}}\label{appen:proof_of_lemma}
Throughout this appendix we will use the following notation. Let $0 < c < d$ be positive integers.
 Let $X_1, \ldots, X_c, Z_1, \ldots, Z_{d-c}$ be independent and identically distributed Bernoulli random variables that are $1$ with probability $p$.  
 Let $X = \sum_{i=1}^c X_i$ and $Y = \sum_{i=1}^c X_i + \sum_{i=1}^{d-c} Z_i$.  Note that $X \sim B(c,p), Y \sim B(d,p)$ are both binomial random variables.
 
 Although the most of the statements in this appendix will be purely probabilistic, one can keep in mind the following scenario.  We have a graph $G = (V,E)$ and 
 a subset of edges $C$.  Say that a vertex $v \in V$ had degree $d$ and has $c$ edges of $C$ incident to it.  We then sample a subset of edges incident to $v$ 
 by independently taking each edge with probability $p$.  $X_i=1$ represents the event that the $\ith$ edge of $C$ incident to $v$ is selected, and 
 $Z_i=1$ the event that the $\ith$ non-edge of $C$ incident to $v$ is selected.  
Then $X = \sum_{i=1}^c X_i$ is the random variable for the total number of edges of $C$ incident to $v$ selected and $Y = X + \sum_{i=1}^{d-c} Z_i$ is the random 
variable for the total number of edges incident to $v$ selected.
\begin{proposition}
\label{prop:var1}
\[
\bbE[X^2 \mid Y = b] = \frac{cb}{d} + \frac{c(c-1)b(b-1)}{d(d-1)} \enspace.
\]
\end{proposition}

\begin{proof}
By linearity of conditional expectation, $\bbE[X^2 \mid Y = b] = \sum_{i,j} \bbE[X_i X_j \mid Y = b]$.  
In the proof of \cref{prop:expectation} we have already computed that $\bbE[X_i^2 \mid Y = b] = \bbE[X_i \mid Y = b]= b/d$.
Recall that $X_{i}$ and $Z_{j}$ are identically distributed, so also $\bbE[Z_i^2 \mid Y = b] = b/d$.
For the same reason, the following expected values are all equal (i) $\bbE[X_i X_j \mid Y = b]$ for $i\neq j$, (ii) $\bbE[Z_i Z_j \mid Y = b]$ for $i\neq j$,
(iii) $\bbE[X_i Z_j \mid Y = b]$ for any $i,j$.
We then obtain the following for any $i\neq j$:
\begin{align*}
\bbE[Y^{2} \mid Y = b] &= b^{2} \\
\bbE[(X_{1}+\ldots+X_{c}+Z_{1}+\ldots+Z_{d-c})^{2} \mid Y = b] &= b^{2} \\
\bbE[d(d-1)X_{i}X_{j} +d X^{2}_{i}\mid Y = b] &= b^{2} \\
\implies \bbE[X_{i}X_{j} \mid Y = b] &= \frac{b(b-1)}{d(d-1)} \enspace .
\end{align*}
There are $c$ terms of the form $\bbE[X_i^2 \mid Y = b]$ and $c(c-1)$ terms of the form $\bbE[X_i X_j \mid Y = b]$, giving the proposition.
\end{proof}

\begin{definition}[Conditional first inverse moment]
Let $d$ be a positive integer and $f,g$ be integers with $0 < f \le g \le d$. Let $p \in (0,1]$.  Let $Y \sim B(d,p)$ be a binomial 
random variable.  Define $Q(d,p,f,g) = \bbE[1/Y \mid f \le Y \le g]$.
\end{definition}

\begin{proposition}
\label{prop:var2}
Let $0 < f \le g \le d$.  Then 
\[
\bbE[X^2/Y^2 \mid f \le Y \le g] = \frac{c(c-1)}{d(d-1)} + \left(\frac{c}{d} - \frac{c(c-1)}{d(d-1)} \right) Q(d,p,f,g) \enspace .
\]
\end{proposition}

\begin{proof}
Let $\gamma = \Pr[f \le Y \le g]$.  Then we have
\begin{align*}
\bbE[X^2/Y^2 \mid f \le Y \le g] &= \frac{1}{\gamma} \sum_{b=f}^g \frac{1}{b^2} \sum_{a=0}^c a^2 \Pr[X = a, Y = b]  \\
&= \frac{1}{\gamma} \sum_{b=f}^g \frac{\Pr[Y=b]}{b^2} \bbE[X^2 \mid Y = b] \\
&= \frac{1}{\gamma} \sum_{b=f}^g \Pr[Y = b] \left(\frac{c}{bd} + \frac{c(c-1)(b-1)}{d(d-1)b} \right) \\
&= \frac{c(c-1)}{d(d-1)} + \left(\frac{c}{d}- \frac{c(c-1)}{d(d-1)}\right) \bbE[1/Y \mid f \le Y \le g] \enspace . \qedhere
\end{align*}
\end{proof}

In order to apply \cref{prop:var2} we will need to upper bound $Q(d,p,f,g)$.  Calculating the inverse moments of a truncated binomial distribution
is a well-studied problem and precise asymptotic estimates are known, see e.g.\  \cite{MW99}.  For our purposes a looser estimate suffices and 
we opt for a simple self-contained proof adapted from \cite{ChaoS72}.
\begin{proposition}
\label{prop:Qbound}
Let $Y \sim B(d,p)$ and let $0< f \le g$ be such that $\Pr[f \le Y \le g] \ge 1/2$.  Then 
$Q(d,p,f,g) \le \frac{4}{pd}$.
\end{proposition}

\begin{proof}
We have 
\begin{align*}
Q(d,p,f,g) &= \frac{1}{\Pr[f \le Y \le g] } \sum_{b = f}^g \frac{\Pr[Y=b]}{b}  \\
&\le \frac{1}{\Pr[f \le Y \le g] } \sum_{b = 1}^d \frac{1}{b} p^b (1-p)^{d-b} \binom{d}{b} \\
&\le \frac{2}{\Pr[f \le Y \le g] } \sum_{b = 1}^d \frac{1}{b+1} p^b (1-p)^{d-b} \binom{d}{b} \\
&=\frac{2}{\Pr[f \le Y \le g] } \frac{1}{p(d+1)} \sum_{b = 1}^d  p^{b+1} (1-p)^{d-b} \binom{d+1}{b+1} \\
&\le \frac{2}{\Pr[f \le Y \le g] } \frac{1}{p(d+1)} \\
&\le \frac{4}{pd} \enspace . \qedhere
\end{align*}
\end{proof}

\begin{proposition}
\label{prop:var3}
Let $X,Y$ be the random variables defined in \cref{prop:var1}.  Then 
\[
\var[X/Y \mid f \le Y \le g] \le Q(d,p,f,g) \frac{c}{d} \enspace .
\]
\end{proposition}

\begin{proof}
For convenience let $Q = Q(d,p,f,g)$.  We have 
\begin{align*}
\var[X/Y \mid f \le Y \le g]  &= \bbE[X^2/Y^2 \mid f \le Y \le g] - \bbE[X/Y \mid f \le Y \le g]^2 \\
&= \frac{c(c-1)}{d(d-1)} + \left(\frac{c}{d}- \frac{c(c-1)}{d(d-1)}\right) Q - \frac{c^2}{d^2} \\
&= \frac{c}{d} (d-c) \left(\frac{Q}{d-1} - \frac{1}{d(d-1)} \right) \\
&\le Q \frac{c}{d} \enspace . \qedhere
\end{align*}
\end{proof}

We can now derive a more general version of \cref{lem:preserve_main}.
\begin{lemma}
\label{lem:preserve_general}
Let $G = (V,E)$ be a simple $n$-vertex graph and let $C \subseteq E$.  Let $v \in N(C)$ and $k \ge 10$.
Choose a set $R$ by putting each vertex of $V$ into $R$ independently at random with probability $p\ge 2k/d(v)$.  Let $0< f \le pd(v)/2$ and $g \ge 2pd(v)$.  Then 
for any $\alpha > 0$
\[
\Pr_R\left[ \frac{c_R(v)}{d_R(v)} \ge \frac{c(v)}{d(v)} + \alpha \sqrt{\frac{2}{k}} \;\Big|\; f \le d_R(v) \le g \right] \le \frac{1}{\alpha^2} \frac{c(v)}{d(v)} \enspace. 
\]
\end{lemma}

\begin{proof}
Let $c=c(v)$ and $d= d(v)$.  
Let us first upper bound $Q(d,p,f,g)$ with $p \ge 2k/d$.  As $k \ge 10$ by a Chernoff bound (the ``in particular'' of \cref{lem:chernoff}) the probability that 
$R$ contains between $f$ and $g$ neighbors of $v$ is at least $1/2$.  Thus we can apply \cref{prop:Qbound} to see that $Q \le 2/k$.  
Therefore by \cref{prop:var3} we have $\var[c_R(v)/d_R(v) \mid f \le d_R(v) \le g] \le 2c/(k d)$.  We can therefore apply 
Chebyshev's inequality to find
\[
\Pr_R\left[ \frac{c_R(v)}{d_R(v)} \ge \frac{c}{d} + t \sqrt{\frac{2c}{k d}} \;\Big|\; f \le d_R(v) \le g \right] \le \frac{1}{t^2} \enspace. 
\]
Taking $t = \alpha \sqrt{d(v)/c(v)}$ gives the lemma.
\end{proof}
\cref{lem:preserve_main} follows as an easy corollary by taking $\alpha = (1/10)\sqrt{k/2}$.
\chebyshev*

\section{Contraction lemma}
\label{appen:contraction}
To prove \cref{lem:num_con} we will closely follow the proof of the bound 
on the number of vertices after 2-out contraction given in \cite[Lemma 2.5]{GNT20}.
To this end, we need the following definition and lemma.
\begin{definition}[Stochastic domination]
Let $X$ and $Y$ be two random variables not necessarily defined on the same probability space.  We say that $Y$ stochastically 
dominates $X$, written $X \preceq Y$, if for all $\lambda \in \R$ we have $\Pr[X \le \lambda] \ge \Pr[Y \le \lambda]$.  
\end{definition}

\begin{lemma}[{\cite[Lemma 1.8.7]{Doe18}}]
\label{lem:doerr}
Let $X_1, \ldots, X_n$ be arbitrary binary random variables, and let $Y_1, \ldots, Y_n$ be independent binary random variables.  If 
$\Pr[X_i=1 | X_1 = x_1, \ldots, X_{i-1} = x_{i-1}] \le \Pr[Y_i = 1]$ for all $i=1, \ldots, n$ and all $x_1, \ldots, x_{i-1} \in \{0,1\}$ with 
$\Pr[X_1 = x_1, \ldots, X_{i-1} = x_{i-1}] > 0$ then 
\[
\sum_{i=1}^n X_i \preceq \sum_{i=1}^n Y_i \enspace .
\]
\end{lemma}

We are now ready to prove the following.
\gntcontract*

\begin{proof}
We will prove the theorem by considering adding edges to $E'$ in a specific order given by \cref{alg:connected}. 
In this algorithm we maintain a set $\Pcal$ of \emph{processed} vertices, a set $\Acal$ of \emph{active} vertices, and a set $\Scal$ of \emph{sampled} vertices.  
We use the notation $A(v) = \{u \in V: (v,u) \in E\}$ and $u \in_R A(v)$ to denote
choosing an element of $A(v)$ uniformly at 
random.

\begin{algorithm}[H]
\caption{Procedure to add sampled edges}
\label{alg:connected}
\begin{algorithmic}[1]
\State $E' = \emptyset$
\State $\Pcal = \{v \in V: \outdeg(v) < \ell\}$
\label{line:initializeP}
\While{$\overline{\Pcal} \ne \emptyset$}
  \State Select $v \in \overline{\Pcal}$
  \State $\Acal \gets \{v\}, \Scal \gets \emptyset$
  \State $\pflag \gets 0$
  \While{$\Acal \setminus \Scal \ne \emptyset$} \Comment{A run of this while loop is called a \emph{phase}}
    \State Select $v \in \Acal \setminus \Scal$ 
    \State $\Scal \gets \Scal \cup \{v\}$
    \State Sample $u_1, u_2 \in_R A(v)$
    \label{line:sample2}
    \State Update $\Acal \gets \Acal \cup \{u_1, u_2\}$ and $E' \gets E' \cup \{ \{v,u_1\},\{v,u_2\}\}$
    \If{$u_1 \in \Pcal \vee u_2 \in \Pcal$}
      \State $\pflag \gets 1$
    \EndIf    
  \EndWhile
  \State $\Pcal \gets \Pcal \cup \Acal$ \label{line:addA}
\EndWhile
\end{algorithmic}
\end{algorithm}
Note that $\Acal$ is always a connected via edges in $E'$.  We are interested in the properties of $\Acal$ when it is added to the set of processed vertices 
in Line~\ref{line:addA}.  We wish to upper bound the number of times $\kappa$ that $\Acal$ is added to $\Pcal$ and the following 
two conditions hold
\begin{enumerate}
  \item $|\Acal | < \ell$ 
  \item $\pflag = 0$
\end{enumerate}
We can upper bound the number of connected components of $G'$ by $\tau + n/\ell + \kappa$. 
This is because 
\begin{enumerate}
    \item The initial size of $\Pcal$ is $\tau$,
    \item If $\pflag =1$ when $\Acal$ is added to $\Pcal$ this means that $\Acal$ is connected to a set of vertices that has already been processed and hence already counted,
    \item The number of sets added where $\Acal \ge \ell$ is at most $n/\ell$.
\end{enumerate} 
The remaining case is where $\pflag = 0$ and 
$\Acal < \ell$, which is counted by $\kappa$.

Thus our task is to show that with high probability $\kappa \le n/\ell$.  To this end, define a random variable $X_i$ to be $1$ if 
at the end of the $\ith$ phase $|\Acal| < \ell$ and $\pflag = 0$, and $0$ otherwise.  In other words, $X_i = 1$ if and only if the $\ith$ phase contributes to increasing $\kappa$.  

Let us consider the probability that at 
the end of a phase on line~\ref{line:addA} 
$|\Acal| = x$ and $\pflag = 0$ .  In this 
case we chose $2x$ many samples, and 
exactly $x+1$ of these were already in $\Acal$.  Following Ghaffari, Nowicki, and Thorup \cite{GNT20} we say a sample is \emph{caught} if it is 
already in $\Acal$.  The only fact needed to make the \cite{GNT20} proof go through is that
the probability a sample is caught is at most $\frac{x-1}{\ell}$ throughout the course of the phase.  This holds in 
our case as $\Pcal$ is initialized 
to have all vertices of outdegree at 
most $\ell$.  Thus if $\pflag = 0$ at 
the end of a phase then all vertices added
to $\Acal$ during the phase have outdegree at least $\ell$ and the probability 
that a sample on line~\ref{line:sample2} 
in already 
in $\Acal$ is at most $\frac{x-1}{\ell}$.

There are $\binom{2x}{x+1}$ many sequences for the placement of the caught samples.  Thus overall we can upper bound the
probability that $|\Acal| =x$ by 
\[
P_x = \binom{2x}{x+1} \left(\frac{x-1}{\ell}\right)^{x+1} \enspace.
\]
Following the calculation in \cite{GNT20} (displayed equation, bottom of page 7) it follows that
\[
\Pr[|\Acal| \le \frac{\ell}{8e^3}] \le \frac{8}{\ell^3} \enspace .
\]

This means that $\Pr[X_i = 1 | X_1 = x_1 \ldots X_{i-1} = x_{i-1}] \le 8/\ell^3$.  Now define independent random variables 
$Y_i$ that take value $1$ with probability $8/\ell^3$ and $0$ otherwise.  By \cref{lem:doerr}
\[
\Pr[\sum_i X_i > \gamma] \le \Pr[\sum_i Y_i > \gamma] \enspace .
\]
As the $Y_i$ are independent we can upper bound the probability they exceed their expectation by a Chernoff bound.  
We have $\bbE[\sum_i Y_i] \le 8n/\ell^3$.  Thus for any $\gamma \ge 8n/\ell^3$ and $0< \varepsilon$ we have 
by a Chernoff bound that $\Pr[\sum_i Y_i \ge (1+\varepsilon) \gamma)] \le \exp(-\frac{\gamma \varepsilon^2}{2+\varepsilon})$.

If $n/(2\ell) \ge \log^2(n)$ then taking $\gamma = n/(2\ell)$ (which is at least $8n/\ell^3$ as $\ell \ge 4$) and $\varepsilon = 1$ tells us that $\kappa \le n/\ell$ except 
with probability $\exp(-\log^2(n)/3)$.  If $n/(2\ell) < \log^2(n)$ then $8n/\ell^3 \le 64 \log^8(n)/n^2$ and so we can take $\gamma = 8n/\ell^3$ and $\eps = 1/(2\gamma)$
to see that $\kappa \le 1/2$ except with exponentially small probability.
\end{proof}

\section{Proofs in streaming model}

\subsection{Space lower bound in explicit vertex arrival setting} \label{appensec:streaming}

In this section, we sketch a proof of the following observation.

\streamreduction*

For completeness, we describe the reduction of \cite{Zelke11} from the Index function problem\footnote{\cite{Zelke11} refers to the problem as the bit vector probing problem} in the 2-party communication complexity setting to the problem of designing a one-pass streaming algorithm with $o(n^2)$ memory in the edge arrival setting that computes the minimum cut of a given graph. We then note that the exact same reduction can be implemented even if one considers the explicit vertex arrival setting, thus proving \cref{obs:lb_vertex_arrival}.

Formally, in the Index function problem, Alice is given a binary string $x$ of length $\ell$, and Bob receives an index $i\in [\ell]$. Bob's goal is to learn the value of $x_i$. A well known result \cite{KremerNR99} states that any one-way\footnote{A protocol consisting of a single message sent from Alice to Bob.} communication protocol that solves the Index function problem requires $\Omega(\ell)$ bits of communication.

\paragraph{The reduction:} Let $A$ be a one-pass streaming algorithm in the edge arrival setting using $o(n^2)$ bits of memory that computes the minimum cut in any given graph $G=(V,E)$ on $n$ vertices. Now, let $x,i$ be an instance of the Index function problem such that $x$ is of length $\frac{n^2-n}{2}$, and $i\in [\frac{n^2-n}{2}]$. Consider the following communication protocol in which Alice interprets $x$ as the description of a simple graph $G=(V,E)$ on $n$ vertices, Alice feeds the edges of $G$ to $A$ in an arbitrary order, and then sends $A$'s working memory to Bob, along with the degrees of all vertices in $G$. Bob interprets $i$ as a pair of vertices $a,b$ for which he wants to know whether $(a,b)\in E$. Bob continues the execution of $A$ by extending $G$ into a graph $G^*=(V^*,E^*), V\subseteq V^*, E\subseteq E^*$ as follows. Bob adds two cliques $S,T$, each on $3n$ vertices to the graph $G$, connects all vertices in $T$ to all the vertices in the set $V\backslash \set{a,b}$, and connects all vertices in $S$ to both $a$ and $b$. Finally, Bob adds a final vertex $c$ to the graph and connects it to $d_G(a)+d_G(b)-1$ vertices in $T\cup V\backslash \set{a,b}$.
 
Zelke \cite{Zelke11} proves that computing the minimum cut in the resulting graph allows Bob to infer whether $(a,b)\in E$, thus proving the reduction. 

\begin{observation}\label{obs_appen:vertex_explicit_reduction_Streaming}
The above reduction can be implemented in the explicit vertex arrival setting.
\end{observation}

\begin{proof}
We go over all insertions of edges to $A$ and show that they can implemented in the vertex arrival setting. First, Alice inserts $G=(V,E)$ into $A$ in an arbitrary order, thus if $v_1,...,v_n$ is any arbitrary order on the vertices of $G$, we can insert the edges of $G$ into $A$ by inserting the vertices in increasing order along with all incident edges to previously seen vertices. The same trick can be applied to the insertion of $S$ and $T$ into $A$ with the addition that every $s\in S$ inserted is connected to not only previously seen vertices of $S$, but also to $\set{a,b}$, \emph{which are also previously seen vertices}, and similarly for every vertex $t\in T$ and the edges connecting $t$ to $V\backslash \set{a,b}$. Lastly, $c$ can clearly be added in the vertex arrival setting as it is the last vertex, so in particular all of its edges are incident to previously seen vertices. 
\end{proof}

\cref{obs_appen:vertex_explicit_reduction_Streaming} combined with the soundness of the reduction proved in \cite{Zelke11} concludes the proof of \cref{obs:lb_vertex_arrival}.

\subsection{Proof of parallel sampling lemma} \label{app:parallel-sampling}

\sampling*

\begin{proof}
We begin with designing a procedure that outputs independent samples $Y_{1},\ldots,Y_{r}\sim B([n],p)$. The procedure
does not operate on a stream of vertices yet, and instead just samples vertices uniformly at random when necessary.

\begin{algorithm}[!htbp]
\caption{Sampling independent subsets}
\label{alg:stream_sample}
 \hspace*{\algorithmicindent} \textbf{Input:} A probability parameter $p$, number of sets to sample $r$, and size of the universe $n$. \\
 \hspace*{\algorithmicindent} \textbf{Output:} Sets $Y_1, \ldots, Y_r$ that are independent samples from $B([n], p)$ .
\begin{algorithmic}[1]
\State Independently sample $X_1, \ldots, X_r \sim B([n],p)$.
\State For $i=1, \ldots, r$ let $k_i = |X_i|$.
\State Let $f_1 = k_1$ and for $i=2, \ldots r$ let $f_i = |X_i \setminus \bigcup_{j=1}^{i-1} X_j |$.
\State $Z_0 \gets \emptyset$.
\For{$i=1, \ldots,r$} 
	\State Sample a uniformly random set $S_1$ of size $k_i - f_i$ from $Z_{i-1}$.
	\State Sample a uniformly random set $S_2$ of size $f_i$ from $[n] \setminus Z_{i-1}$. \label{lin:subset}
	\State $Y_i \gets S_1 \cup S_2$ and $Z_i \gets Z_{i-1} \cup Y_i$.
\EndFor
\State Output $Y_1, \ldots, Y_r$.
\end{algorithmic}
\end{algorithm}

\begin{proposition}
\label{prop:samplingproc}
Given input parameters $n,p$, and $r$, \cref{alg:stream_sample} samples independent and identically distributed sets 
$Y_1, \ldots, Y_r \sim B([n],p)$.
\end{proposition}

\begin{proof}
Let $(Y_1, \ldots, Y_r) \sim \Acal$ indicate a sample from \cref{alg:stream_sample}.
For any sets $R_1, \ldots, R_r \subseteq [n]$ we show that the probability the algorithm outputs 
$R_1, \ldots, R_r$ is equal to $\Pr_{X_1,\ldots,X_{r} \sim B([n],p)}[X_1 = R_1 \wedge \cdots \wedge X_r = R_r]$.  
We will decompose the latter probability as a product of conditional probabilities.  The first conditional 
event we consider is that $X_1, \ldots, X_r$ satisfy some basic size and intersection requirements to be 
equal to $R_1, \ldots, R_r$.  Specifically, let $\Scal$ be the event that $|X_i| = |R_i|$ for $i=1, \ldots, r$ 
and that $|X_i \setminus \bigcup_{j=1}^{i-1} X_j| = |R_i \setminus \bigcup_{j=1}^{i-1} R_j|$ for $i = 2, \ldots, r$.  
Then 
\begin{eqnarray*}
&&\Pr_{X_1,\ldots,X_{r} \sim B([n],p)}[X_1 = R_1 \wedge \cdots \wedge X_r = R_r] \\
&=&  
\Pr_{X_1,\ldots,X_{r} \sim B([n],p)}[\Scal] \cdot \Pr_{X_1,\ldots,X_{r} \sim B([n],p)}[X_1 = R_1 \wedge \cdots \wedge X_r = R_r \mid \Scal ] \enspace.
\end{eqnarray*}
Let $\Scal'$ be the analogous event that $Y_1, \ldots, Y_r \sim \Acal$ satisfy the same size and 
intersection requirements of $R_1, \ldots, R_r$.  Note that by the first 3 lines of the algorithm we have
$\Pr_{X_1,\ldots,X_{r} \sim B([n],p)}[\Scal] = \Pr_{(Y_1, \ldots, Y_r)\sim \Acal}[\Scal']$.  To prove the lemma it thus suffices to show 
\[
\Pr_{X_1,\ldots,X_{r} \sim B([n],p)}[X_1 = R_1 \wedge \cdots \wedge X_r = R_r \mid \Scal] = \Pr_{(Y_1, \ldots, Y_r) \sim \Acal}[Y_1 = R_1 \wedge \cdots \wedge Y_r = R_r \mid \Scal'] \enspace .
\]
This will follow from showing 
\begin{enumerate}
	\item $\Pr_{X_1,\ldots,X_{r} \sim B([n],p)}[X_1 = R_1 \mid \Scal] = \Pr_{(Y_1, \ldots, Y_r) \sim \Acal}[Y_1 = R_1 \mid \Scal']$, and 
	\item for $i=2, \ldots, r$
	\begin{align*}
	\Pr_{X_1,\ldots,X_{r} \sim B([n],p)}[X_i = R_i \mid \Scal, X_j = R_j, &1\le j < i] \\
	&= \Pr_{(Y_1, \ldots, Y_r) \sim \Acal}[Y_i = R_i \mid \Scal, Y_j = R_j, 1\le j < i] \enspace .
	\end{align*}
\end{enumerate}
Item~1 follows directly as, conditioned on $\Scal$, $X_1$ is a uniformly random subset of $[n]$ of size $|R_1|$, as is $Y_1$ conditioned on $\Scal'$.  Consider 
now the second item for an arbitrary $i$.  
Let $W = \bigcup_{j=1}^{i-1} R_j$.  Let $k_i = |R_i|$ and $f_i = |R_i \setminus W|$.  Conditioned on $\Scal$, we know that $|X_i|=k_i$ and 
$|X_i \setminus W| = f_i$.   Thus subject to the conditional, $X_i$ is the union of a uniformly random chosen set of size $f_i$ from $[n] \setminus W$ and 
a uniformly random chosen set from $W$ of size $k_i - f_i$.  Similarly, conditioned on $\Scal$ and that $Y_j = R_j$ for $j=1,\ldots, i-1$ so that 
$\bigcup_{j=1}^{i-1} Y_j = W$, in the algorithm $Y_i$ is defined to be the union of a uniformly chosen set from $[n] \setminus W$ of size $f_i$ and 
a uniformly random chosen set from $W$ of size $k_i - f_i$.  Thus the two sides of the equation in item~2 are equal.
\end{proof}

Next, we show how to apply  \cref{alg:stream_sample} on a stream of vertices. To this end, line~\ref{lin:subset} is
implemented by repeatedly reading subsequent vertices from the given stream $S$ (instead of repeatedly sampling
a vertex uniformly at random from the remaining vertices).
As each read vertex is included
in $Y_{i}$, after reading the first $\bigcup_{i=1}^{r}Y_{i}$ vertices we can output the generated subsets $Y_{1},\ldots,Y_{r}$
as required.
It remains to argue that the probability of generating $Y_{1}=R,\ldots,Y_{r}=R$, over a random stream $S$
and the random choices made by the algorithm, is equal to $\prod_{i=1}^{r}\Pr_{X_{i}\sim B([n],p)}[X_{i} = R_{i}]$.
By \cref{prop:samplingproc}, this is the case for the subsets $Y_{1},\ldots,Y_{r}$ generated by \cref{alg:stream_sample}.
Next, we argue that the probability of the original \cref{alg:stream_sample} generating $Y_{1}=R,\ldots,Y_{r}=R$
(over the random choices of the procedure)
is the same as the probability of the modified \cref{alg:stream_sample} generating $Y_{1}=R,\ldots,Y_{r}=R$
(over a random stream and the random choices made by the algorithm).
This is the case because while the latter samples the next vertex uniformly at random from the remaining vertices,
while the latter read the next vertex from the stream, which for a random stream is chosen uniformly
at random from the remaining vertices.
\end{proof}

\subsection{Proof of algorithm in complete vertex arrival setting} \label{app:vertex-arrival}

\complete*

\begin{proof}
Similarly to the case of random vertex arrivals,
we run in parallel $\log (n)$ independent instances of an algorithm, each of which uses a different estimate $d = 2^\ell$ for the minimum degree $\delta(G)$, with $\ell = 0,1,2,\dots,\ceil{\log(n)}-1$.
Each algorithm aborts if it uses more than $\tOh(n)$ memory, and we will show that if $\ell$ is such that $d \leq \delta(G) < 2d$
then with high probability the corresponding algorithm will not abort and have correct outcome $\lambda(G)$.
Since we know $\delta(G)$ exactly by the end of the stream (we can keep track of all degrees with $\tOh(n)$ memory), we can filter out the correct outcome at the end of the algorithm.

In the remainder we describe the algorithm for an estimate $d$ on the minimum degree.
As in the case of random vertex arrival, the algorithm attempts uniform star contraction on $G$ with $p=\frac{1200 \ln n}{d}$,
and in parallel constructs a sparse $2d$-edge connectivity certificate on the contracted graph, so that at the end of the stream
we can compute the edge connectivity of this certificate.
We will run $r \in \Theta(\log n)$ parallel repetitions of this. 
In the $\ith$ repetition, we first sample the set $R_{i}$ by choosing each vertex with probability $p=\frac{1200 \ln n}{d}$.
Then, we read the vertices from the input stream. As each vertex $v$ arrives with all of its incident edges,
when $v\notin R_{i}$ we are able to choose a uniformly random edge incident on $R_{i}$ and contract it in $G'_{i}$.
In parallel, we build a $2d$-connectivity certificate $F^i_1 \cup \cdots \cup F^i_{2d}$ of $G_i'$ of $G'_{i}$.
We summarize the $\ith$ parallel repetition in full detail. Recall that each repetition is aborted
as soon as its memory usage exceed $\tOh(n)$.

\begin{enumerate}
\item
Construct $R_{i}$ by choosing each vertex with probability  $p=\frac{1200 \ln n}{d}$.
Initialize $F^i_1,\dots,F^i_{2d}$ as empty forests and set $R_i = \emptyset$.
Initialize a mapping $r_i:V\to V$ to be the identity (through the stream this will keep track of the contracted vertices).
\item
For the $\jth$ vertex arrival $v$ with edges $e_1,...,e_\ell$ between $v$ to all other vertices, the following is done:
 \begin{enumerate}
    \item \textbf{Uniform star contraction:}
    If $v\in R_{i}$, do nothing.
    If $v\notin R_{i}$ we consider the set $N_{R_i}(v)$, i.e., the set of neighbors of $v$ in $R_i$. If this set is empty, we abort the $\ith$ instance.
    Otherwise, we pick a uniformly random center $w$ from the center neighborhood $N_{R_i}(v)$ (if it exists) and set $r_i(v) = w$.
    This amounts to contracting the edge $\{v,w\}$. For each $e_t$ among $e_1,\dots,e_\ell$, except for the contracted edge which is discarded, change the endpoints of $e_t=\{v,u\}$ to be $\{r_i(v),r_i(u)\}$, discard any self loops. At the end of the stream, the vertices with the same $r_i(\cdot)$ values constitute a vertex in $G_i'$.
    \item \textbf{Maintaining of $2d$-edge connectivity certificate:}
    For each (relabelled) incident edge $e_{t}$ among $e_1,\dots,e_{\ell}$, add $e_t$ to $F^i_k$ where $k$ is the minimal index for which $F^i_k\cup \set{e_t}$ contains no cycles. If there is no such $k$, discard the edge.
 \end{enumerate}
\item
If the repetition did not abort by the end of the stream, we compute the edge connectivity of the connectivity certificate $\lambda_i = \lambda(G'_i) \geq \lambda(G)$.
\end{enumerate}
Finally, we combine the $r$ parallel repetitions by outputting $\min\{\delta(G),\lambda_1,\dots,\lambda_r\}$.

\paragraph{Analysis.}
The analysis is very similar to that of \Cref{thrm:random_vertex_arrival}, we include it here for completeness.
It is enough to prove the correctness and a $\tOh(n)$ memory bound only for the algorithm that has an estimate $d$ such that
$d\leq \delta(G) < 2d$. 
To this end, we only need to argue that, with constant probability, in a single repetition we have $|G'| = \tOh(n/d)$ and
if $\lambda(G)<\delta(G)$ then $\lambda(G)=\lambda(G')$ (it is easy to see that we never have $\lambda(G') < \lambda(G)$).

A single repetition implements uniform star contraction with $p=\frac{1200 \ln n}{d}$.
First, we want to analyse $|G'|$.
 By item~1 of \cref{prop:sizeR}, $|R|\leq 2pn$ except with probability $n^{-400}$.
 Next, by item~2 of \cref{prop:sizeR} and $d\leq d(v)$,
the probability that a vertex in $[n]\setminus R$ has no neighbor in $R$ is at most
$n^{-6}$. By a union bound, $|G'| = |R| \leq 2pn = \tOh(n/d)$ except with probability
at most $n^{-400}+n^{-5}$.
Second, we want to lower bound the probability that $\lambda(G)=\lambda(G')$,
assuming that $\lambda(G')<\delta(G)$.
Let $C$ be a non-trivial minimum cut of $G$. 
By \cref{lem:good_for_contracting}, $H = (V, \cutd(V \setminus R)$ is $(2/3,8)$-good for contracting with respect to $C$ with 
probability at least $2/3$ over the choice of $R$.
By \cref{cor:sample_cor} performing a random 1-out contraction on $H$ does not contract any
edge of $C$ with probability at least $3^{-12}$ by \cref{cor:sample_cor}. 
Thus $\lambda(G')=\lambda(G)$ with probability at least~$2/3\cdot 3^{-12}$.  
\end{proof}

\section{Reduction from minimum degree to edge connectivity} \label{appensec:reduction}

Here we mention a simple reduction from computing the minimum degree of a simple graph $G = (V, E)$ to computing the edge connectivity of a graph $G' = (V', E')$ with $|V'| = 2|V|$.

\begin{lemma}\label{lem:reduction-degree-connectivity}
Given a simple graph $G = (V, E)$ on $n$ vertices for which we need to find the edge connectivity, we can construct another simple graph $G' = (V',E')$ such that:
\begin{itemize}
    \item The size of the vertex set $|V'|= 2n$, and
    \item the edge connectivity $\lambda(G') = \delta(G) + n$.
\end{itemize}
\end{lemma}

\begin{proof}
The construction of $G'$ is simple: The vertex set of $G'$ is $V' = V \cup K$ where $K$ is a set of size $n$. The edge set $E' = E_1 \cup E_2 \cup E_3$ constitutes the following three types of edges: 
\begin{enumerate}
    \item \textbf{Original edges:} $E_1 = E$ consists of all the original edges of $G$.
    \item \textbf{Cross edges:} $E_2 = \{ \{u,v\} \mid u \in V, v \in K\}$ consists of all edges between $V$ and $K$.
    \item \textbf{Clique edges:} $E_3 = \{ \{u,v\} \mid u \ne v, u,v \in K \}$ creates a clique on $K$.
\end{enumerate}
In simpler words, $G'$ consists of $G$ and a clique on vertices $K$ with all cross edges present between $V$ and $K$.

We can immediately note that, for any $v \in V$, the degree $d_{G'}(v)$ in $G'$ is $d_G(V)+n$. For all vertices $v \in K$, the degree is 
$2n-1$. Hence a vertex $v$ that has minimum degree in $G$ also has minimum degree in $G'$, which is at most $2n -1$.

We now argue that the value of any non-trivial cut in $G'$ is at least $2n-1$. Let $X \subset V'$ denote a side of the cut and let $a = |X \cap V|$ and $b = |X \cap K|$. For a non-trivial cut we have $a+b > 1$ and we can assume w.l.o.g.~that $b \leq n/2$ (otherwise consider $X^c$).
Now note that
\begin{align*}
|\cut_{G'}(X)|
&= |\cut_G(X \cap V)| + |\cut_K(X \cap K)| + |E(X \cap V,K \backslash X)| + |E(V\backslash X, X \cap K)| \\
&\geq |\cut_K(X \cap K)| + |E(X \cap V,K \backslash X)| + |E(V\backslash X, X \cap K)| \\
&= b(n-b) + a(n-b) + (n-a)b.
\end{align*}
Now if $b = 0$ then $a>1$ and the right hand side is $an > 2n-1$. If $b=1$ then $a\geq 1$ and the right hand side is $2n+a(n-2)-1 \geq 2n-1$ assuming $n \geq 2$. Finally, if $n/2 \geq b>1$ then we can lower bound the right hand side by $a(n-b) + (n-a)b = nb + a(n-2b) \geq nb \geq 2n$.
\end{proof}

\end{document}